\documentclass[11pt]{article}
\usepackage[authoryear]{natbib}

\usepackage{amsmath}
\usepackage{amssymb}
\usepackage{bbm}
\usepackage{graphicx}
\usepackage[usenames,dvipsnames]{color}
\definecolor{burgundy}{rgb}{0.89, 0.11, 0.11}
\definecolor{airforceblue}{rgb}{0.36, 0.54, 0.66}
\definecolor{myblue}{rgb}{0.0,0.0, 0.5}
\usepackage[colorlinks=true,urlcolor=RoyalBlue,linkcolor=BrickRed,citecolor=RoyalBlue]{hyperref}

\usepackage[capitalise]{cleveref} 
\usepackage{enumitem}
\usepackage{amsthm}
\newtheorem{theorem}{Theorem}
\theoremstyle{remark}
\newtheorem{remark}{Remark}
\theoremstyle{proposition}

\newcommand{\change}[1]{{\color{black}#1}}

\newcommand{\changeB}[1]{{\color{black}#1}}

\renewcommand{\baselinestretch}{1}
\usepackage{caption}
\renewcommand{\thesubsubsection}{\arabic{section}.\arabic{subsubsection}}

\newcommand{\captionfonts}{\normalsize}

\makeatletter  
\long\def\@makecaption#1#2{%
  \vskip\abovecaptionskip
  \sbox\@tempboxa{{\captionfonts #1: #2}}%
  \ifdim \wd\@tempboxa >\hsize
    {\captionfonts #1: #2\par}
  \else
    \hbox to\hsize{\hfil\box\@tempboxa\hfil}%
  \fi
  \vskip\belowcaptionskip}
\makeatother   

\begin{document}
\hspace{13.9cm}

\ \vspace{5mm}\\

{\Large Pulse shape and voltage-dependent synchronization}\\[0.15cm]
{\Large \indent in spiking neuron networks}

\ \vspace{5mm}\\
\indent{\bf Bastian Pietras$^{1,*}$}\\[0.15cm]
\indent{$^{1}$Department of Information and Communication Technologies,\\
\indent$^{ }$ Universitat Pompeu Fabra, T\`anger 122-140, 08018, Barcelona, Spain}\\
\indent{$^*$bastian.pietras@upf.edu}\\
%


\thispagestyle{empty}
\markboth{}{NC instructions}
\ \vspace{-10mm}\\

\begin{center} {\bf Abstract} \end{center}
Pulse-coupled spiking neural networks are a powerful tool to gain mechanistic insights into how neurons self-organize to produce coherent collective behavior.
These networks use simple spiking neuron models, such as the $\theta$-neuron or the quadratic integrate-and-fire (QIF) neuron, that replicate the essential features of real neural dynamics. 
Interactions between neurons are modeled with infinitely narrow pulses, or spikes, rather than the more complex dynamics of real synapses. 
To make these networks biologically more plausible, it has been proposed that they must also account for the finite width of the pulses, which can have a significant impact on the network dynamics. 
However, the derivation and interpretation of these pulses is contradictory and the impact of the pulse shape on the network dynamics is largely unexplored.
Here, I take a comprehensive approach to pulse-coupling in networks of QIF and $\theta$-neurons. 
I argue that narrow pulses activate voltage-dependent synaptic conductances and show how to implement them in QIF neurons such that their effect can last through the phase after the spike.
Using an exact low-dimensional description for networks of globally coupled spiking neurons, I prove for instantaneous interactions that collective oscillations emerge due to an effective coupling through the mean voltage. 
I analyze the impact of the pulse shape by means of a family of smooth pulse functions with arbitrary finite width and symmetric or asymmetric shapes. 
For symmetric pulses, the resulting voltage-coupling is \changeB{not very} effective in synchronizing neurons, but pulses that are slightly skewed to the phase after the spike readily generate collective oscillations. 
The results unveil a voltage-dependent spike synchronization mechanism at the heart of emergent collective behavior, which is facilitated by pulses of finite width and complementary to traditional synaptic transmission in spiking neuron networks.

\newpage
\section{Introduction}
Self-organization in large neural networks crucially relies on rapid and precise, in short, highly effective neuronal communication.
Brisk synaptic interactions allow for emergent collective behavior and can orchestrate neural synchronization, which is believed \change{to be}
fundamental to cognitive functions and consciousness.
A key player in the synaptic transmission process is the spike---as has been nicknamed the action potential of a neuron. 
As a central information unit of the brain it has risen to fame and fortune~\citep{wilson1999spikes,gerstner2002spiking,humphries2021spike}. 
Spikes are believed \change{to be} critical for information processing and coding, and more so as the basis of communication between neurons.
Once the membrane potential of a neuron exceeds some threshold, the soma quickly depolarizes and the neuron ``spikes''.
Straight off, a fast electrical impulse travels along the neuron's axon to the presynaptic knobs, where it triggers various biochemical processes to release neurotransmitters and eventually induce a postsynaptic current in the connected cell~\citep{Destexhe1994,sabatini1999timing,lavi2015shaping,wang2021theory}.

\changeB{In developing a mechanistic understanding of the collective behavior of large neural networks, incorporating a high degree of biological detail is challenging.
For computational and mathematical convenience, 
pulse-coupled spiking neuron networks have been proposed.
This approach proves instrumental in comprehending the information processing capabilities of neurons and enables efficient simulation of neural networks.
In these pulse-coupled spiking neuron networks, synaptic interactions are typically modeled using ``$\delta$-spikes'', defined here as infinitely narrow Dirac $\delta$-pulses emitted by presynaptic neurons at their spike times.
The term ``pulse'' specifically refers to the chemical synaptic transmission process, whereas ``spike'' signifies the firing of an action potential.
Biologically intricate synaptic transmission is reduced to the causal effect of a Dirac $\delta$-pulse---which in fact resembles a spike---initiating a response in the connected postsynaptic neuron.
A similar reductionist approach is integral to spiking neuron models, focusing on subthreshold membrane properties while excluding mechanisms responsible for generating action potentials (i.e., voltage-dependent sodium and potassium channels). 
The firing of an action potential, or a neuron's ``spike'', occurs instantaneously upon reaching a specific membrane potential threshold, yet its actual dynamics is left unspecified. 
The neuron's spike induces a presynaptic pulse that will be perceived by any connected postsynaptic neuron.
While often simplified to a $\delta$-spike, this spike-pulse interaction at the presynaptic site can assume various forms and shape the collective dynamics---this is a central theme of this work.}

\changeB{On an individual level, the assumption of instantaneous spikes is most successfully caricatured in integrate-and-fire neuron models~\citep{Burkitt_2006}. 
But also in biophysically realistic Hodgkin-Huxley-like conductance-based neuron models, the spike generation can be very rapid compared to relatively slow subthreshold integration.
The separation of time scales becomes extreme if neurons are Class 1 excitable and near the onset of firing. 
\cite{Ermentrout_Kopell_1986} proved that any such Class 1 excitable neuron can be transformed into a canonical, one-dimensional phase model---the $\theta$-neuron; see \cref{fig:cartoon}.
It is one of the simplest spiking neuron models and closely related to the quadratic integrate-and-fire (QIF) neuron~\citep{Ermentrout_1996}.
}
\begin{figure*}
    \centering
    \includegraphics[width=\textwidth]{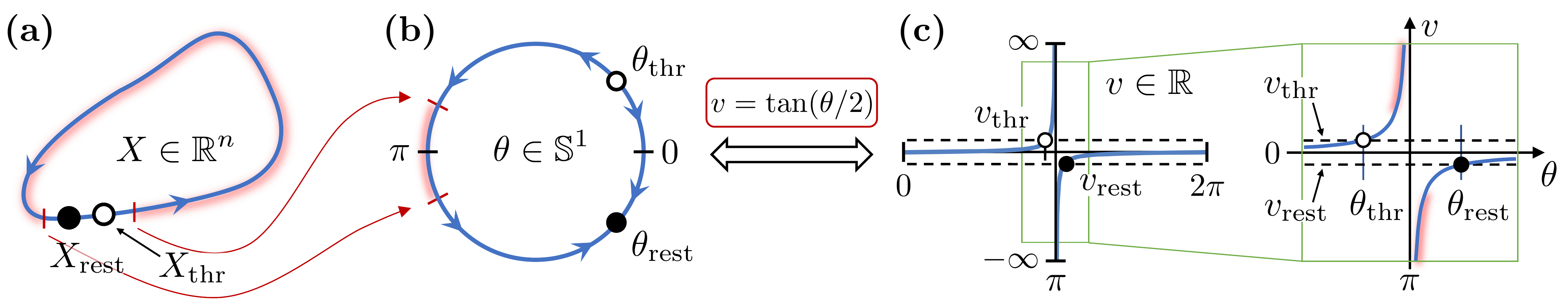}
    \caption{Reduction from (a) a high-dimensional, conductance-based neuron model close to a SNIC bifurcation to (b) the one-dimensional canonical $\theta$-neuron model with $\eta<0$. The trajectory describing the action potential along the limit cycle is shaded in red and compressed around $\pi$ in the $\theta$-neuron.
    Via the transform $v=\tan(\theta/2)$, the $\theta$-neuron is equivalent to (c) the QIF neuron when reset and peak values are taken at infinity.}
    \label{fig:cartoon}
\end{figure*}

In network models of these spiking neurons, the default description of synaptic interactions is with $\delta$-spikes, but a growing number of studies have proposed synaptic transmission via pulses of finite width. 
This raises the question how pulses of finite width should be interpreted.
They may represent the release of neurotransmitters at the presynaptic site, the conversion of neurotransmitter signals into postsynaptic currents at the postsynaptic site, or a combination of both. Despite widespread use, the nature of these pulses remains largely unclear. 
Meanwhile, the shape of the pulse is critical for neuronal synchronization and can either facilitate or impede collective oscillations. 
Currently, however, there is no definitive understanding how the pulse shape affects collective dynamics.

In this paper I revisit the use of pulse-coupling in networks of QIF and $\theta$-neurons to provide, first, a clear biological interpretation for pulses of finite width and, second, a comprehensive view on the impact of the pulse shape on the collective dynamics of globally coupled spiking neurons.
A better understanding of these fundamental aspects of synaptic transmission will greatly contribute to our understanding of neural network behavior.

\subsection{The ``pulse-problems'' of $\boldsymbol{\theta}$- and QIF neurons}
\label{subsec:1A}
QIF and $\theta$-neurons are two paradigmatic spiking neuron models. 
Each one is a canonical model for biologically detailed conductance-based neurons and can be derived from realistic, high-dimensional neuronal dynamics~\citep{Hoppensteadt_Izhikevich_1997}.
On the individual level, QIF and $\theta$-neurons are closely related; they become equivalent if threshold and reset values of the QIF neuron are taken at infinity (\cref{fig:cartoon}).
In network models, however, pulsatile synaptic interactions between QIF neurons are often modeled differently from those between $\theta$-neurons.
This dichotomy stems in part from the difficulty \changeB{of deriving} canonical network models from networks of synaptically coupled conductance-based neurons.
Such a network reduction is more challenging, if at all possible, than a single neuron reduction~\citep{Hoppensteadt_Izhikevich_1997,Pietras_Daffertshofer_2019}. 
Therefore, pulse-coupled spiking neural networks often lack mathematical rigor, especially when refraining from $\delta$-spike interactions (but see \cref{sec:2}).

Given the opposing approaches to pulse-coupling in networks of QIF or of $\theta$-neurons, it is inevitable to question the biological plausibility of pulses of finite width. 
Shall those pulses replicate the biochemical processes at the presynaptic site (converting the presynaptic action potential into the release of neurotransmitters), at the postsynaptic site (converting the neurotransmitter signal into a postsynaptic current), or at both sites?  
Moreover, and to retain the computational advantages of spiking networks, pulses are usually defined as functions of the state variable of the presynaptic neuron.
This can be limiting because chemical synaptic transmission is a dynamic process that, \change{in fact}, involves both the presynaptic and the postsynaptic neuron.
It would thus be consistent to restrict the interpretation of those pulses to the presynaptic site of a synapse, where the pulse reflects the conversion from an action potential into the release of neurotransmitter.
The state variables of the QIF and of the $\theta$-neuron, however, do not describe realistic voltage traces in the course of an action potential---in contrast to conductance-based neuron models (from which they may have been reduced). It is therefore not clear if the pulse function can actually relate the state variables to voltage-dependent mechanisms that eventually trigger neurotransmitter release.
In \cref{sec:2,sec:3}, I will revisit pulsatile interactions between $\theta$-neurons and between QIF neurons, and propose biologically plausible interpretations for pulses of finite width with symmetric and asymmetric shapes.



Next to a clear and biologically plausible interpretation, the other ``pulse-problem'' is the largely unexplored effect of the pulse shape on the collective dynamics of pulse-coupled QIF or $\theta$-neurons.
Pulses of finite width have frequently been used in the literature to describe synaptic interactions especially between $\theta$-neurons, but also between QIF neurons.
Most of the time, the pulses are assumed \change{to be} symmetric about the neuron's spike time
without varying the shape much. The mechanisms by which the pulse shape can affect the network dynamics, promote synchrony, and induce collective oscillations even for instantaneous coupling, remain largely unknown.

A few insights, though, can be borrowed from the vast results on coupled phase oscillators---the phase-representation of the $\theta$-model invites one to invoke phase oscillator theory; nonetheless, direct analogies should be regarded with care (see \cref{subsec:previous}).
Traditionally, the Winfree model~\citep{winfree1967,winfree1980geometry} has paved the way to study synchronization of periodically firing neurons.
It pinpoints a pulse-response coupling, where the pulse can be shown to result from the interplay between presynaptic action potential and a synaptic activation function (``voltage-dependent conductance'')~\citep{ermentrout_kopell_1990}. 
Furthermore, the Winfree model has allowed for exploring the effect of the response and pulse functions on synchronization properties of the network~\citep{ariaratnam2001phase,pazo_montbrio_2014,gallego_et_al_2017}.
Similarly, coupled active rotators~\citep{Shinomoto_Kuramoto_1986,Kuramoto1987} served as a basis to study how the width of (symmetric) pulses affects collective dynamics~\citep{okeeffe_strogatz_2016}.
In line with the results on the Winfree model, broad pulses were reported to entail collective dynamics that can be different from those generated by narrow pulses~\citep{pazo_montbrio_2014,gallego_et_al_2017,okeeffe_strogatz_2016};
it remains unclear, however, how these results carry over to networks of $\theta$-neurons. 
A crucial tool for distilling the effect of the pulse shape on the collective dynamics of Winfree oscillators and active rotators has been an exact dimensionality reduction first proposed by~\cite{ott_antonsen_2008}.
Although immensely powerful, the ``Ott-Antonsen ansatz'' requires the pulse function to be analytically tractable, which may come at the cost of biological realism.
The introduction of the Ott-Antonsen ansatz in networks of $\theta$-neurons~\citep{luke_barreto_so_2013} has inspired a plethora of $\theta$-neuron network studies using symmetric and broad pulses ever since~\citep{So-Luke-Barreto-14,luke2014macroscopic,Laing_2014,laing_2015,laing2016bumps,laing2016travelling,roulet2016average,laing2017,chandra2017modeling,laing2018dynamics,laing2018chaos,aguiar2019feedforward,lin_barreto_so_2020,laing2020effects,laing2020moving,means2020permutation,blasche2020degree, bick_goodfellow_laing_martens_2020,juettner_martens_2021,omel2022collective,birdac2022dynamics}.
Despite the prevailing uncertainty of their biological interpretation, these pulses nowadays seem well established in the community. \changeB{Still, a comprehensive picture how their width influences collective dynamics is missing.
On top of it,} they are not versatile enough, either, to study the effect of pulse asymmetry.
In a nutshell, a systematic investigation how the pulse shape---be it symmetric or asymmetric---affects the collective dynamics of $\theta$- or QIF neurons has remained elusive.

\subsection{Synopsis \& outline}
\change{I strive for resolving the ``pulse-problems'' of $\theta$- and QIF neurons:
first, by providing a biological interpretation for pulses of finite width; second, by analyzing the impact of the pulse shape on the collective dynamics. 

As to the biophysical interpretation of pulses of finite width, I present two alternative views in \cref{sec:2all} that are equally valid but depend on the particular modeling assumptions.
The connection between $\theta$- and QIF neurons confounds a clear separation of these alternative interpretations, as one can easily transform one model into the other. Yet, there is a subtle difference between the two spiking neuron models (\cref{sec:2intro}).
In the context of weakly coupled Class 1 excitable neurons close to the onset of firing, pulses in the canonical network model of $\theta$-neurons can describe instantaneous synaptic transmission at both the pre- and postsynaptic sites; these pulses can have arbitrary shapes as long as they are sufficiently narrow (\cref{sec:2}).
When one regards QIF neurons as ``the simplest model of a spiking neuron'' \citep{Izhikevich_2007} irrespective of the foregoing network setting, I show in \cref{sec:3} how to introduce voltage-dependent pulses $p(v)$ that replicate the transmission process of conductance-based neurons exclusively at the presynaptic site.
Crucially, I propose a modification of the synaptic activation function $p$ to account for the artificial shape of the QIF's action potential.

The correspondence between the QIF and the $\theta$-model allows one to use the two pulse-interpretations in $\theta$- and QIF networks interchangeably, nonetheless I advise caution not to mix up the respective underlying assumptions. 
Towards analytic tractability, I capitalize on the QIF-$\theta$-correspondence in \cref{subsec:accessible} and approximate the voltage-dependent pulses $p(v)$ by pulses $p_{r,\varphi,\psi}(\theta)$ formulated in terms of the $\theta$-phase via $\theta=2\arctan(v)$. 
The pulse parameters $r,\varphi,\psi$ allow for interpolating between discontinuous $\delta$-spikes and continuous pulses of finite width with symmetric or asymmetric shapes, enabling a systematic study how the pulse shape affects the collective dynamics of the network.

Conveniently,} the family of smooth pulse functions $p_{r,\varphi,\psi}(\theta)$ is admissible to an exact reduction of globally coupled spiking neurons in the thermodynamic limit~\citep{pcp_arxiv_2022}, \change{Here, I build on recent advances in coupled oscillator theory~\citep{ott_antonsen_2008,cestnik_pikovsky_2022,cestnik_pikovsky_2022chaos},
which} allows for a comprehensive analysis of the collective dynamics thanks to an exact mean-field description in terms of the firing rate $R$ and the mean voltage $V$ (\cref{sec:4}).
Taking the population average yields an expression of the mean pulse activity $P_{r,\varphi,\psi}=\langle p_{r,\varphi,\psi}\rangle$ that is fully determined by $R$ and $V$.
For instantaneous synaptic transmission, the exact mean-field dynamics converges towards an invariant two-dimensional manifold~\citep{pcp_arxiv_2022}, on which the ordinary differential equations for the firing rate and voltage (``RV dynamics'') are closed in $R$ and $V$~\citep{montbrio_pazo_roxin_2015}. 
\change{\cref{sec:5} is devoted to the mathematical analysis of the two-dimensional RV dynamics. I analyze how the different pulse parameters---width $r$, asymmetry $\varphi$, and shift $\psi$---affect the region of collective oscillations and 
prove} that collective oscillations emerge due to an effective coupling through the mean voltage~$V$ (\cref{subsec:5A}). 
This voltage-dependence readily arises for global pulse-coupling as $P_{r,\varphi,\psi}=P_{r,\varphi,\psi}(R,V)$ explicitly depends on $V$, except for the limit of $\delta$-spikes, $(r,\varphi,\psi)\to (1,0,\pi)$.
In other words, pulse-coupling generally facilitates collective oscillations through an effective voltage-coupling, 
but if neurons interact via $\delta$-spikes, the recurrent input no longer depends on the mean voltage and collective oscillations become impossible.
Moreover, the pulse shape determines the effectiveness of the pulse-mediated voltage-coupling and, thus, plays a crucial role for the emergence of collective oscillations. 
For symmetric pulses, collective oscillations are confined to a small parameter region and require unrealistically strong inhibition the narrower the pulse (\cref{subsec:3B}).
Additionally, broad symmetric pulses have a nongeneric effect on the collective dynamics that is not present in narrow pulses, see also \citep{pazo_montbrio_2014,okeeffe_strogatz_2016,gallego_et_al_2017}. 
In contrast to symmetric pulses, narrow pulses that are slightly skewed to the phase
after the spike readily generate collective oscillations in networks of inhibitory neurons (\cref{subsec:3C}), whereas pulses that are slightly skewed to the phase before the actual spike generate collective oscillations among \change{excitatory} neurons (\cref{subsec:3D}).
Together, the results shed new light on the voltage-dependent spike synchronization mechanism that is typically not captured in traditional mean-field, or firing rate, models~\citep{Wilson-Cowan_1972}, but which crucially underlies collective oscillations in neural networks.

In the Discussion in \cref{sec:6}, I first review previous approaches to pulse-coupling in networks of spiking neurons, which may have led to misconceptions about the interpretation and the effect of pulses of finite width in networks of $\theta$- and QIF neurons (\cref{subsec:previous}).
I then revisit the three pulse parameters in more detail and draw connections to delayed synaptic interactions and electrical coupling via gap junctions (\cref{subsec:D1}). 
Finally, I return to the question whether instantaneous pulses of finite width can replace more complex synaptic transmission in spiking neuron networks including synaptic kinetics and conductance-based synapses---I argue in the negative (\cref{subsec:D2}).
Conclusions, final remarks and an outlook will be given in \cref{sec:7}.
Mathematical details can be found in \cref{appsec:proof,appsec:WB,appsec:A,appsec:B,appsec:C,sec:app_pulses,sec:colvar_derivation,appsec:E}.

\section{Pulses in spiking neuron networks}
\label{sec:2all}

\change{
\subsection{Two canonical spiking neuron models}
\label{sec:2intro}
I focus on (1) the ``theta''-neuron as the canonical model for Class 1 excitable\footnote{Class 1 (or I) excitable neurons are also known as ``Type I membranes'' whose neuronal dynamics exhibit ``Type I excitability''.} neurons close to saddle-node bifurcation on an invariant circle (SNIC bifurcation) 
and on (2) the Quadratic Integrate-and-Fire (QIF) neuron as the canonical model for a biologically detailed conductance-based neuron close to a saddle-node bifurcation~\citep{Hoppensteadt_Izhikevich_1997}.
Canonical means that, in the first case, any Class 1 excitable neuron close to a SNIC bifurcation can be reduced to a one-dimensional neuron with a phase variable $\theta$, whose dynamics is governed by 
\begin{equation}
    \dot \theta = \tfrac{d}{dt} \theta(t) = (1-\cos\theta) + (1+\cos \theta) \eta
    \label{eq:EK_canonical}
\end{equation} 
with excitability parameter $\eta$.
Despite its formulation in terms of the phase variable $\theta$,
\cref{eq:EK_canonical} does not result from phase reduction (see also \cref{subsec:previous}). 
Instead, \cref{eq:EK_canonical} is the Ermentrout-Kopell canonical model~\citep{Ermentrout_Kopell_1986,Ermentrout_1996,Gutkin_Ermentrout_1998} for a SNIC bifurcation, which occurs when $\eta=0$. 
For $\eta<0$, the neuron is in an excitable regime as the dynamics \eqref{eq:EK_canonical} exhibits a pair of stable and unstable equilibria.
The stable equilibrium represents the neuron's resting state.
The unstable equilibrium represents a threshold: when an external input drives the neuron across this threshold, $\theta$ will move around the circle in the course of the neuron's action potential and approaches the resting state from below. 
The neuron is said to fire a spike when $\theta$ crosses $\pi$.
For $\eta>0$, the two equilibria have disappeared in a saddle-node bifurcation, leaving a limit cycle, and the neuron is in a tonically (periodically) spiking regime.
The reduction from a general conductance-based Class 1 excitable neuron close to a SNIC bifurcation to the $\theta$-neuron~\eqref{eq:EK_canonical} is sketched in \cref{fig:cartoon} for $\eta<0$: 
close to the saddle-node bifurcation, a small neighborhood of the resting potential is blown up and the trajectory describing the neuron's action potential along the limit cycle is compressed to an open set around $\pi$ (shaded in red).

The QIF neuron is characterized by a voltage variable $v$ that follows the subthreshold dynamics
\begin{equation}
\dot v = \tfrac{d}{dt} v(t) = v^2 + \eta
\label{eq:qif0}
\end{equation}
with a fire-and-reset rule: when $v$ exceeds a peak value $v_p$, the voltage is reset to a reset potential $v_r$ and the neuron is said to elicit a spike.
The parameter $\eta$ plays a similar role as in the $\theta$-neuron \eqref{eq:EK_canonical}: For $\eta<0$, there is a pair of stable and unstable equilibria and the QIF neuron will settle into its resting potential at $v = -\sqrt{-\eta}$.
For $\eta>0$ the two equilibria disappear in a saddle-node bifurcation and the voltage of the QIF neuron will diverge; due to the fire-and-reset mechanism, the neuron then enters into a periodically firing regime.

If peak and reset potentials are taken at infinity, $v_p=-v_r=\infty$, the QIF neuron becomes equivalent to the $\theta$-neuron via the variable transform $v=\tan(\theta/2)$, see \citep{Ermentrout_1996} and \cref{fig:cartoon}.
Independent from the equivalence with the $\theta$-model, however, the QIF model has its own right to exist.
For example, the QIF dynamics can be obtained from conductance-based neuron models with a parabolic-like voltage nullcline, which is a general feature of Hodgkin-Huxley-like neurons with positive feedback (so-called amplifying) ionic currents, through a quadratization procedure~\citep{Rotstein_2015,Turnquist_Rotstein_2018}.
Moreover, the subthreshold QIF dynamics \eqref{eq:qif0} is the topological normal form of a saddle-node (fold) bifurcation~\citep{Kuznetsov_1998}; the QIF model thus describes any neuronal model close to a saddle-node bifurcation~\citep{Hansel_Mato_2001,Hansel_Mato_2003}. 
Equipped with finite threshold and reset values, the QIF model is also the simplest spiking neuron model with a spike generation mechanism (i.e., a regenerative upstroke of the membrane potential), with a soft (dynamic) threshold and a spike latency~\citep{Izhikevich_2007}.
The approximation with infinite threshold and reset values (as considered throughout this work) then allows for valuable insights into the collective dynamics of QIF neurons as typical Class 1 excitable systems, though not necessarily claiming an explicit biophysical interpretation of the underlying neuronal dynamics.

Despite the close interconnection between $\theta$- and QIF neurons, the two canonical spiking neuron models have subtle differences even when the QIF neuron is equipped with finite threshold and reset values. 
These differences become clear when considering networks of them and when introducing pulse-coupling that is different from the commonly employed $\delta$-spikes.
}

\subsection{Pulses in the canonical model of weakly connected Class 1 excitable neurons}
\label{sec:2}

\changeB{
  Pulse-coupled neural networks are celebrated for their utility and practicability, yet often met with skepticism and dismissed as mere toy models. Nevertheless, a handful of broadly applicable conditions have been distilled to reduce a large class of realistic conductance-based neural network models to a canonical model of pulse-coupled $\theta$-neurons, exhibiting qualitatively identical dynamical properties.
  \cite{Izhikevich_1999} proposed a formal derivation from weakly connected networks of biologically plausible and biophysically detailed Class 1 excitable neurons. 
  This derivation yielded $\theta$-neurons coupled through $\delta$-spikes.
  However, synaptic interactions in the canonical model of pulse-coupled $\theta$-neurons occur not through (discontinuos) $\delta$-spikes, but via (smooth) localized pulses of finite width---a distinction I will rigorously establish below.
  These narrow pulses capture the swift impact of a presynaptic action potential, eliciting a postsynaptic response in the connected neuron.
  Such dynamics result directly from \cref{thm:1} (referenced below), contingent upon satisfying the following general conditions \citep{Izhikevich_1999}:
}
\begin{enumerate}[label=(\roman*)]
    \item \label{item1} Neurons are Class 1 excitable, i.e.\ action potentials can be generated with arbitrarily low frequency, depending on the strength of the applied current.
    \item \label{item2} Neurons are weakly connected, i.e.\ the amplitudes of postsynaptic potentials (PSP) are much smaller than the amplitude of an action potential or than the mean excitatory PSP size necessary to discharge a quiescent neuron.
    \item \label{item3} Synaptic transmission has an intermediate rate, which is slower than the duration of an action potential, but faster than the interspike period.
    \item \label{item4} Synaptic connections between neurons are of conventional type, i.e.\ axo-dendritic or axo-somatic.
    \item \label{item5} Synaptic transmission is negligible when presynaptic neurons are at rest, i.e.\ spontaneous release of neurotransmitters does not affect significantly spiking of postsynaptic neurons. 
\end{enumerate}
\changeB{These assumptions are generically met by a large class of neural networks, as extensively discussed by \cite{Izhikevich_1999} regarding their biological plausibility.
In brief, Assumption~\ref{item1} reflects the general belief that the majority of mammalian neurons are actually Class 1 excitable \citep{Izhikevich_1999,Pfeiffer_et_al_2023};
\ref{item2} reflects the \emph{in vitro}-observation that amplitudes of postsynaptic potentials (PSP’s) are much smaller than the mean excitatory PSP necessary to discharge a quiescent cell, and tiny compared to the amplitude of an action potential \citep{Hoppensteadt_Izhikevich_1997};
finally, \ref{item3}--\ref{item5} reflect the idea that conventional synaptic transmission occurs from one pre- to one postsynaptic neuron triggered by, and lasting slightly but not much longer than, the presynaptic action potential.
}


\changeB{Assumption~\ref{item1} ensures that an individual conductance-based neuron is close to a SNIC bifurcation, enabling its reduction to the canonical $\theta$-neuron~\citep{Ermentrout_Kopell_1986,Ermentrout_1996}.
When embedded in a network and interacting with each other, Assumptions \ref{item2}--\ref{item5} become necessary for the reduction of these neurons to a canonical network model of $\theta$-neurons. 
\emph{Network reduction}, however, is significantly more intricate than reducing a single neuron due to the interdependence introduced by coupling terms~\citep{Pietras_Daffertshofer_2019,Hoppensteadt_Izhikevich_1997}.
Any transformation simplifying the dynamics of one neuron immediately influences the dynamics of others; therefore, the reduction of a neural network must occur for all neurons simultaneously. 
I pursue this network reduction in \cref{appsec:proof}, where I prove that the canonical network model for weakly coupled Class 1 excitable conductance-based neurons is given by pulse-coupled $\theta$-neurons with smooth pulses of short but finite width, as follows:

  \begin{theorem}
  \label{thm:1}
  Consider an arbitrary weakly connected neural network of the form 
  \begin{equation}
  \dot X_j = F_j(X_j,\lambda) + \varepsilon G_j(X_1,\dots,X_N; \lambda,\varepsilon)
  \label{eq:thm1}
  \end{equation}
  satisfying that
  each (uncoupled) equation $\dot X_i = F_j(X_j,\lambda)$ undergoes a SNIC bifurcation for some $\lambda=\lambda_0$,
  that each function $G_j$ has the pair-wise connected form 
  $$G_j(X_1,\dots,X_N; \lambda_0,0)  = \sum_{k=1}^N G_{jk} (X_j,X_k)$$
  and each $G_{jk} (X_j,X_k)=0$ for $X_k$ from some open neighborhood of the saddle-node bifurcation point.
  Then, there is $\varepsilon_0 >0$ such that for all $\varepsilon < \varepsilon_0$ and all $\lambda = \lambda_0 + \mathcal O(\varepsilon^2)$, there is a piece-wise continuous transformation that maps solutions of \eqref{eq:thm1} to those of a canonical network model of pulse-coupled $\theta$-neurons, which can be approximated by
  \begin{equation}
    \theta'_j = (1-\cos \theta_j) + (1+\cos \theta_j) \Big[ \eta_j + \sum_{k=1}^N p_{jk} (\theta_k) \Big]
    \label{eq:thm3}
  \end{equation}
  with constants $\eta_j\in \mathbb{R}$.
  The functions $p_{jk}(\theta_k)$ describe smoothed $\delta$-pulses of strength $s_{jk}=\mathcal O(\|G_{jk}\|)$, where $\|.\|$ denotes the supremum norm; 
  i.e.~the pulse strength is proportional to the amplitude of the coupling in \eqref{eq:thm1}.
  The duration of every pulse in \eqref{eq:thm3} is short as $p_{jk}(\theta_k) = 0$ if $|\theta_k - \pi| > 2\sqrt\varepsilon$ for all $j,k=1,\dots,N$, see \cref{fig:cartoon_pulse}.
  \end{theorem}
}

\begin{figure*}
    \centering
    \includegraphics[width=\textwidth]{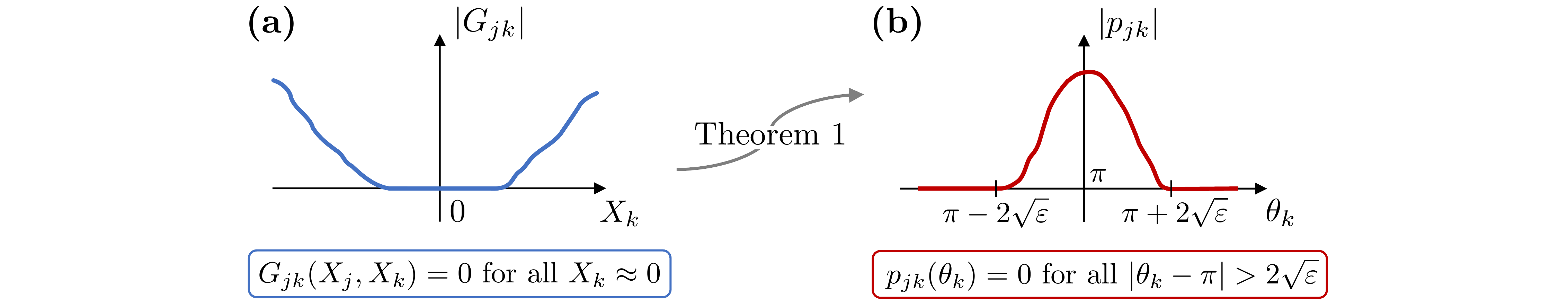}
    \caption{\change{(a) The coupling function $G_{jk}$ of the conductance-based neuron model \eqref{eq:thm1} with state variables $X_{j} \in\mathbb R^n$ satisfying Assumption \ref{item5} has a ``dead zone'' \citep{Ashwin_et_al_2021} around the resting potential $|X_k|=0$ of postsynaptic neuron $k$.
    (b) The smooth pulse function $p_{jk}$ in the pulse-coupled network \eqref{eq:thm3} of $\theta$-neurons is non-zero around the spike time $\theta_k = \pi$ of neuron $k$ and can be derived from weakly coupled Class 1 excitable neurons via \cref{thm:1}, see also \cref{appsec:proof}.}}
  \label{fig:cartoon_pulse}
\end{figure*}

\changeB{For the proof of the theorem including all mathematical details, see \cref{appsec:proof}.

In the canonical pulse-coupled network \cref{eq:thm3}, individual $\theta$-neurons are solely characterized by the parameter $\eta_j$, whereas the uncoupled dynamics of the underlying neuronal model \eqref{eq:thm1} are given by the function $F_j$ describing, e.g., a conductance-based Hodgkin-Huxley-like neuron. 
Interactions between two synaptically connected neurons $k$ and $j$ arise according to Assumption~\ref{item5} only when the presynaptic neuron, say neuron $k$, is firing a spike.
The event of such an action potential lasts $4\sqrt\varepsilon$ units of slow time---when $\theta_k$ crosses a $2\sqrt\varepsilon$-neighborhood of $\pi$.
As Assumption~\ref{item2} requires that $\varepsilon\ll 1$ is small,
it seems reasonable to approximate an emitted pulse by the Dirac $\delta$-function as $p_{jk}(\theta_k)\approx s_{jk}\delta(\theta_k-\pi)$ \cite[cf.~Proposition 8.12 in][]{Hoppensteadt_Izhikevich_1997}.
In general, however, the pulses $p_{jk}(\theta_k)$ are smooth and of finite width, which can have important consequences for the collective network dynamics depending on the shape of the pulse, see \cref{sec:5}. While an exact reduction of pulse functions $p_{jk}$ from weakly coupled, Class 1 excitable, conductance-based neurons is beyond the scope of this work due to technical difficulties detailed in \cref{appsec:proof}, exploring the impact of general pulse shapes on the collective dynamics is a valuable pursuit.
The $\theta$-neuron network exhibits rich collective behavior already with narrow pulse-coupling. 
For a more comprehensive understanding, I will study pulse-coupled networks of $\theta$-neurons with various pulse shapes in \cref{sec:5}, where I consider pulses also beyond the limits imposed by \cref{thm:1}, and discuss their biological interpretation in \cref{subsec:D2}.

\begin{remark}
  The shape of the pulses $p_{jk}(\theta_k)$ in \cref{eq:thm3} depends nonlinearly on the shape of the presynaptic action potential (governed by $F_j$) in interplay with the coupling function $G_{jk}$.
  For example, a pulse $p_{jk}$ may describe a conductance-based synapse in the underlying high-dimensional model \cref{eq:thm1} of the form 
  \begin{equation}
  G_{jk}(X_j,X_k) = - g_\text{syn}(X_k,t) \big[ v\big(X_{j}(t)\big) - E_\text{syn}\big] ;
  \label{eq:gjk_example}
  \end{equation} 
  here, $v(X_{j})$ projects the vector $X_j$ onto its corresponding voltage value, $E_\text{syn}$ is a reversal potential, and $g_\text{syn}(X_k,t)\ge 0$ denotes the synaptic conductance that is activated by the action potential of neuron $k$.
  As long as the temporal dynamics of $g_\text{syn}(t)$ is short and does not violate Assumption \ref{item3}, \cref{eq:thm3} is a valid description for the smooth pulsatile synaptic transmission resulting from \cref{eq:gjk_example}.
\end{remark}

\begin{remark}
  The coupling in the canonical $\theta$-neuron network of \cref{thm:1} is of pulse-response type $q(\theta_j)p(\theta_k)$ and reflects that weakly connected Class 1 excitable neurons \eqref{eq:thm1} have Type I-phase response curves (PRC); 
  for Type I-PRCs, positive (negative) inputs always lead to a positive (negative) phase shift of the postsynaptic neuron. 
  In \cref{eq:thm3}, the PRCs reduce to $q(\theta_j) = 1+\cos\theta_j$. 
  The Type I-property of the PRCs is inherited from the fact that each neuron is assumed to be $\varepsilon^2$-close to the SNIC bifurcation,
  and it is independent from the coupling function $G_{jk}$. Consequently, the pulses $p_{jk}(\theta_k)$ can indeed represent conductance-based synapses of the form \eqref{eq:gjk_example}.
  Relaxing the $\varepsilon^2$-closeness assumption will allow for obtaining also Type II-PRCs\footnote{Note that in the proof of \cref{thm:1}, I approximated the phase $\theta_j$ in the function $p_{jk}(\theta_j,\theta_k)\approx p_{jk}(\theta_k)$ to be constant during an incoming spike, see \cref{eq:app_pulse_def}. This assumption breaks down if neuron $j$ is no longer $\varepsilon^2$-close to the SNIC bifurcation and the PRC $\tilde q(\theta_j)$ has to be found from $\tilde q(\theta_j) \tilde p(\theta_k)=(1+\cos\theta_j)p_{jk}(\theta_j,\theta_k)$.}%
  , which may entail very distinct synchronization properties of the network~\citep[see, e.g.,][]{smeal2010phase,pazo_montbrio_2014}. 
\end{remark}

In this manuscript, I will focus on pulses $p_{jk}(\theta) = s_{jk}p(\theta)$ with $p(\theta)\geq 0$ non-negative. I will include no habituation, no delay nor synaptic kinetics, nor any synaptic fatigue. My point here is to illustrate some of the consequences of the pulse function $p(\theta)$ on the network's collective behavior.
}

\subsection{Voltage-dependent pulses of QIF neurons}
\label{sec:3}

\Cref{eq:thm3} describes the dynamics of a network of $\theta$-neurons with instantaneous pulse-coupling, which has been rigorously derived from a universal class of weakly connected Class 1 excitable neurons.
Assuming weak interactions and homogeneous stereotypical action potential shapes, I now set $s_{jk} = J/N$ for all $j,k=1,\dots,N$. Then, $p_{jk}(\theta_k) = J p(\theta_k)/N$ with a general pulse function $p$ and coupling strength $J \in \mathbb R$ re-scaled with respect to the network size $N$.
By identifying $v_j = \tan(\theta_j/2)$, the $\theta$-neuron becomes equivalent to the QIF neuron (\cref{fig:cartoon}b and c; \cite{Ermentrout_1996}) and \cref{eq:thm3} can be transformed into a network model of QIF neurons, whose voltage variables $v_j$ follow the subthreshold dynamics
\begin{equation}
\frac{d}{d\tau} v_j = v_j^2 + \eta_j + \frac{J}{N} \sum_{k=1}^N p_k(\tau) .
\label{eq:qif1}
\end{equation}
The pulses $p_k$ received by postsynaptic neuron $j$,
\begin{equation}
	p_k(\tau) = p\big(\theta_k(\tau)\big) = p\big(2\arctan v_k(\tau) \big),
	\label{eq:pulse1}
\end{equation}
are then implicitly defined in terms of the presynaptic voltage $v_k$ through the corresponding $\theta$-phase.
Because of the connection with the $\theta$-model~\eqref{eq:thm3}, the pulses $p_k$ may already represent complex synaptic transmission through, e.g., conductance-based synapses of the form $I_\text{syn} = g_\text{syn}(v_j,t) [E_\text{syn}-v_k(t)]$, cf.~\cref{eq:gjk_example}.
This pulse-interpretation becomes rigorous due to the exact  correspondence between QIF and $\theta$-neurons when the QIF dynamics \eqref{eq:qif1} is equipped with a fire-and-reset rule that takes peak and reset potentials at infinity (\cref{fig:cartoon}).
\changeB{The equivalence with the $\theta$-model is not the only raison d'\^etre of the QIF model (\cref{sec:2intro}), and there are alternative biologically plausible interpretations for pulses of finite width between interacting QIF neurons.
In particular, the definition \eqref{eq:qif1} in terms of voltage variables $v_j$ calls for an interpretation of \eqref{eq:pulse1} as voltage-dependent pulses independent from the derivation of the canonical pulse-coupled $\theta$-model~\eqref{eq:thm3}. 

\begin{remark}\label{remark2}
The QIF dynamics \eqref{eq:qif1} were obtained through the forward transformation $v_j = \tan(\theta_j/2)$ from \cref{eq:thm3}, where the $\theta$-neuron represents the canonical model for a Class 1 excitable neuron close to a SNIC bifurcation.
As a matter of course, one can start with interacting QIF neurons and use the inverse transformation $\theta_j = 2 \arctan(v_j)$ to obtain a network model of $\theta$-neurons as in \cref{eq:thm3}, which does not necessarily represent the canonical model for Class 1 neurons, as has been proposed in~\citep{borgers_kopell_2005,kotani2014population}.
\end{remark}
}

\begin{figure}[!t]
  \centering
  \includegraphics[width=0.6\columnwidth]{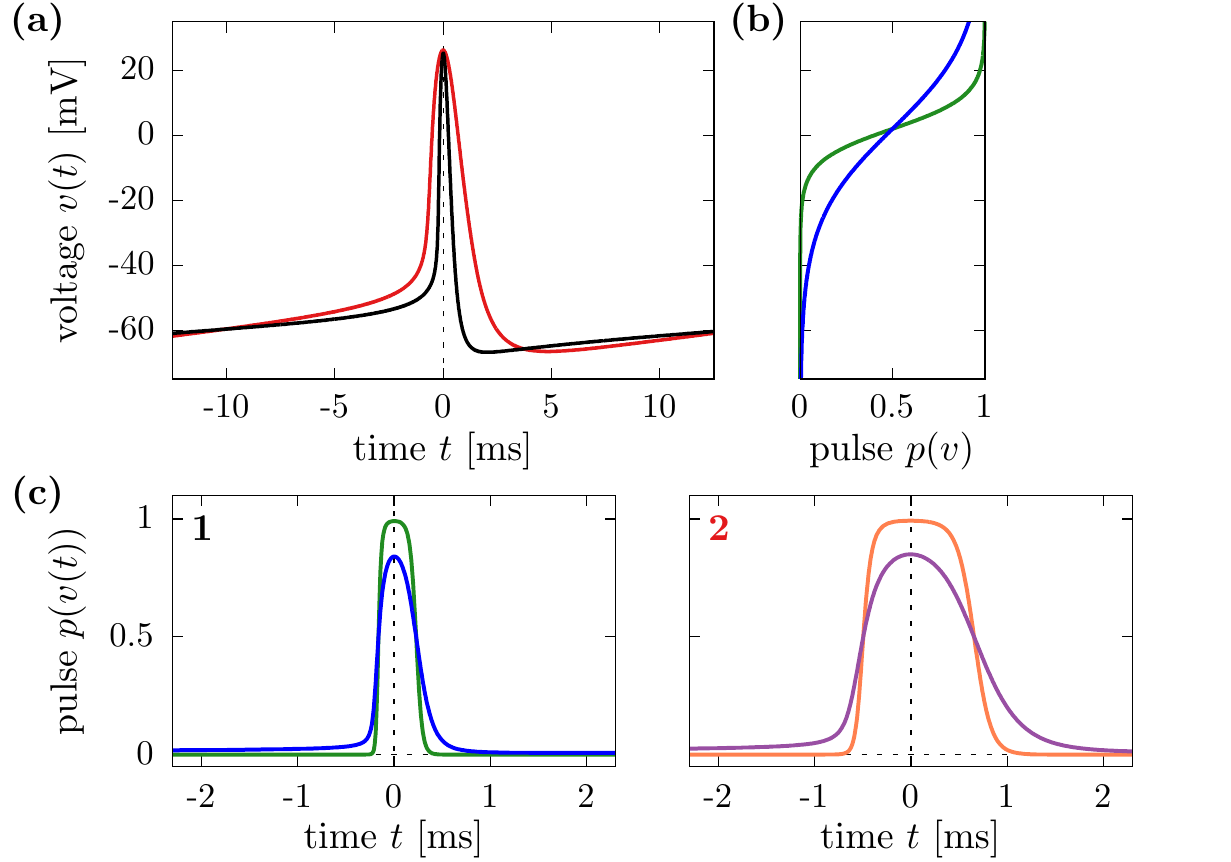}
  \caption{Voltage-dependent pulses of the conductance-based Wang-Buzs\'aki~(WB) neuron. (a) Voltage traces of a periodically spiking WB neuron with narrow (black) or broad (red) action potentials (APs), both have frequency $f\approx1/(24\text{ms})=41.7$Hz.
  (b) Sigmoidal pulse function \eqref{eq:sigmoid} with hard (green) and soft (blue) thresholds.
  The APs in (a) are transformed via (b) into pulses:
  (c1) a narrow AP and hard threshold result in a square-like pulse (green) that is symmetric about the spike time $t=0$, whereas a soft threshold smooths and skews the pulse (blue);
  (c2) similar to (c1) but for the wide AP.
  Parameters of pulses in (b): $v_s=2$mV, $k_s=5$ (green) as in \citep{Destexhe1994}; $v_s=2$mV, $k_s=14$ (blue).}
  \label{fig1}
\end{figure}

\begin{figure}[!ht]
\centering
\vspace{0.5cm}
\includegraphics[width=0.6\columnwidth]{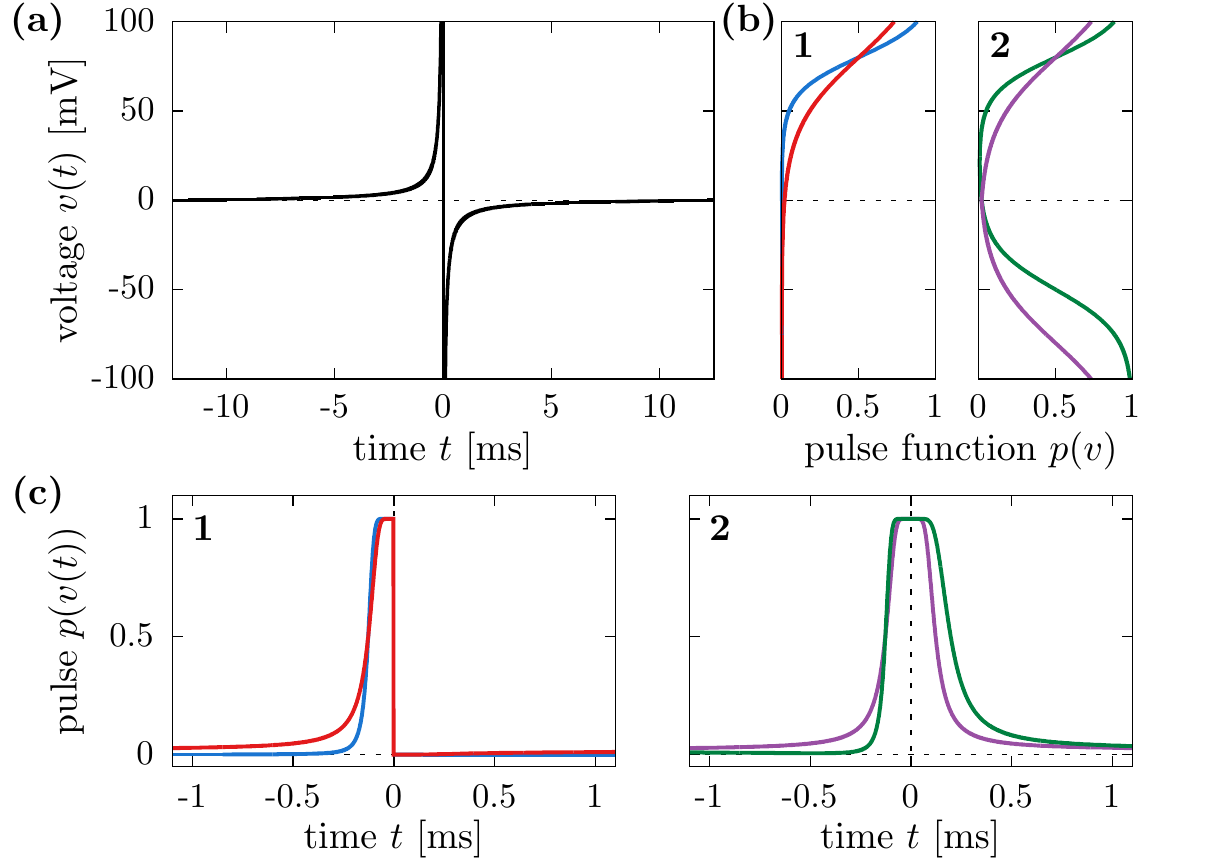}
\caption{Voltage-dependent pulses of the QIF neuron.
(a) Voltage trace of a periodically spiking QIF neuron with frequency $f\approx 41.7$Hz.
(b) Pulse functions consisting of (b1) one sigmoidal function as in \cref{fig1}(b) or of (b2) a combination of two sigmoidals as in \cref{eq:two_sigmoids}.
(c1/c2) The action potential of the QIF neuron is transformed into the pulses using the pulse functions in (b1/b2).
Pulse parameters:
(b1) \cref{eq:sigmoid} with $v_s=80$mV, $k_s=10$ (blue) and $k_s=20$ (red).
(b2) $p_\parallel(v)$ with $v_s=80$mV, $k_s=20$ (violet); 
$p_\nparallel(v)$ with $v_{s-}=-50$mV, $k_{s-}=12.5$, $v_{s+}=80$mV, $k_{s+}=10$ (green).}
\label{fig1b}
\end{figure}

How can one introduce biologically plausible pulses of finite width between QIF neurons solely taking the QIF dynamics \eqref{eq:qif1} into account?
A natural candidate are voltage-dependent pulses $p_k = p(v_k)$ that can be identified with the synaptic conductance $g_\text{syn}(v_k,t)$, where the presynaptic voltage acts instantaneously on $g_\text{syn}$, i.e.~$g_\text{syn}(v_k,t)=g_\text{syn}(v_k(t))$. 
Such an instantaneous relation has, e.g., been manifested through the voltage-gated calcium influx at the presynaptic axon during an action potential~\citep{catterall2011voltage} or through the sigmoidal relationship between neurotransmitter concentration and presynaptic voltage~\citep{Destexhe1994}. 
One may additionally consider first- or higher-order synaptic kinetics of $g_\text{syn}(t)$ in response to a presynaptic pulse $p(v_k)$; for simplicity, however, I will focus on instantaneous interactions (but see \cref{subsec:D2}).
In conductance-based neuron models, instantaneous voltage-dependent pulses during synaptic transmission have frequently been described by sigmoidal functions of logistic\footnote{An alternative formulation of \cref{eq:sigmoid} uses the hyperbolic tangent, $t_\infty(v) = \lbrace 1+\tanh[(v-v_t)/k_t]\rbrace/2$, which coincides with $s_\infty(v)$ when choosing $v_s=v_t$ and $k_s=k_t/2$.} form
\begin{equation}
    s_\infty(v) = 1 / \lbrace 1 + \exp [-(v-v_s)/k_s ]\rbrace,
    \label{eq:sigmoid}
\end{equation}
with presynaptic voltage $v$ and steepness and threshold parameters $k_s$ and $v_s$~\citep{ermentrout_kopell_1990,Wang_Rinzel_1992,Destexhe1994,golomb1996propagation,Wang_Buzsaki_1996,Ermentrout_Kopell_1998,kopell2000gamma,borgers2017introduction}.
Exemplary pulse profiles for conductance-based Hodgin-Huxley-like neurons are shown in \cref{fig1}; see \cref{appsec:WB} for details on the employed Wang-Buzs\'aki model.
The pulse shape depends on the width of the action potential and on the slope $k_s$ of the sigmoidal pulse function $s_\infty(v)$.
Narrow action potentials and hard thresholds generate pulses that are almost symmetric about the spike time $t=0$ (green pulse in \cref{fig1}c1).
Typically, however, pulses exhibit a fast rise shortly before and a slower decay after the spike (violet pulse in \cref{fig1}c2).

\changeB{QIF neurons, by comparison, have stereotyped and rather artificial action potentials (\cref{fig1b}a).
Using the sigmoidal pulse function $s_\infty(v)$, \cref{eq:sigmoid}, leads to degenerate pulses that terminate sharply at the time of the spike (\cref{fig1b}c1).
To compensate} for the hardwired artificial shape of the QIF's action potential,
one can extend the pulses ad hoc and replace $v$ by its absolute value $|v|$, see the violet pulse function in \cref{fig1b}(b2).
This tactic generates pulses $p_\parallel(v) := s_\infty(|v|)$ that are symmetric about the spike time $t=0$ (violet pulse in \cref{fig1b}c2).
A symmetric pulse is a reasonable approximation of pulses in conductance-based neuron models with narrow action potentials and a sigmoidal function with a hard threshold $k_s \to 0$ (green pulse in \cref{fig1}c1).

\begin{figure}[!t]
    \centering
    \includegraphics[width=0.6\columnwidth]{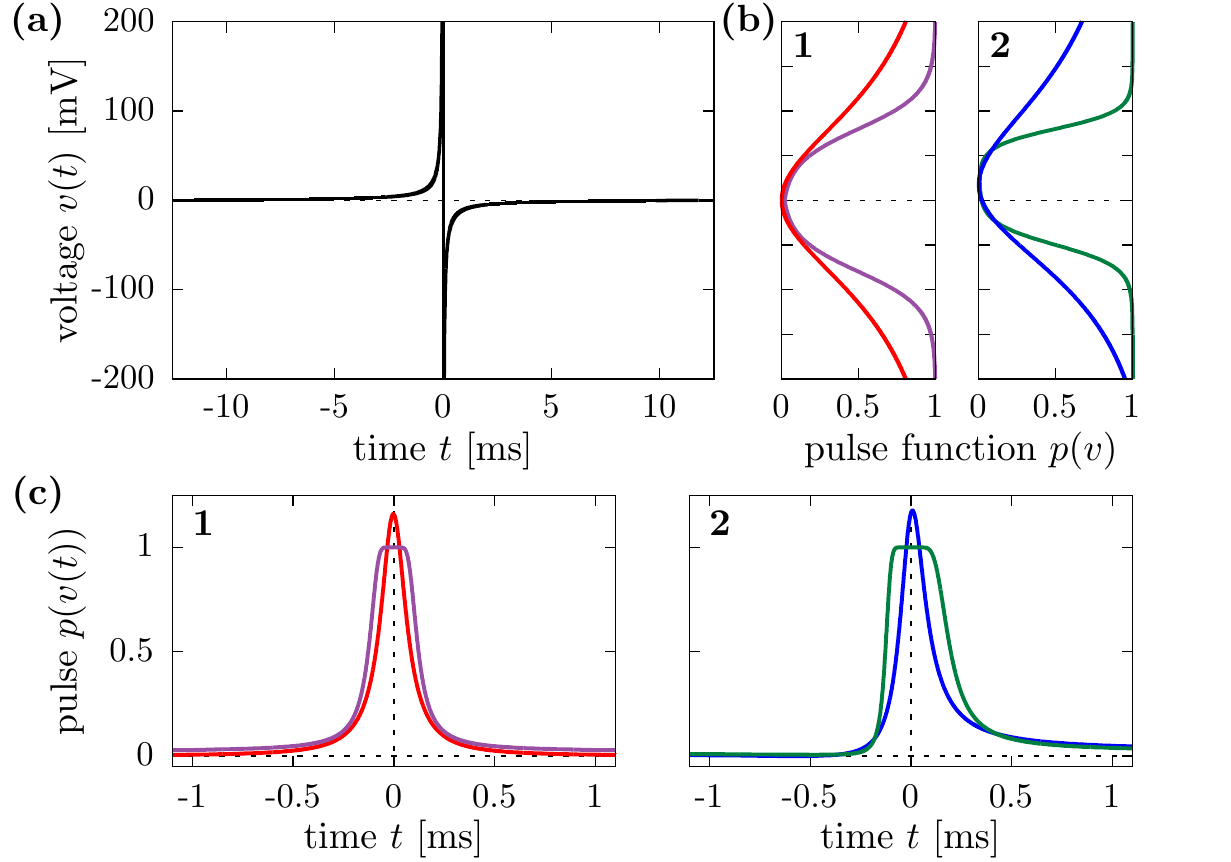}
    \caption{Smooth voltage-dependent pulses of (a) a periodically spiking QIF neuron with frequency $f\approx 41.7$Hz.
    (b) The pulse functions \eqref{eq:two_sigmoids} shown in \cref{fig1b}(b2) can be approximated around the resting potential by the phase-dependent pulse function \eqref{eq:KJ_theta} (red and blue traces; rescaled for convenience).
    The resulting smooth pulses in (c) capture the nature of symmetric (c1) and asymmetric (c2) pulses quite accurately.
    Parameters for $p_{r,\varphi,\psi}(2\arctan v)/c$ are: (b1) $r=0.985, \varphi=0,\psi=\pi,c=115$; (b2) $r=0.985, \varphi=\pi/12,\psi=\pi,c=35$.}
    \label{fig1c}
\end{figure}%
For general asymmetric pulses as in \cref{fig1}(c), however, one has to come up with a different solution that breaks the symmetry of $p_\parallel(v)$ with respect to $v=0$. 
For instance, one can combine two sigmoidal functions as
\begin{equation}
    p_\nparallel(v) = \frac1{1+e^{(v-v_{s-})/k_{s-}}} + \frac1{1+e^{-(v-v_{s+})/k_{s+}}}
    \label{eq:two_sigmoids}
\end{equation}
with thresholds $v_{s+}>0> v_{s-}$ and steepness parameters $k_{s\pm}>0$.
For $k_{s+}<k_{s-}$, QIF neurons emit asymmetric pulses with a steep upstroke and a moderate downstroke (green pulse function in \cref{fig1b}b2, the resulting pulse $p_\nparallel(v(t))$ in \cref{fig1b}c2).

\subsection{An accessible family of pulse functions}
\label{subsec:accessible}
In anticipation of an exact low-dimensional description for the collective dynamics of QIF neurons, it is advantageous to refrain from the explicit voltage-dependence of the pulses \eqref{eq:two_sigmoids} and to formulate them in the corresponding $\theta$-phase description~\citep{pcp_arxiv_2022}.
This allows not only for a directly comparable treatment of networks of $\theta$- and QIF neurons, but it also overarches the opposing approaches to pulse-coupling introduced in \cref{sec:2,sec:3}, see also \cref{subsec:D2}.

By substituting $v=\tan(\theta/2)$ in \cref{eq:two_sigmoids}, the voltage-dependent pulse $p_\nparallel(v)$ becomes phase-dependent. However, the pulse $p_\nparallel(\tan(\theta/2))$ is only implicitly defined in $\theta$, which hampers analytic tractability.
As an alternative, I propose to use pulses $p(\theta)$ of the form~\eqref{eq:pulse1} that start from the $\theta$-phase description and a priori assume an implicit voltage-dependence;
in special cases the implicit voltage-dependence can actually become explicit, see, e.g., \cref{eq:RP_voltage} below. 
Importantly, the explicit pulses~\eqref{eq:two_sigmoids} can be approximated quite accurately around the resting potential by pulses $p(\theta) = p(2\arctan v)$ that employ the $\theta$-phase transformation of the QIF neuron (\cref{fig1c}b).
Specifically, I propose the following family of accessible, analytically favorable, smooth pulse functions
\begin{align}
p_{r,\varphi,\psi}(\theta) &=1+ \frac{1-r^2}{1-r \cos \varphi} \frac{\cos(\theta-\psi-\varphi)-r\cos(\varphi)}{1-2r\cos(\theta-\psi) +r^2} \; ,
\label{eq:KJ_theta}
\end{align} 
which provide adequate approximations of the QIF pulses that can be symmetric about the spike time (\cref{fig1c}c1) or asymmetric with a steep upstroke and a more moderate downstroke after the spike (\cref{fig1c}c2).
The shape of the pulses~\eqref{eq:KJ_theta} is determined by three parameters $r\in[-1,1], \varphi\in [-\pi,\pi)$ and $\psi \in[0,2\pi)$:
$r<1$ tunes the width of the pulse, $\varphi\neq 0$ tunes its asymmetry, and $\psi > \pi$ ($\psi < \pi$) shifts the pulse to the right (left) from the spike (at $\theta=\pi$).
Crucially, these pulses admit an exact, low-dimensional description for the collective dynamics of globally coupled spiking neurons in terms of their firing rate and mean voltage (\cref{sec:4}), and thus allow for a comprehensive mathematical analysis how the pulse shape affects the synchronization properties of the network (\cref{sec:5}). 
%

To provide more details on the pulse functions \eqref{eq:KJ_theta}, first note that
for $\varphi=0$ and $\psi=\pi$, \cref{eq:KJ_theta} reduces to a ``Rectified-Poisson'' (RP) pulse $p_{\text{RP},r}(\theta) := p_{r,0,\pi}(\theta)$ with sharpness parameter $r\in (-1,1)$~\citep{gallego_et_al_2017}.
Via the transform $\theta=2\arctan(v)$, one obtains the explicit voltage-dependence of the RP pulse
\begin{equation}
 p_{\text{RP},r}(v) = \frac{2(1-r)v^2}{(1+r)^2+(1-r)^2v^2} \ ,
 \label{eq:RP_voltage}
\end{equation}
see \cref{fig1c}(b1,c1) for an example.
The formulation in terms of the voltage $v$ readily allows for using the pulses interchangeably in networks of $\theta$- and QIF neurons.
In the same manner, one can also obtain the general voltage description of \cref{eq:KJ_theta}; as it is more convoluted, I refrain from presenting it here.

The RP pulse~\eqref{eq:RP_voltage} is a smooth generalization of the discontinuous $\delta$-spike:
When a neuron spikes and its voltage $v$ diverges, then $\lim_{v\to\infty}p_{\text{RP},r}(v)=2/(1-r)$.
Thus, in the limit $r\to 1$, $p_{\text{RP},1}(v)$ becomes a Dirac $\delta$-pulse and is zero except when $v$ diverges. 
On the other hand, for $r\to-1$, $p_{\text{RP},-1}(v)=1$ is flat and the pulse is always ``on'' independent of the neuron's actual state.
In the case $r=0$, the RP pulse reads $p_{\text{RP},0}(v)=2v^2/(1+v^2)$,
which corresponds to a cosine-pulse in the $\theta$-phase description, $p_{\text{RP},0}(\theta)=(1-\cos \theta)$.

Traditionally, cosine-pulses have been generalized as ``Ariaratnam-Strogatz'' (AS) pulses according to~\citep{ariaratnam2001phase},
\begin{equation}
    p_{\text{AS},n}(\theta)=a_n (1-\cos\theta)^n, 
    \label{eq:AS_pulse}
\end{equation}
where $a_n=2^n(n!)^2 / (2n)!$ is a normalization constant and $n\in\mathbb N$ a sharpness parameter:
for $n=1$ one finds $p_{\text{AS},1}(\theta)=p_{\text{RP},0}(\theta)$; for $n\to\infty$, the AS pulse \eqref{eq:AS_pulse} converges to a Dirac $\delta$-pulse.
AS pulses have widely been used in network models of $\theta$-neurons~\citep{goel2002synchrony,Ermentrout_2006,luke_barreto_so_2013,luke2014macroscopic,So-Luke-Barreto-14,Laing_2014,laing_2015,laing2016bumps,laing2016travelling,roulet2016average,laing2017,chandra2017modeling,laing2018dynamics,laing2018chaos,aguiar2019feedforward,lin_barreto_so_2020,laing2020effects,laing2020moving,means2020permutation,blasche2020degree, bick_goodfellow_laing_martens_2020,juettner_martens_2021,omel2022collective,birdac2022dynamics}, but suffer from numerical and also analytical difficulties when the pulses become more and more localized, $n \gg 1$.
While RP pulses~\eqref{eq:RP_voltage} share the properties of AS pulses---they are symmetric about the pulse peak at $\theta=\pi$ and vanish at $\theta=v=0$---, RP pulses overcome the numerical and analytical shortcomings of the AS pulse~\eqref{eq:AS_pulse}, as will become clear below. 
In addition, the RP pulses can be readily extended to pulses that are skewed and/or whose peak is shifted away from $\theta=\pi$ (where $v=\tan(\theta/2)$ diverges) by introducing the asymmetry and shift parameters $\varphi\neq 0$ and $\psi\neq \pi$ in \cref{eq:KJ_theta}.

Pulses of the form \eqref{eq:KJ_theta} satisfy some convenient analytic properties as they correspond to a rescaled Kato-Jones distribution~\citep{kato_jones_2015}.
The Kato-Jones distribution is a four-parameter family of unimodal distributions on the circle, whose general form reads
\begin{equation}
    p_{KJ}(\theta)=1+2\; \! \text{Re}\Big\{C\sum_{n=1}^\infty (\xi e^{-i\theta})^n \Big\}
    \label{eq:KJ_distribution}
\end{equation}
with the complex-valued parameters $C=a e^{i\varphi}$ and $\xi=r e^{i\psi} \in\mathbb C$.
The Kato-Jones distribution \eqref{eq:KJ_distribution} generalizes the wrapped Cauchy distribution, which is obtained for $C\equiv1$.
One obtains \cref{eq:KJ_theta} when setting $a=(1-r^2)/(2r)/(1-r\cos\varphi) \geq 0$.
This choice guarantees that the pulse is always non-negative, $p_{r,\varphi,\psi}(\theta)\ge 0$, and vanishes at least once.
Alternatively, one could reduce the number of parameters by enforcing \cref{eq:KJ_distribution} to vanish always at $\theta=0$, which is in line with \cref{thm:1} but may entail pulses that change signs. 
In this manuscript, however, I focus on pulses $p(v)\ge 0$ so that the distinction whether the pulse-coupling is excitatory or inhibitory is uniquely determined by (the sign of) the coupling strength $J$.

Taken together, the proposed family of pulses \eqref{eq:KJ_theta} satisfies the properties:
$p_{r,\varphi,\psi}(\theta)$ is $2\pi$-periodic, non-negative, unimodal, vanishes at least at one point, and has normalized area in the course of an action potential, $\int_0^{2\pi} p_{r,\varphi,\psi}(\theta)d\theta=2\pi$.
Moreover, the pulses $p_{r,\varphi,\psi}(\theta)$ are smooth, except for the limit $r\to 1$, in which they converge to a Dirac $\delta$-pulse, 
\begin{equation}
    p_{\delta,\psi}(\theta) := \lim_{r\to1} p_{r,\varphi,\psi}(\theta) = 2\pi \delta(\theta-\psi).
    \label{eq:delta_pulse}    
\end{equation}
The phase $\theta=\psi$, at which the $\delta$-pulse is emitted, can be linked to a ``threshold'' voltage $v_\text{thr}=\tan(\psi/2)$ that, according to~\cite{Ermentrout_1996}, ``accounts for the possibility that the synaptic conductance begins before the presynaptic voltage reaches its maximal value'', cf.~also~\citep{gutkin2001turning}.
If $\psi>\pi$, then $v_\text{thr} < 0$ and the $\delta$-pulse is emitted after the presynaptic voltage has reached its maximal value $v_p$ and the neuron is recovering from the reset after its spike.

\section{Exact collective dynamics of globally pulse-coupled QIF neurons}
\label{sec:4}

To determine the effect of the pulse shape on the collective network dynamics, I will make use of a recently proposed exact reduction for globally coupled spiking neurons.
\change{I present the theory only for networks of QIF neurons, but the exact reduction equally holds for the corresponding $\theta$-neuron dynamics, see \cref{appsec:A,appsec:B}}.  Building on \cref{eq:qif1}, I consider the membrane potentials $v_{j}$ of QIF neurons $j=1,\dots,N$, that follow the subthreshold dynamics~\citep{Izhikevich_2007,Ermentrout_1996,latham_et_al_2000} 
\begin{subequations}
\begin{equation}
\tau_m\dot{v}_j = v_j^2 + I_0 + I_\text{syn} + I_j \;.
\end{equation}
The membrane time constant $\tau_m$ is typically of the order $\tau_m=10$ms, $I_0$ is a global input common to all neurons, 
$I_\text{syn}$ a global recurrent synaptic input to be defined below, 
and $I_j$ describes neuron-specific, independent inputs
\begin{equation}
    I_j(t) = \gamma [ c\;\! \eta_j + (1-c)\xi_j(t)]
    \label{eq:QIF_individual_input}
\end{equation}
\label{eq:QIF}
\end{subequations}
comprising heterogeneous and noisy inputs; $c \gamma$ determines the degree of heterogeneity and $(1-c)\gamma$ is the noise intensity.
Quenched heterogeneity $\eta_j$ (as in \cref{eq:qif1}) is sampled from a normalized Cauchy-Lorentz distribution.
$\xi_j(t)$ describes independent Cauchy white noise with $\langle \xi_j(t) \rangle_t=0$ and $\langle \xi_j(t)\xi_k(s) \rangle_t = \delta_{j,k}\delta(t-s)$.
The parameter $c\in[0,1]$ in \cref{eq:QIF_individual_input} weighs the relative deterministic and stochastic contributions to $I_j$ on the microscopic level, but does not have an impact on the macroscopic dynamics~\citep{clusella_montbrio_2022,pcp_arxiv_2022}.
The subthreshold dynamics~\eqref{eq:QIF} is complemented by a fire-and-reset mechanism:
upon reaching a threshold $v_{p}$, the voltage $v_j$ is reset to the potential $v_\text{r}$ and neuron $j$ is said to elicit a spike.
Spike times $T_j$ of neuron $j$ are defined implicitly by $v_j(T_j)=v_p$.
As the quadratic term in \cref{eq:QIF} causes the voltage to diverge in finite time, I consider $v_p = -v_r = \infty$, so that the QIF neuron is equivalent to the $\theta$-model~\citep{Ermentrout_1996} when identifying the voltage $v_j$ with a phase $\theta_j$ via the transformation $v_j = \tan(\theta_j/2)$;
neuron $j$ spikes when $\theta_j$ crosses $\pi$, see \cref{fig:cartoon}. 

The spike times $T_j^{k=1,2,\dots}$ of neurons $j=1,\dots,N$ allow for linking the theoretical model, \cref{eq:QIF}, with experimental observations via the population activity commonly defined in terms of the firing rate
\begin{equation}
	R_N(t) = \lim_{\tau_r\to 0} \frac1N \sum_{j=1}^N \sum_k \frac1{\tau_r} \int_{t-\tau_r}^t \delta(t-T^k_j) dt,
	\label{eq:RN}
\end{equation}
where the subscript $N$ indicates the network size and $\tau_r$ is a small time window over which to sum spikes.
The firing rate $R$ is at the heart of mean-field models in computational neuroscience, whose ultimate goal is to provide a self-consistent, at best exact (i.e.\ matching an underlying microscopic network model), dynamic description of the population activity that is closed in a few macroscopic variables.
A singular example for such an exact mean-field model was proposed by \cite{montbrio_pazo_roxin_2015} for QIF neurons; see~\citep{luke_barreto_so_2013,Laing_2014} for previous work on the related $\theta$-neurons, yet without providing an explicit differential equation for the firing rate $R(t)$. 
Montbri\'o et al.'s approach yielded ordinary differential equations that describe the dynamics of the firing rate $R$ and the mean voltage $V = \langle v_j \rangle = \tfrac{1}{N}\sum_{j=1}^N v_j$.
A major success of Montbri\'o et al.'s firing rate equations was to pinpoint a spike synchronization mechanism through the cooperative interplay between the neurons' voltage and firing dynamics that is central for collective network oscillations, inherent in almost all neuron network models, but not captured by traditional firing rate models by default.
In the following, I will employ a reduction strategy that builds on the ideas in \citep{montbrio_pazo_roxin_2015} and which has been developed further in \citep{pcp_arxiv_2022}.

\subsection{Firing rate and voltage (RV) dynamics}
A remarkable feature of the globally coupled QIF neurons described by \cref{eq:QIF} is that their collective dynamics is captured by a low-dimensional system of differential equations~\citep{montbrio_pazo_roxin_2015,pcp_arxiv_2022}.
The low-dimensional description becomes exact in the thermodynamic limit of infinitely many neurons, $N\to\infty$, which I adopt from now on.
The macroscopic state of the QIF network is then given by the probability density function $\mathcal W(v,t)$, so that $d\mathcal W(v,t)dv$ indicates the fraction of neurons with membrane potential in the interval $[v,v+dv)$ at time $t$.
Correspondingly, in the $\theta$-phase description one can formulate the probability density $\mathcal P(\theta,t)$ with variable $\theta=2\arctan(v)$.
The latter distribution density can be expanded in Fourier space as $\mathcal P(\theta,t) = (2\pi)^{-1} \{ 1 + \sum_{n\geq 1}  Z_n(t) e^{-in\theta} + c.c. \}$, where the modes $Z_n(t)$ are the Kuramoto-Daido order parameters~\citep{Kuramoto_1984,daido_1996} and their dynamics are given in Appendix~\ref{appsec:B}.
The two characteristic observables of neural networks---the firing rate $R$ and mean voltage $V$---can conveniently be expressed as
\begin{equation}
 \pi \tau_m R - iV = 1 +  2\sum_{n=1}^\infty (-1)^n Z_n = \Phi + \lambda \frac{\mathcal M(-\sigma)}\sigma,
 \label{eq:RV_Pls}
\end{equation} 
where $\Phi(t), \lambda(t), \sigma(t)$ are complex-valued variables and
$\mathcal M(k)$ is a constant function that depends on the initial distribution $\mathcal W(\theta,t=0)$ of the QIF neurons.
As detailed in~\citep{pcp_arxiv_2022}, see also Appendix \ref{appsec:B}, the dynamics of the three complex variables $\Phi, \lambda, \sigma\in \mathbb C$ are governed by
\begin{align}
\tau_m\dot{\Phi} = i \Phi^2 - i I(t) + \gamma, \quad
\tau_m \dot{\lambda} = 2i\Phi\lambda, \quad  \tau_m \dot{\sigma} = i\lambda,
\label{eq:Pls}
\end{align}
where $I(t)= I_0 + I_\text{syn}(t)$; note that the macroscopic dynamics~\eqref{eq:Pls} are independent of the parameter $c$ in \cref{eq:QIF_individual_input}.
In the presence of Cauchy white noise and/or Cauchy-Lorentz heterogeneity, $\gamma >0$, $\lambda \to 0$ asymptotically in time \citep{pcp_arxiv_2022}, and the collective dynamics becomes uniquely determined by $\Phi \to \pi \tau_m R - iV$ or, equivalently, by the firing rate $R$ and the mean voltage $V$.
More precisely, the collective dynamics converges to an invariant two-dimensional manifold \citep{ott_antonsen_2008,luke_barreto_so_2013,montbrio_pazo_roxin_2015}, also called the Lorentzian manifold $\{ \lambda = 0 \}$, which contains all possible attractors of the full six-dimensional dynamics~\eqref{eq:Pls} and is given by the time-dependent total voltage density of the QIF neurons in form of a Cauchy-Lorentz distribution with mean $V(t)$ and half-width $\pi\tau_m R(t)$~\citep{montbrio_pazo_roxin_2015},
\begin{equation}
    \mathcal W(v,t)= \frac1\pi \frac{\pi\tau_m R(t)}{ [ v- V(t)]^2 + [\pi\tau_m R(t)]^2} \; .
    \label{eq:CL_density}
\end{equation}
For the analysis of asymptotic regimes, it suffices to restrict the focus on the firing rate and voltage (RV) dynamics on the invariant Lorentzian manifold, which can be guaranteed by initializing the voltages $v_j(0)$ according to a Cauchy-Lorentz distribution\footnote{For different choices of initial conditions $\{v_j(0)\}_j$, one has to resort to the full six-dimensional dynamics \eqref{eq:Pls} and define $\mathcal{M}(k)$ appropriately \citep[see][]{pcp_arxiv_2022}.}.
Setting $\lambda\equiv 0$ in \cref{eq:RV_Pls,eq:Pls} leads to the RV dynamics on the Lorentzian manifold~\citep{montbrio_pazo_roxin_2015,pcp_arxiv_2022,clusella_montbrio_2022}
\begin{subequations}
\begin{align}
\tau_m\dot{R} &= \frac{\gamma}{\pi\tau_m} + 2RV \;, \\
\tau_m\dot{V} &= V^2 - (\pi \tau_m R)^2 + I_0 + I_\text{syn} \; ,\label{eq:FRE_Vdyn}
\end{align}
\label{eq:FRE}
\end{subequations}
that exactly describes the collective dynamics of large pulse-coupled networks of QIF neurons, \cref{eq:QIF}.
In the next step, I will incorporate smooth pulsatile synaptic transmission mediated by the pulses~\eqref{eq:KJ_theta} and show how the global recurrent current $I_\text{syn}$ can be expressed in terms of the macroscopic variables $R$ and $V$, which leads to an exact mean-field model that is closed in the population firing rate $R$ and mean voltage $V$.


\subsection{Recurrent synaptic input in the RV dynamics}
\label{subsec:recurrent}

\cite{montbrio_pazo_roxin_2015} considered QIF neurons that interacted with each other in a global (``all-to-all'') and instantaneous fashion through ``$\delta$-spikes'', i.e.\ by emitting Dirac $\delta$-pulses at their spike times $T^k_j$. Then, 
$I_\text{syn}$ in \cref{eq:QIF,eq:FRE} becomes proportional to $\tau_m R(t)$ as given by \cref{eq:RN}.
However, the $\delta$-spike assumption for recurrent, instantaneous coupling is particularly limiting~\citep{afifurrahman2021collective} 
not only for biological, but also for numerical reasons\footnote{
Dirac $\delta$-pulse interactions $\dot x(t) = f(t,x(t)) + g(t,x(t))\delta(t-t_0)$ for some functions $f$ and $g$ should be understood throughout this manuscript as 
$\dot x(t) = f(t,x(t))$ for $t \neq t_0$ and
$x(t_0^+) = x(t_0^-) + g(t_0,x(t_0^-))$ at $t = t_0$~\citep[see, e.g.,][]{feketa2021survey}.
}~\citep{catlla2008spiking,Klinshov_et_al_2021,feketa2021survey} and also
because the limit of infinitely narrow pulses does not commute with the thermodynamic limit of infinitely large networks~\citep{zillmer2007stability,montbrio_pazo_roxin_2015} nor does it allow for collective oscillations~\citep{montbrio_pazo_roxin_2015,ratas_pyragas_2016,juettner_martens_2021}.
Here, I avoid those difficulties by modeling interneuronal communication with smooth pulses 
of finite width that unfold in the course of an action potential of presynaptic neuron $j$ in a similar manner as neurotransmitters are released in response to the depolarization of the voltage $v_j$  (\cref{sec:3}). 

For globally coupled QIF neurons, the recurrent input $I_\text{syn}(t)=\langle s_{j}(t)\rangle$ is the population mean over all individual postsynaptic responses $s_{j}(t)$ to a pulse $p_j(t)$ emitted by presynaptic neuron $j$. 
As discussed above, I consider voltage-dependent, smooth pulses $p_j(t) = p\big(\theta_j(t)\big) = p\big(2\arctan v_j(t) \big)$ with the pulse function $p(\theta)=p_{r,\varphi,\psi}(\theta)$ given by \cref{eq:KJ_theta}.
The $\theta$-phase formulation of the pulses guarantees analytic tractability thanks to favorable properties of the corresponding probability density $\mathcal P(\theta,t)$ of the ensemble of $\theta$-neurons~\citep{pcp_arxiv_2022}, but does not limit the generality of results for general voltage-dependent pulses. 
Moreover, the pulses~\eqref{eq:KJ_theta} allow for a smooth approximation of $\delta$-spikes,  \cref{eq:delta_pulse} with $\psi=\pi$. 
Summing over all neurons $j=1,\dots,N,$ yields the connection to the firing rate $R$ in \cref{eq:RN} as:
\begin{equation}
P_{N,\delta} := \frac1N \sum_{j=1}^N p_{\delta,\pi}(\theta_j) = \frac{2\pi}N \sum_{j=1}^N \frac{\delta(t-T_j^k)}{2/\tau_m} = \pi \tau_m R_N,
\label{eq:P2R}
\end{equation}
where the first equality follows from the change of variables formula for Dirac $\delta$-functions:
$\delta(t-T_j^k) = \delta\big( g(t) \big) \big| \dot g(T_j^k) \big|$ with $g(t) = \theta_j(t) - \pi$ and $\dot g(T_j^k) = \dot \theta_j(T_j^k) = 2/\tau_m$. The sum over the spike times $T_j^k$ is taken within a short time window of length $\tau_r$ as in \cref{eq:RN}.

\Cref{eq:P2R} explicitly links the mean presynaptic pulse activity $P_{\delta,\pi}=\langle p_{\delta,\pi}\rangle$ to the firing rate $R$, see also \cref{eq:app_dd}.
Ironically, the $\delta$-spike assumption again conceals an even deeper connection: 
in the thermodynamic limit $N\to\infty$, the mean pulse activity $P_{r,\varphi,\psi}$ is fully determined by the firing rate $R(t)$ \emph{and} the mean voltage $V(t)$. 
The description of $P_{r,\varphi,\psi}$ in terms of $R$ and $V$ becomes explicit on the Lorentzian manifold, where
\begin{gather*}
    P_{r,\varphi,\psi} (t) = \left\langle p_{r,\varphi,\psi} \big( \theta_j(t)\big) \right\rangle =\lim_{N\to\infty} \frac{1}{N}\sum_{j=1}^N p_{r,\varphi,\psi} \big( \theta_j(t)\big) \\
    = \int_0^{2\pi} p_{r,\varphi,\psi}(\theta) \mathcal P(\theta,t) d\theta = P_{r,\varphi,\psi}\big(R(t),V(t)\big)\;.
\end{gather*}
In general, $P_{r,\varphi,\psi}$ depends on $R$ and $V$ only through the three complex variables $\Phi,\lambda,\sigma$, see \cref{eq:appD_eq1} in \cref{sec:app_pulses}, but on the Lorentzian manifold $\{ \lambda=0\}$ it reduces to
\begin{equation}
\begin{gathered}
P_{r,\varphi,\psi} (R,V) = \text{Re}  \left( \Psi \right) \quad \text{where} \quad \Psi= \\
\frac{(1-r^2)(1+\pi\tau_m R - iV)e^{-i\varphi} + (r-\cos\varphi) \big[ 1-re^{-i\psi}+ (\pi\tau_m R - iV)(1+re^{-i\psi})\big]}
{r(1-r\cos\varphi)  \big[ 1-re^{-i\psi}+ (\pi\tau_m R - iV) (1+re^{-i\psi})\big]} 
\end{gathered}
\label{eq:p_mean}
\end{equation}
see \cref{appsec:C,sec:app_pulses,sec:colvar_derivation} for a rigorous derivation.
Albeit complex in its general description in terms of $R$ and $V$ as well as of the parameters $r,\varphi$ and $\psi$, 
\cref{eq:p_mean} dramatically simplifies for certain pulse shapes.

For RP pulses $p_{\text{RP},r}=p_{r,0,\pi}$, see  \cref{eq:RP_voltage}, that are symmetric ($\varphi=0$) about the peak phase when the neuron spikes ($\psi=\pi$), the mean presynaptic pulse activity $P_{\text{RP},r} = \langle p_{\text{RP},r}\rangle = P_{r,0,\pi}$ reads:
\begin{align}
\hspace{0.0cm}P_{\text{RP},r}(R,V) =
\frac{2\pi\tau_m R[1+r+(1-r)\pi\tau_m R] + 2(1-r)V^2}{[1+r+(1-r)\pi \tau_mR]^2 + (1-r)^2V^2}.\label{eq:RP_mean}
\end{align}
Taking the limit $r\to1$,
yields the mean activity of $\delta$-spikes, $\lim_{r\to1} P_{r,0,\pi} = \pi \tau_m R$, which coincides with the population firing rate as already shown in \cref{eq:P2R}.

Dirac $\delta$-pulses do not necessarily need to be emitted at the instant a neuron spikes (when $v\to\infty$), but when $v$ crosses a virtual threshold voltage $v_\text{thr} < \infty$.
The corresponding peak phase is $\psi=2 \arctan(v_\text{thr})$ and the mean presynaptic activity $P_{\delta,v_\text{thr}} := \langle p_{\delta, 2\arctan(v_\text{thr})}\rangle$ for pulses of the form \cref{eq:delta_pulse} becomes
\begin{equation}
P_{\delta,v_\text{thr}}(R,V) = \frac{\pi\tau_m R (1+ v_\text{thr}^2)}{(\pi\tau_m R)^2 + (V- v_\text{thr})^2} \; .
\label{eq:deltav_mean}
\end{equation}

Importantly, for the general family \eqref{eq:KJ_theta} of pulse functions $p_{r,\varphi,\psi}$, the mean presynaptic pulse activity \eqref{eq:p_mean} explicitly depends on the mean voltage $V$, i.e. $\partial_V P_{r,\varphi,\psi} \neq 0$, except for the limit $(r,\varphi,\psi)\to(1,0,\pi)$.
As I will show in \cref{subsec:5A}, it is exactly this voltage-dependence that is crucial for collective oscillations in case of instantaneous pulse-coupling.
The emergent macroscopic oscillations, however, are sensitive to the pulse shape:
Skewing the pulses slightly to the phase after the spike yields robust collective oscillations of inhibitory neurons, which are not present for symmetric pulse-coupling (\cref{fig:sim}). 
A thorough analysis how the pulse shape affects network synchronization is the focus of \cref{sec:5}. 
To anticipate, for RP pulses that are symmetric about the peak phase $\theta=\pi$, collective oscillations are restricted to a narrow and rather unrealistic region in parameter space. 
By contrast, for (asymmetric) pulses with their peak phase after the actual spike, $\theta>\pi$, collective oscillations emerge almost naturally in inhibitory networks ($J<0$) with excitatory input currents ($I_0 >0$).
This strongly reminds of interneuronal network gamma (ING) oscillations in inhibitory networks including synaptic kinetics with finite rise and decay times~\citep{Wang_Buzsaki_1996,brunel_hakim_1999,brunel_wang_2003}.

As a wrap-up, I have presented an exact low-dimensional description for globally pulse-coupled QIF neurons~\eqref{eq:QIF} in the thermodynamic limit, which equally holds for networks of $\theta$-neurons, see \cref{eq:theta}.
For pulses $p_{r,\varphi,\psi}$ of the general form \eqref{eq:KJ_theta}, the mean pulse activity $P_{r,\varphi,\psi}=\langle p_{r,\varphi,\psi}\rangle$ can conveniently be expressed in terms of the macroscopic variables. 
The time-asymptotic macroscopic dynamics of the QIF neurons is restricted to the invariant Lorentzian manifold, on which the collective behavior is exactly described by the RV dynamics \eqref{eq:FRE}.
For instantaneous pulse-coupling, one has to substitute
\begin{equation}
    I_\text{syn}=J P_{r,\varphi,\psi}(R,V)
    \label{eq:inst_pulse_coupling}
\end{equation}
in \cref{eq:QIF}, where $J$ indicates the pulse-coupling strength and the mean pulse activity $P_{r,\varphi,\psi}$ is given by \cref{eq:p_mean}.
The RV dynamics \eqref{eq:FRE} is two-dimensional and closed in two macroscopic variables---the firing rate $R$ and the mean voltage $V$.

\begin{figure}[!t]
\centering
\includegraphics[width=0.5\textwidth]{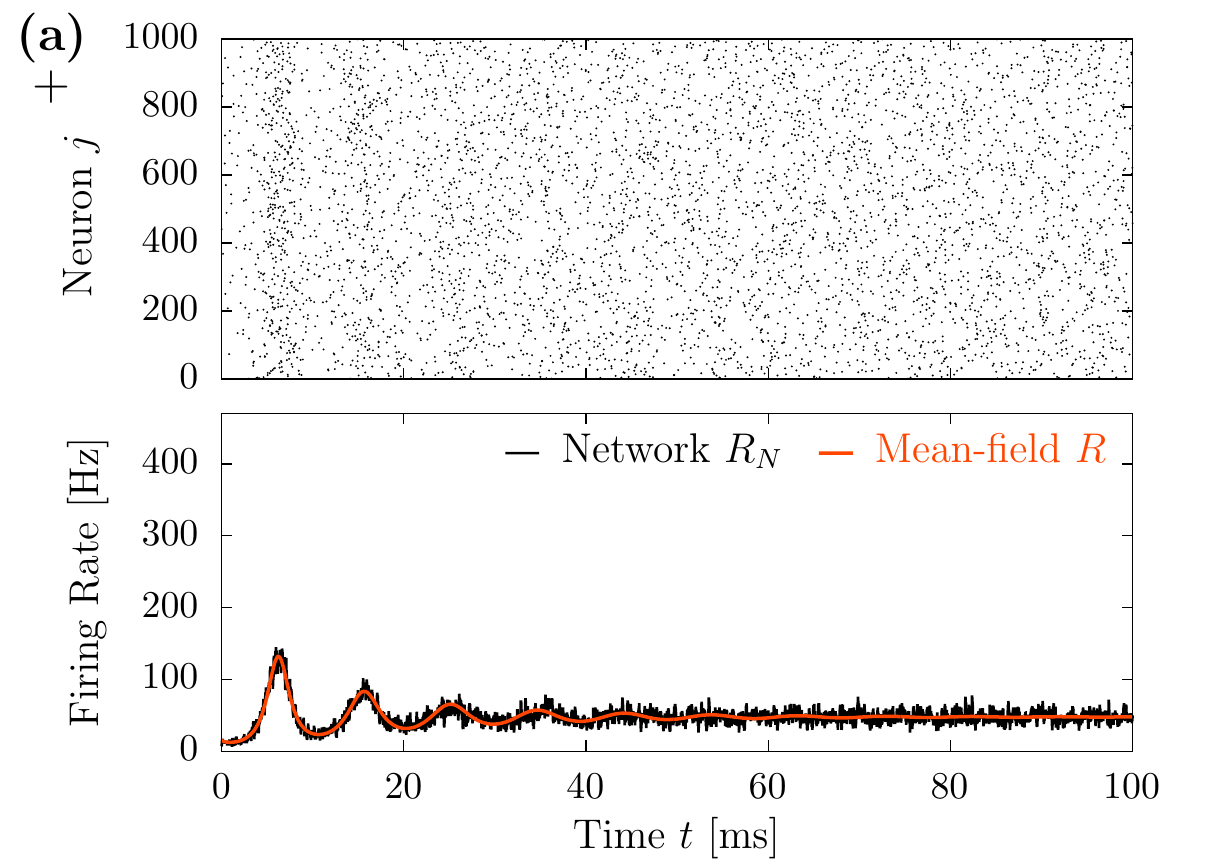}
\hspace{-0.25cm}
\includegraphics[width=0.5\textwidth]{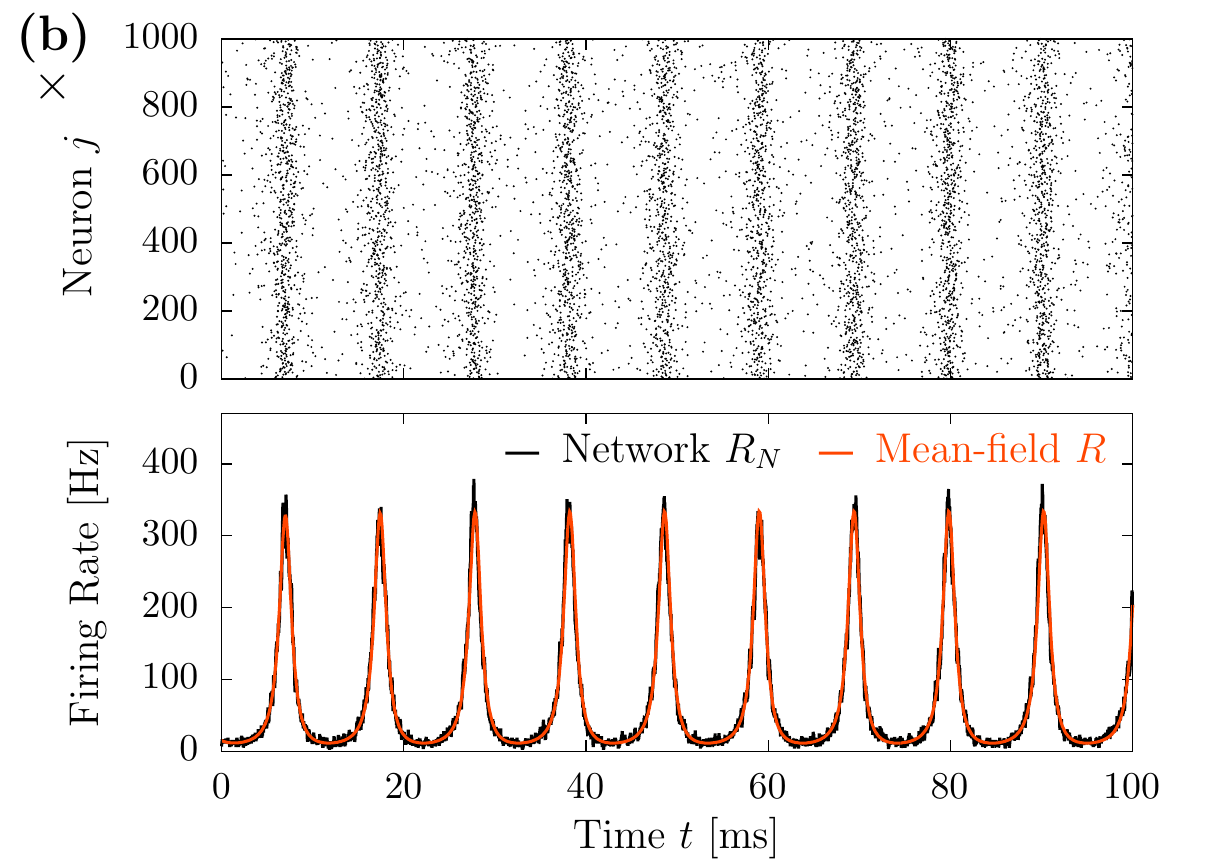}
\caption{Collective dynamics of inhibitory QIF neurons globally coupled via instantaneous pulses $p_{r,\varphi,\psi}$ with $r=0.95$ and $\psi=\pi$.
(a) Symmetric RP pulses ($\varphi=0$, red curve in \cref{fig:RP_pulses_all}a) drive the network into an asynchronous state.
(b) Asymmetric pulses slightly skewed to the phase after the spike ($\varphi=\pi/12$, red curve in \cref{fig:KJ_pulses_all}a) yield robust ING oscillations.
Top: raster plots show spike times of $1'000$ QIF neurons obtained from network simulations according to \cref{eq:QIF} with $N=10'000$ neurons.
Bottom: Excellent agreement between firing rates $R_N$ (black) and $R$ (orange) obtained from integrating the QIF network and the RV dynamics \eqref{eq:FRE} with Eqs.~\eqref{eq:RP_mean}
and \eqref{eq:KJasymm_mean} for (a) and (b), respectively. 
Network parameters are $\gamma=1, I_0=20, J=-12$; see $+$ in \cref{fig:RP_pulses_all}(c4) and $\times$ in \cref{fig:KJ_pulses_all}(d).
The membrane time constant is $\tau_m=10$~ms. 
Network simulations were performed using Euler-Maruyama integration with $dt=5\cdot10^{-3}$ms and the firing rate $R_N$ was computed using \cref{eq:RN} with $\tau_r=10^{-1}$ms.}
\label{fig:sim}
\end{figure}

The reduction of the RV dynamics greatly facilitates the investigation how the pulse shape affects collective dynamics of pulse-coupled spiking neurons, as the mean-field equations \eqref{eq:FRE} accurately reproduce the microscopic network dynamics \eqref{eq:QIF} of globally coupled QIF neurons with instantaneous pulses $p_{r,\varphi,\psi}$ (\cref{fig:sim}).
Initial voltages $v_j(0)$ of the microscopic network were chosen to follow a Cauchy-Lorentz distribution of half-width $R(0)$ and centered at $V(0)$ for an immediate match between network and RV simulations. 
For arbitrary initial conditions, one has to resort to the six-dimensional dynamics (\ref{eq:RV_Pls}~\&~\ref{eq:Pls}) with a properly chosen constant function $\mathcal M(k)$ for such a perfect agreement, see \citep{pcp_arxiv_2022}.

In the following, I will capitalize on the fact that the reduced RV dynamics \eqref{eq:FRE} perfectly capture network simulations of pulse-coupled spiking neurons,
and analyze them with respect to emergent collective behavior for various choices of the pulse function~\eqref{eq:KJ_theta}.
In doing so, I also show that the collective dynamics in \cref{fig:sim} can be predicted by linear stability analysis, see $+$ in \cref{fig:RP_pulses_all}(c4) and $\times$ in \cref{fig:KJ_pulses_all}(d) below.
The main focus of \cref{sec:5} lies on voltage-mediated synchronization, in general, and on collective oscillations in inhibitory networks, in particular, with instantaneous pulses whose peak occurs at the time of the presynaptic spike or shortly thereafter.

\section{Collective oscillations with instantaneous pulses}
\label{sec:5}
Networks of inhibitory interneurons provide a mechanism for coherent brain oscillations, particularly in the gamma band, as reciprocal inhibition turned out to be effective for neuronal synchrony~\citep{wang2010review}.
For Class 1 excitable neurons, such as $\theta$-neurons or the QIF model \eqref{eq:QIF}, collective oscillations---as a hallmark of synchrony---can be realized even with relatively fast inhibitory synapses~\citep{Ermentrout_1996}. 
Synaptic latency (including axonal delay) 
as well as rise and decay times contribute to determining synchronous firing patterns~\citep{vanVreeswijk1994,Wang_Buzsaki_1996,white1998synchronization,brunel_wang_2003,maex_schutter_2003,brunel_hansel_2006,wang2010review}.
Exact RV dynamics corresponding to \cref{eq:FRE} have been successfully employed to describe collective oscillations of inhibitory QIF neurons interacting via $\delta$-spikes with first-order~\citep{devalle2017,dumont2019macroscopic} and second-order synaptic kinetics~\citep{byrne2017mean,coombes_byrne_2019,clusella_ruffini_2022} or with delay~\citep{pazo_montbrio_2016,ratas2018macroscopic,Devalle2018}; see also \citep{keeley_byrne_2019,byrne2020next,byrne2022mean,di2018transition,bi2020coexistence,di2022coherent,avitabile2022cross} for extensions.
For instantaneous synapses, however, $\delta$-spike-interactions do not allow for macroscopic oscillations~\citep{montbrio_pazo_roxin_2015,juettner_martens_2021},
whereas instantaneous pulses of finite width have been reported to induce collective oscillations in inhibitory networks of the equivalent $\theta$-neurons~\citep{luke_barreto_so_2013,So-Luke-Barreto-14,laing_2015,lin_barreto_so_2020,juettner_martens_2021} interacting via symmetric AS pulses~\eqref{eq:AS_pulse}.
For globally coupled QIF neurons with (non-smooth) pulses of finite width, collective oscillations have so far only been found in excitatory networks, cf.~Eq.~(29) in~\citep{ratas_pyragas_2016}.

\changeB{In this section, I systematically study how the pulse shape affects collective oscillations in the two-dimensional RV dynamics~\eqref{eq:FRE} with instantanteous pulse-coupling.
I focus on identifying the transitions from asynchronous to synchronous collective behavior upon varying the coupling strength $J$ and the mean input $I_0$. Mathematically speaking, I will perform bifurcation analysis of macroscopic stationary (asynchronous) states and characterize  stability changes when varying control, or bifurcation, parameters.
When the asynchronous state loses stability and gives rise to collective oscillations (synchrony), these oscillations can be distinguished right at their onset having either (a) small amplitude and finite frequency, (b) high amplitude and finite frequency, or (c) high amplitude and arbitrarily low frequency.
Cases (a) and (b) correspond to a so-called Hopf bifurcation that is either super- or subcritical; case (c) corresponds to a SNIC or a homoclinic bifurcation.
The simplest scenario and quite common for collective oscillations to emerge is a supercritical Hopf bifurcation. 
The other cases require bistability of multiple steady states, and one typically finds saddle-node bifurcations close-by.

Before exploring the bifurcation structure of the collective dynamics for particular pulse shapes, I first prove in Section~\ref{subsec:5A} that the RV dynamics~\eqref{eq:FRE} with instantaneous synaptic transmission can generate collective oscillations through a Hopf bifurcation only if the recurrent synaptic input explicitly depends on the mean voltage $V$.
Such an effective voltage-coupling occurs for pulses of finite width, \cref{eq:RP_mean}, but also for Dirac $\delta$-pulses emitted before or after, but not at, the presynaptic spike, \cref{eq:deltav_mean}.
However, for symmetric RP pulses, as I show in \cref{subsec:3B}, collective oscillations are only feasible in a small range of inhibitory coupling ($J<0$) and mean inputs $I_0>0$, limited by Hopf and homoclinic bifurcation boundaries. 
For asymmetric right-skewed pulses, the bifurcation structure simplifies and yields collective oscillations for a large range of inhibitory coupling and positive mean inputs, see \cref{subsec:3C}, substantiating the notion of \cref{fig:sim}(b) that narrow pulses slightly skewed to the phase after the spike are a promising candidate for generating ING oscillations.
For completeness, I investigate the onset of collective oscillations for left-skewed pulses in \cref{subsec:3D}, though they are biologically less plausible and require recurrent excitation instead of inhibition.
}


\subsection{Collective oscillations \change{through} voltage-coupling}
\label{subsec:5A}
\changeB{The collective dynamics for networks of globally coupled spiking neurons with instantaneous synaptic transmission is exactly described by the RV dynamics \eqref{eq:FRE}.
On the Lorentzian manifold, the RV dynamics is two-dimensional and closed in the firing rate $R$ and mean voltage $V$ as the recurrent synaptic input $I_\text{syn}(t)$ depends instantaneously, and smoothly, on $R$ and $V$, i.e.\ $I_\text{syn}(t)=I_\text{syn}(R(t),V(t))$. 
}%
The onset of collective oscillations via a Hopf bifurcation can be determined using the eigenvalues 
\begin{equation*}
    \lambda_{1,2}= \frac12 \left[ \text{tr}(\text{Jac}) \pm \sqrt{\text{tr}(\text{Jac})^2 - 4\text{det}(\text{Jac})} \right]
\end{equation*} of the Jacobian of the RV dynamics \eqref{eq:FRE},
\begin{align}
	\text{Jac}= \begin{pmatrix}
	2V^*& 2R^*\\
	-2\pi^2\tau_m^2 R^* + \partial_R I_\text{syn}^* & 2V^* + \partial_V I_\text{syn}^*
	\end{pmatrix}\;, \label{eq:Jacobian}
\end{align}
evaluated at the fixed-point solution $(R^*,V^*)$ with 
\begin{equation}
    V^* = -\gamma/(2\pi\tau_m R^*) \quad \text{and} \quad I_\text{syn}^* = I_\text{syn}(R^*,V^*) ;
    \label{eq:fixed_point_relations}
\end{equation} 
the Jacobian~\eqref{eq:Jacobian} is obtained in rescaled time $\tilde t=\tau_m t$, and $\partial_x = \partial/ \partial x$ with $x\in\{R,V\}$.
\changeB{By definition \eqref{eq:RN}, the firing rate $R\geq 0$ is non-negative such that the fixed-point solution $(R^*,V^*)$ satisfies $V^*\leq 0$ because of \cref{eq:fixed_point_relations}; the inequality becomes strict for $\gamma > 0$.}

A Hopf bifurcation of the fixed point $(R^*,V^*)$ occurs if $\text{tr}(\text{Jac})=0$ and $\text{det}(\text{Jac}) >0$.
From $\text{tr}(\text{Jac})=0$, one has $\partial_V I_\text{syn}^* = - 4 V^*$. Since $V^*<0$ is negative for $\gamma>0$, the necessary condition for collective oscillations in \cref{eq:FRE} emerging through a Hopf bifurcation of the fixed point $(R^*,V^*)$ reads:
\begin{equation}
    \partial_V I_\text{syn}^* := \partial_V I_\text{syn}(R^*,V^*) > 0 \ . 
    \label{eq:Hopf_condition}
\end{equation}
That is, collective oscillations require recurrent coupling through the mean voltage $V$.
Put differently, if the recurrent synaptic input $I_\text{syn}$ is independent of $V$, i.e.~$\partial_V I_\text{syn}^* = 0$, then \cref{eq:Hopf_condition} cannot be satisfied for any fixed-point solution $(R^*,V^*)\in\mathbb R^+\times\mathbb R^-$, and the network of globally coupled QIF neurons will not synchronize. 

To show that collective oscillations are indeed possible if the first condition $\text{tr}(\text{Jac})=0$ is satisfied, it is necessary to prove that the second condition $\text{det}(\text{Jac})>0$ can also be satisfied. 
Assuming that $\text{tr}(\text{Jac})=0$ holds and using the fixed-point equation in \cref{eq:fixed_point_relations}, one obtains 
\[
\text{det}(\text{Jac}) 
=  (2\pi \tau_mR^*)^2 - \frac{\gamma^2}{(\pi \tau_m R^*)^{2}} - \frac{4\gamma}{\pi\tau_m} \frac{\partial_R I_\text{syn}^*}{\partial_V I_\text{syn}^*} \;.
\]
By assumption, $I_\text{syn}(R,V)$ is smooth, so for $\partial_V I^*_\text{syn}>0$ the right-hand side is well-behaved and can be expressed as a smooth function in $R^*$.
While $\text{det}(\text{Jac}) <0$ for small firing rates $R^*\to 0$, e.g.~due to strong inhibitory input $I_0 \ll 0$, the determinant $\text{det}(\text{Jac})$ will become positive for large firing rates $R^*\gg1$ (and/or for $\partial_R I^*_\text{syn}\ll 0$).
As $R^*=R^*(I_0)$ depends smoothly on the external current $I_0$---the functional dependence of $R^*$ on $I_0$, i.e.~the transfer function or $fI$-curve, can be obtained by setting the right-hand side of \cref{eq:FRE_Vdyn} to zero and inserting~\eqref{eq:fixed_point_relations}, see also \citep{devalle2017}---,
one can expect that for strong excitatory drive, $I_0\gg 1$, $R^*$ increases sufficiently such that both conditions, $\text{tr}(\text{Jac})=0$ and $\text{det}(\text{Jac})>0$, are simultaneously satisfied and give rise to collective oscillations.

Hence, only if the recurrent synaptic input $I_\text{syn}$ explicitly depends on the mean voltage $V$ can collective oscillations in the RV-dynamics~\eqref{eq:FRE} emerge through a Hopf bifurcation from a fixed-point solution $(R^*,V^*)$ with sufficiently large firing rate $R^*\gg0$. In other words, an asynchronous, typically high-activity state becomes unstable through an effective voltage-coupling and the spiking neurons start firing synchronously (as in \cref{fig:sim}b).

Now, when does the recurrent input $I_\text{syn}$ explicitly depend on the mean voltage $V$? 
This occurs naturally for electrical coupling through gap junctions, see \citep{Pietras_et_al_2019} and the Discussion~\cref{subsec:D1}.
In the absence of gap junctions, and as shown before in \cref{sec:3}, an effective voltage-coupling can also be achieved with instantaneous pulses that are different from Dirac $\delta$-pulses emitted at the presynaptic spike, see also \citep{ratas_pyragas_2016}.
For instantaneous pulses $p_{r,\varphi,\psi}$ given by \cref{eq:KJ_theta}, the synaptic current $I_\text{syn}(t) = J P_{r,\varphi,\psi}(R,V)$ depends explicitly on the voltage $V$, cf.~\cref{eq:p_mean}, and thus generally allows for collective oscillations of globally coupled QIF neurons; one can achieve that \cref{eq:Hopf_condition} holds true by freely tuning the coupling strength $J$.
Whether collective oscillations occur in realistic parameter regimes, e.g., with global inhibition $J<0$ and excitatory drive $I_0>0$, crucially depends on the pulse-shape parameters $r,\varphi,\psi$, as I will demonstrate in the following \cref{subsec:3B,subsec:3C,subsec:3D}.

To anticipate whether particular pulse shapes $p_{r,\varphi,\psi}$ require excitatory or inhibitory coupling for neural synchronization, one can compute $\partial_V P_{r,\varphi,\psi}(R,V)$ and recall that any fixed-point solution $(R^*, V^*)$ satisfies $R^*>0$ and $V^*<0$ for $\gamma>0$.
Then, for RP pulses \eqref{eq:RP_voltage}, one has
\begin{align}
\partial_V P_{\text{RP},r}(R^*,V^*) =
\frac{4(1-r^2) [1+r+(1-r)\pi\tau_m R^*] V^*}{\left\lbrace[1+r+(1-r)\pi \tau_mR^*]^2 + (1-r)^2{V^*}^2\right\rbrace^2} < 0.\label{eq:dRP_dV}
\end{align}
Thus, to satisfy condition \eqref{eq:Hopf_condition}, $J \partial_V P_{\text{RP},r}(R^*,V^*) >0$, RP pulses require inhibitory coupling ($J<0$) to generate collective oscillations.
For shifted Dirac $\delta$-pulses \eqref{eq:delta_pulse}, the sign of
\begin{align}
\hspace{0.0cm}\partial_V P_{\delta,\text{thr}}(R^*,V^*) =
\frac{2\pi\tau_m R^* (v_\text{thr}-V^*)(1+v_\text{thr}^2)}{[(\pi \tau_mR^*)^2 + (V^*-v_\text{thr})^2]^2}\label{eq:dPdelta_dV}
\end{align}
depends on the virtual threshold value $v_\text{thr}$.
If a neuron emits the Dirac $\delta$-pulse shortly after its spike (i.e.~when it is recovering from its reset, $v \leftarrow -\infty$), then $v_\text{thr} < 0$ and $\partial_V P_{\delta,\text{thr}}(R^*,V^*) < 0$, so collective oscillations occur for inhibitory coupling ($J<0$).
If the pulse is emitted before the neuron spikes, then $v_\text{thr} \gg 0$, $\partial_V P_{\delta,\text{thr}}(R^*,V^*)$ can become positive, and collective oscillations require excitatory coupling ($J>0$).
For other pulse shapes, it is more laborious to draw similar conclusions from $\partial_V P_{r,\varphi,\psi}(R^*,V^*)$ alone.
However, the general trend, which will become clear in the following \cref{subsec:3B,subsec:3C,subsec:3D}, is that collective oscillations occur with inhibitory coupling when the mean of the pulse $p_{r,\varphi,\psi}(\theta)$ coincides with the spike or is shifted to its right, whereas pulses whose mean is shifted to the left of the spike rather require excitatory coupling to synchronize the network.

\subsection{Rectified-Poisson (RP) pulses $(\varphi=0,\psi=\pi)$}
\label{subsec:3B}
For symmetric pulses ($\varphi=0$) of finite width ($r<1$) centered around the spike threshold ($\psi=\pi$),
the mean pulse activity $P_{\text{RP},r}$ is given by \cref{eq:RP_mean}.
In \cref{fig:RP_pulses_all}(a), it can be appreciated that the smaller $r<1$, the wider the RP pulse \eqref{eq:RP_voltage} about the peak phase $\pi$ at which the presynaptic voltage diverges (bottom panel).
The shape of the RP pulses (color coded for different $r$) is similar to AS pulses \eqref{eq:AS_pulse} (grey).
The major advantage of RP pulses is that $P_{\text{RP},r}$ is a simple function in $R$ and $V$, see \cref{eq:RP_mean}, whereas no such closed form exists for arbitrary AS pulses with shape parameter $n\in \mathbb N$, see \cref{appsec:C,sec:app_pulses}. 
It is thus straightforward to perform linear stability analysis of the RV dynamics~\eqref{eq:FRE} with RP pulses, but not with arbitrary AS pulses. 
All bifurcation boundaries can be obtained analytically, except for that of the homoclinic, which is a global bifurcation and has to be determined numerically~\citep{doedel2007auto,gast2019pyrates}.
In the following, I set $\gamma=1$ and refrain from a denormalization with respect to $\gamma$ as it would affect the mean pulse activity $P_{r,\varphi,\psi}$ in a nontrivial way; for different choices of $\gamma >0$, I have not noticed any qualitative differences.

For excitatory coupling ($J>0$), RP pulses are not capable of generating collective oscillations, as expected for reasonable parameter choices and as predicted by \cref{eq:dRP_dV}. 
Instead, the cusp-shaped region in \cref{fig:RP_pulses_all}(b) features bistability between two asynchronous states: a low-activity state (LAS) and a high-activity state (HAS). 
In line with the results for $\delta$-spikes by~\cite{montbrio_pazo_roxin_2015}, the narrower the RP pulse ($r\nearrow1$), the bigger the cusp-shaped region of bistability until it finally coincides with the one found for $r=1$ (black); the boundaries of the bistability regions are SN bifurcations and meet at a codimension-2 bifurcation point (Cusp).

\begin{figure}[!t]
  \centering
  \begin{minipage}[c]{.6\textwidth}
  \centering
  \hspace{-1.65cm}\includegraphics[height=5.25cm]{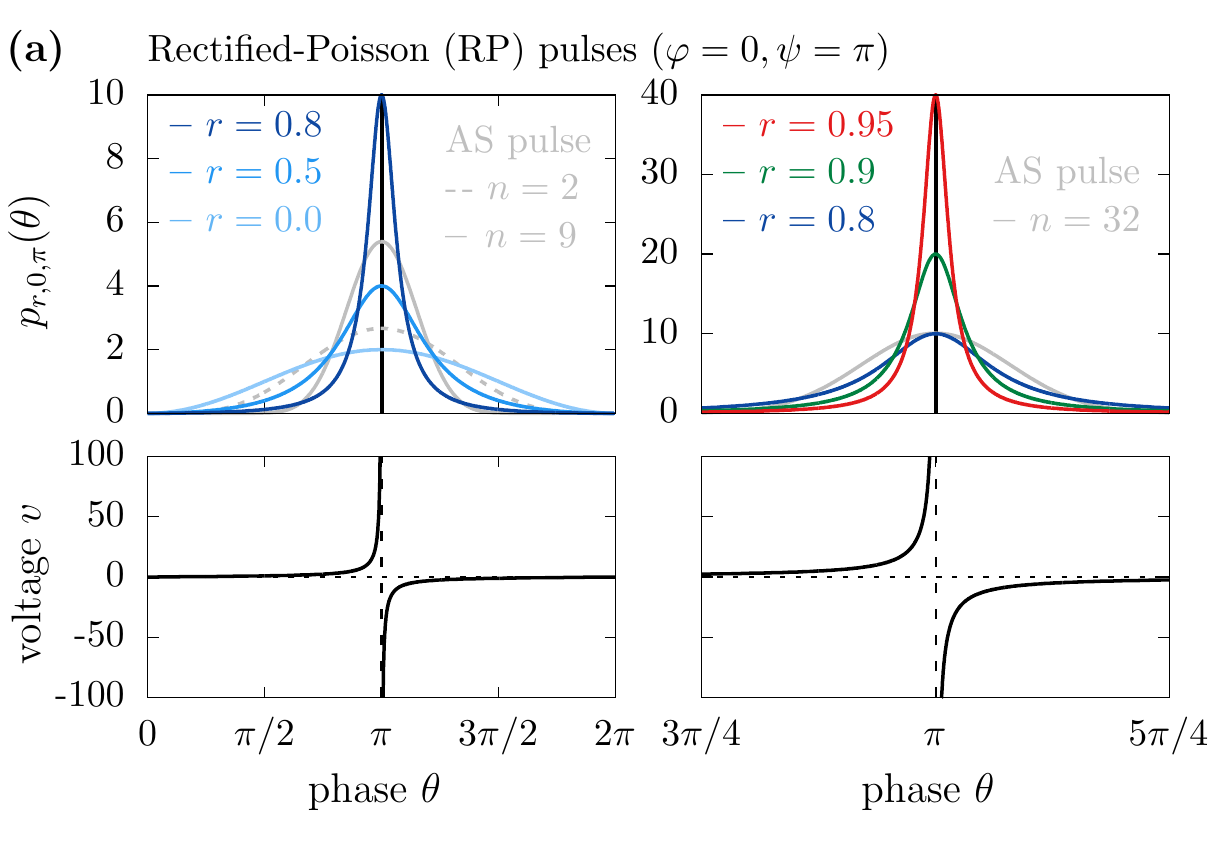}
  \end{minipage}
  \hspace{-1.65cm}\begin{minipage}[c]{.39\textwidth}
  \begin{flushleft}
  \includegraphics[height=4.5cm]{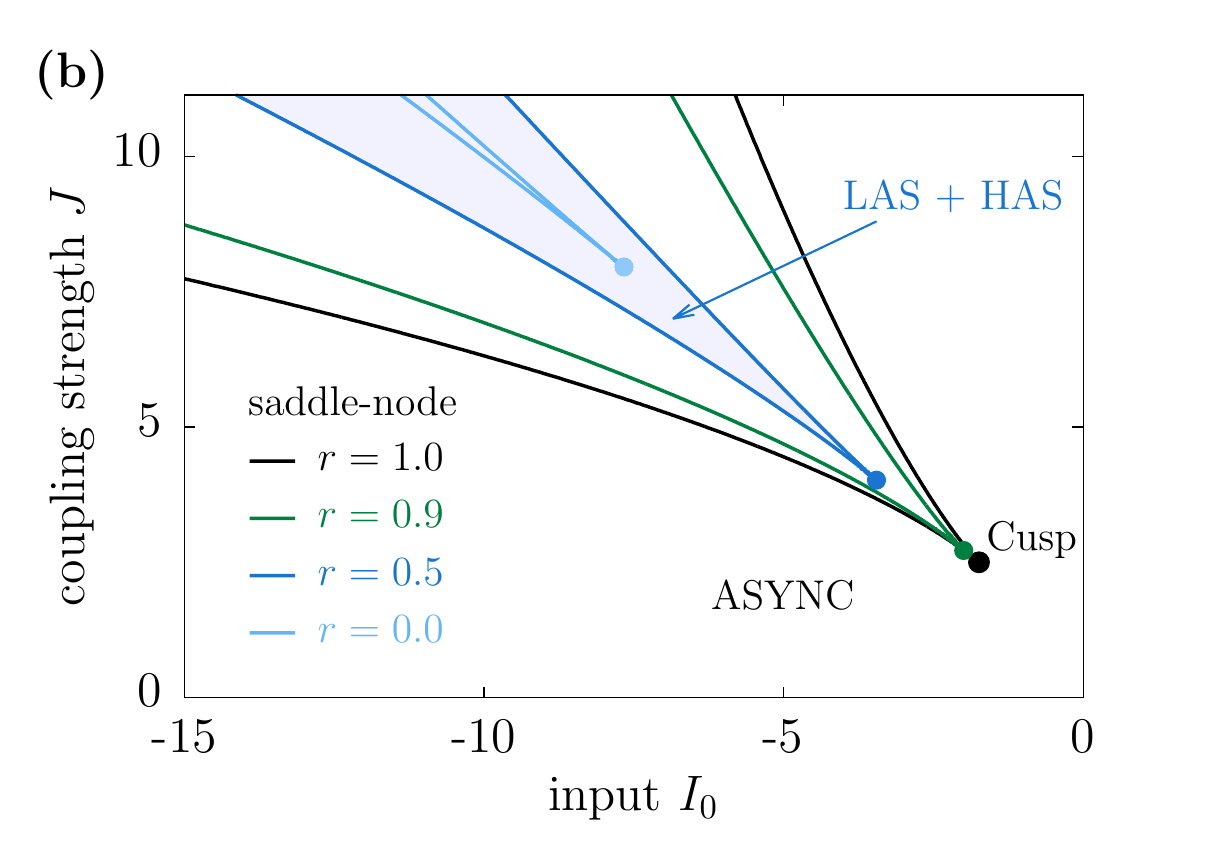}\hfill
  \end{flushleft}
  \end{minipage}\\[-0.25cm]
  \includegraphics[height=4.5cm]{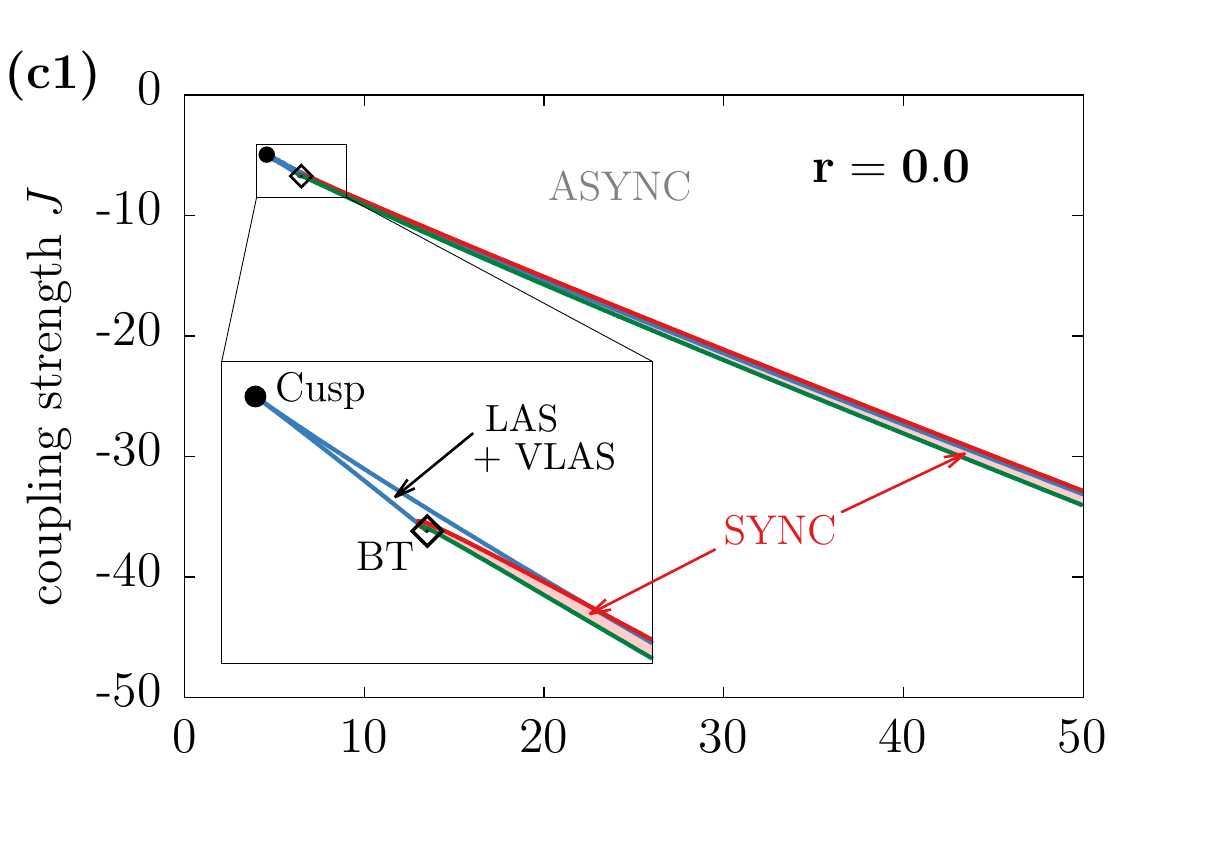}\hspace{-0.5cm}
  \includegraphics[height=4.5cm]{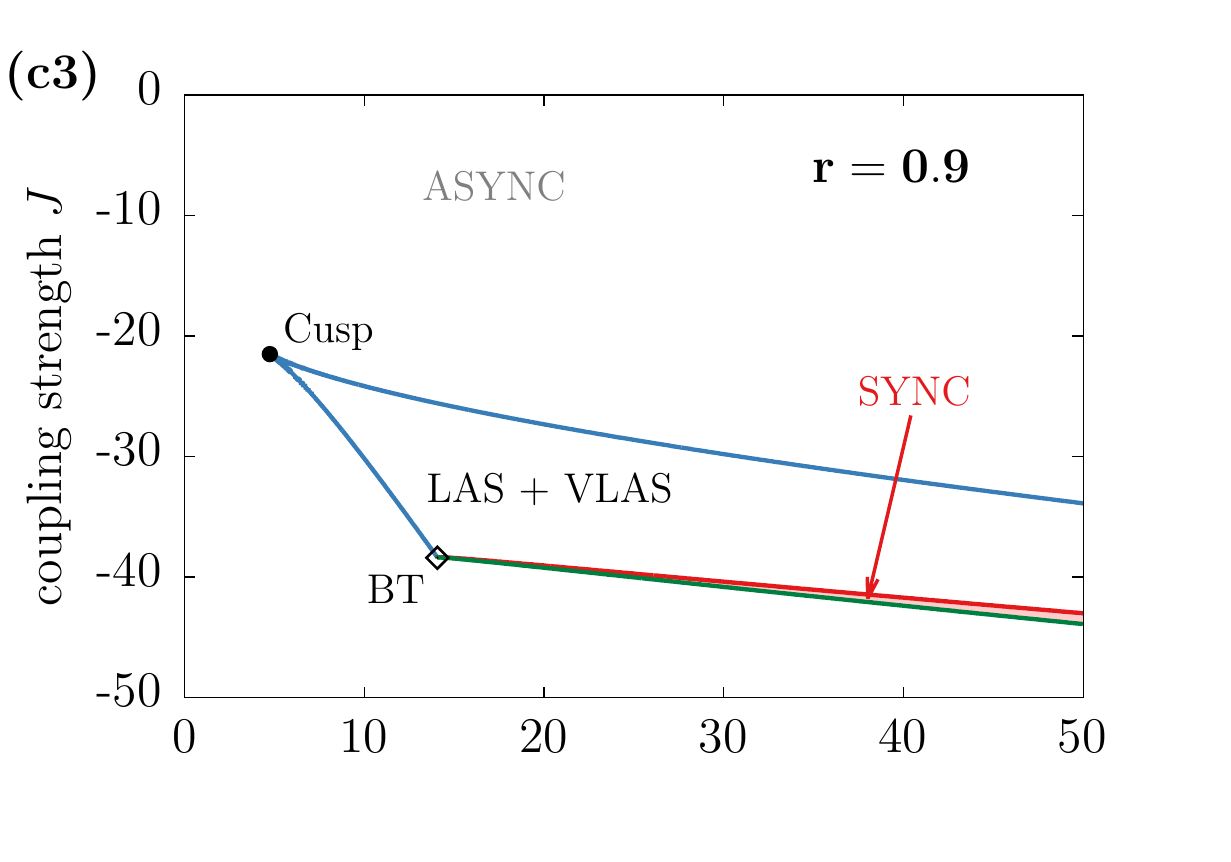}\hspace{0cm}\\[-0.5cm]
  \includegraphics[height=4.5cm]{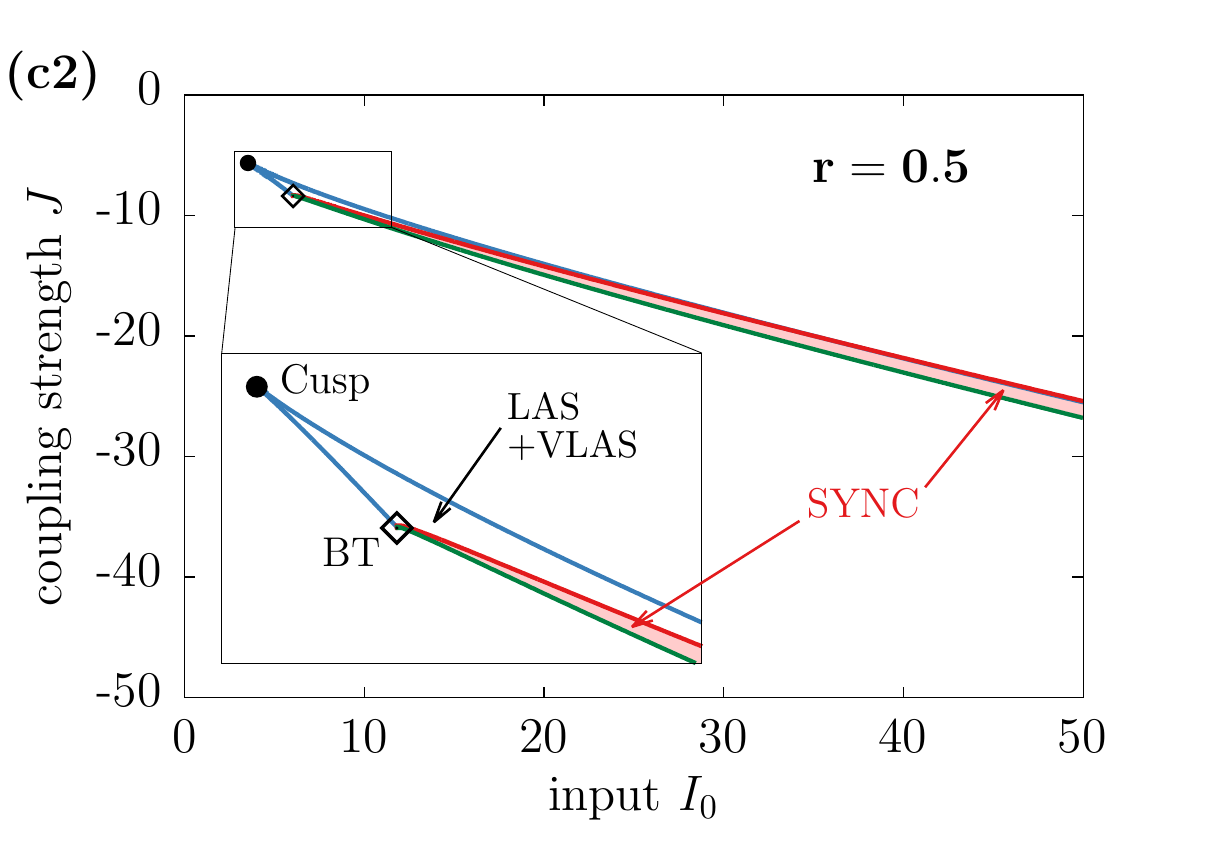}\hspace{-0.5cm}
  \includegraphics[height=4.5cm]{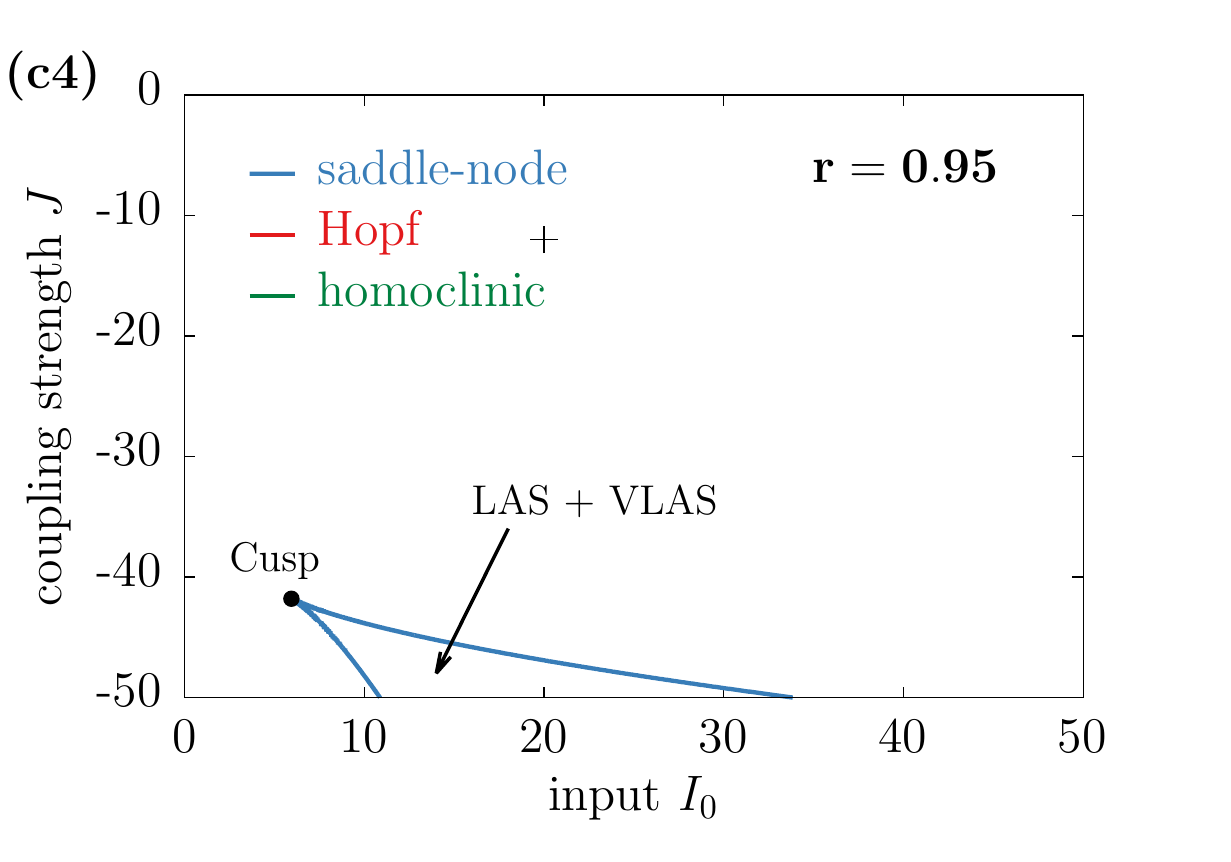}
  \caption{Phase diagrams of the RV dynamics \eqref{eq:FRE} for instantaneous RP pulses. 
  (a) Comparison of pulse profiles: RP pulses of various width $r \leq 1$ (color coded) are symmetric about the peak phase $\theta=\pi$ corresponding to the instant where the voltage $v$ diverges; see bottom panels for the link between $v$ and the theta-phase $\theta = 2\arctan(v)$.
  The larger $r\nearrow1$, the narrower the pulse (note the different $x$-axis-scale in the right panel).
  AS pulses~\eqref{eq:AS_pulse} with different width parameters $n\in\mathbb N$ are shown in grey for comparison.
  The RP pulse with $r=0$ coincides with $p_{\text{AS},1}$.
  (b) Bistability between low- and high-activity, asynchronous states (L/HAS) in excitatory networks ($J>0$). Saddle-node bifurcation boundaries are color-coded according to the pulse width $r$ of the RP pulses shown in (a).
  (c1--4) Collective oscillations in narrow parameter regions (red, SYNC) for inhibitory networks ($J<0$) and decreasing pulse width $r=0, 0.5, 0.9, 0.95$.
  Insets in (c1--2) zoom in on the general bifurcation structure near the Bogdanov-Takens (BT) point, where saddle-node (blue), (supercritical) Hopf (red) and homoclinic (green) bifurcation curves meet.
  The $+$ in (c4) denotes the parameters used for numerical simulations in \cref{fig:sim}(a).}
  \label{fig:RP_pulses_all}
\end{figure}

For inhibitory networks ($J<0$) with positive input currents ($I_0 > 0$), see \cref{fig:RP_pulses_all}(c), the asynchronous steady state loses stability at a supercritical Hopf bifurcation (red curve), \change{where collective oscillations arise with small amplitude. With increasing $|J|$, i.e.~stronger inhibition, the amplitude of the collective oscillations grows, whereas their frequency decreases up to the critical point at the green curve, where the frequency becomes zero and collective oscillations are destroyed; this critical point is called a homoclinic bifurcation. The red-shaded region between the red and the green curve, where collective oscillations exist and neurons spike in synchrony (that is why this region is labelled SYNC), occurs for limited values of inhibition only.}
Similar results were obtained by~\cite{juettner_martens_2021} for the case $r=0$ and by~\cite{luke_barreto_so_2013,lin_barreto_so_2020} for wide AS pulses ($n=2,9$).
There is also a region of bistability between a LAS and a very low-activity state (VLAS), bounded by the (blue) saddle-node curve, that meets the Hopf and homoclinic curves at the Bogdanov-Takens (BT) point, which is another codimension-2 bifurcation.
For larger $r$, i.e.\ for narrower pulses, this BT point moves further down so that collective oscillations require strong inhibitory coupling $J \ll 0$.
On top of that, collective oscillations remain confined to a narrow region in parameter space, which becomes almost negligible for RP pulses that are reasonably localized with $r > 0.8$, see Fig.~\ref{fig:RP_pulses_all}(c3--4).

\subsection{Right-skewed pulses $(\varphi>0~\text{\normalfont and/or}~\psi>\pi)$}
\label{subsec:3C}
Shifting ($\psi \neq \pi$) and/or skewing ($\varphi \neq 0$) the symmetric RP pulses critically affects the collective dynamics of globally coupled QIF neurons.
By introducing a virtual threshold $v_\text{thr}=\tan(\psi/2)$ at which the emitted pulse is strongest, one can shift the pulse to the right ($\psi>\pi$) or to the left ($\psi<\pi$) of the QIF threshold $v_p$ at infinity (\cref{fig:KJ_pulses_all}a, blue curve);
note that $v_\text{thr} \to v_p=\infty$ for $\psi\to\pi$.
By introducing an asymmetry parameter $\varphi\neq0$, one can skew the pulse (\cref{fig:KJ_pulses_all}a, red and green curves).
\changeB{Asymmetric pulses $p_{r,\varphi,\psi}(\theta)$ do not take on their maximum at $\theta = \psi$ when $v$ diverges and the neuron spikes, but their peak is shifted to the right ($\varphi>0$) or to the left ($\varphi<0$) of the spike. 
Though the peak shift is rather small\footnote{
An asymmetric pulse $p_{r,\varphi,\psi}$ with $\varphi\neq 0$ and $\psi=2\arctan(v_\text{thr})$ has its peak when the presynaptic neuron's voltage $v$ reaches $v_{max} = \frac{(1+r) v_\text{thr} \cos(\varphi/2)+(1-r)\sin(\varphi/2)}{(1+r) \cos(\varphi/2)-(1-r) v_\text{thr} \sin(\varphi/2)}$, and the pulse vanishes at $v_{min} = \frac{(1+r)v_\text{thr} \sin(\varphi/2)-(1-r)\cos(\varphi/2)}{(1+r) \sin(\varphi/2)+(1-r) v_\text{thr} \cos(\varphi/2)}$. 
The corresponding peak and trough phases $\theta$ can be found via $\theta_{max/min}=2\arctan(v_{max/min})$.},
the mean of the pulse moves clearly away from the threshold phase $\psi=2\arctan(v_\text{thr})$}, see the vertical red-dashed line in Fig.~\ref{fig:KJ_pulses_all}(a).

\begin{figure*}[!t]
  \centering
  \hspace{0.3cm}
  \includegraphics[height=4.8cm]{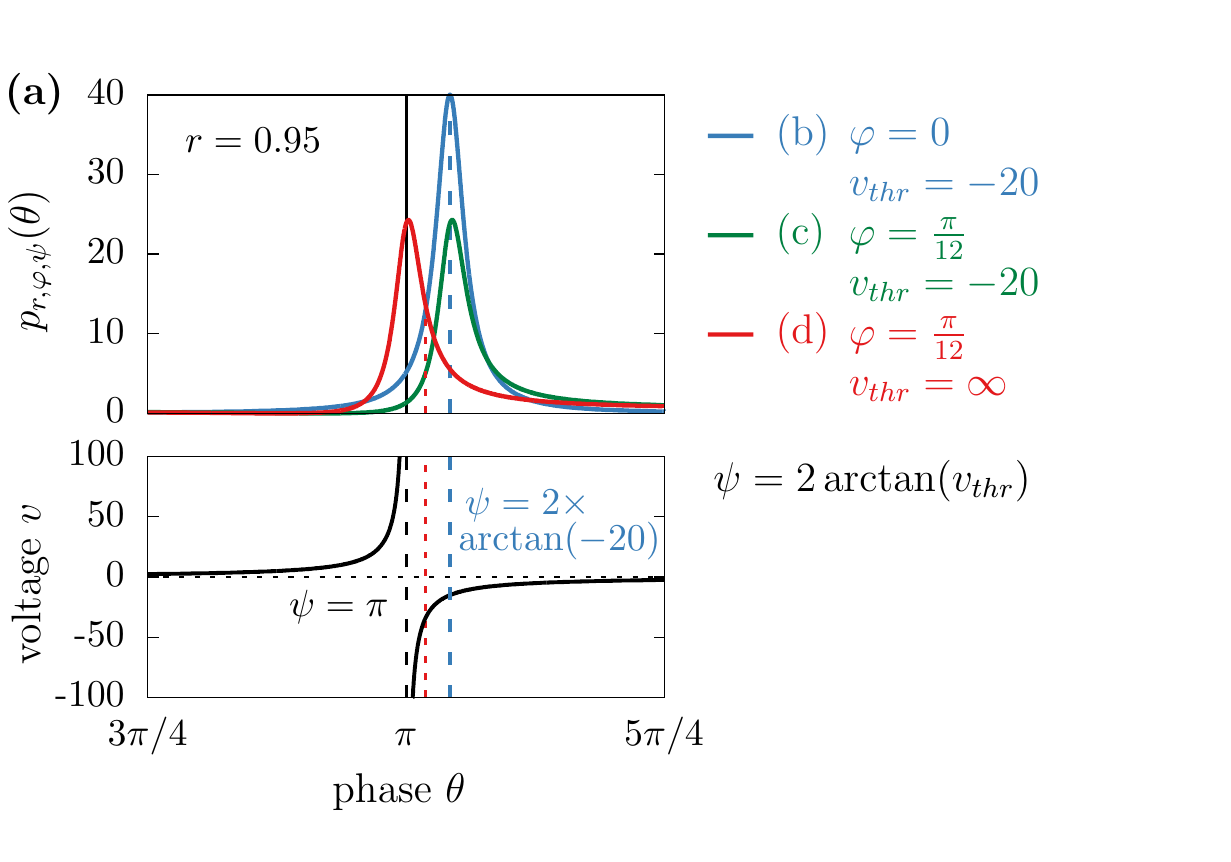}\hspace{-0.9cm}
  \includegraphics[height=4.8cm]{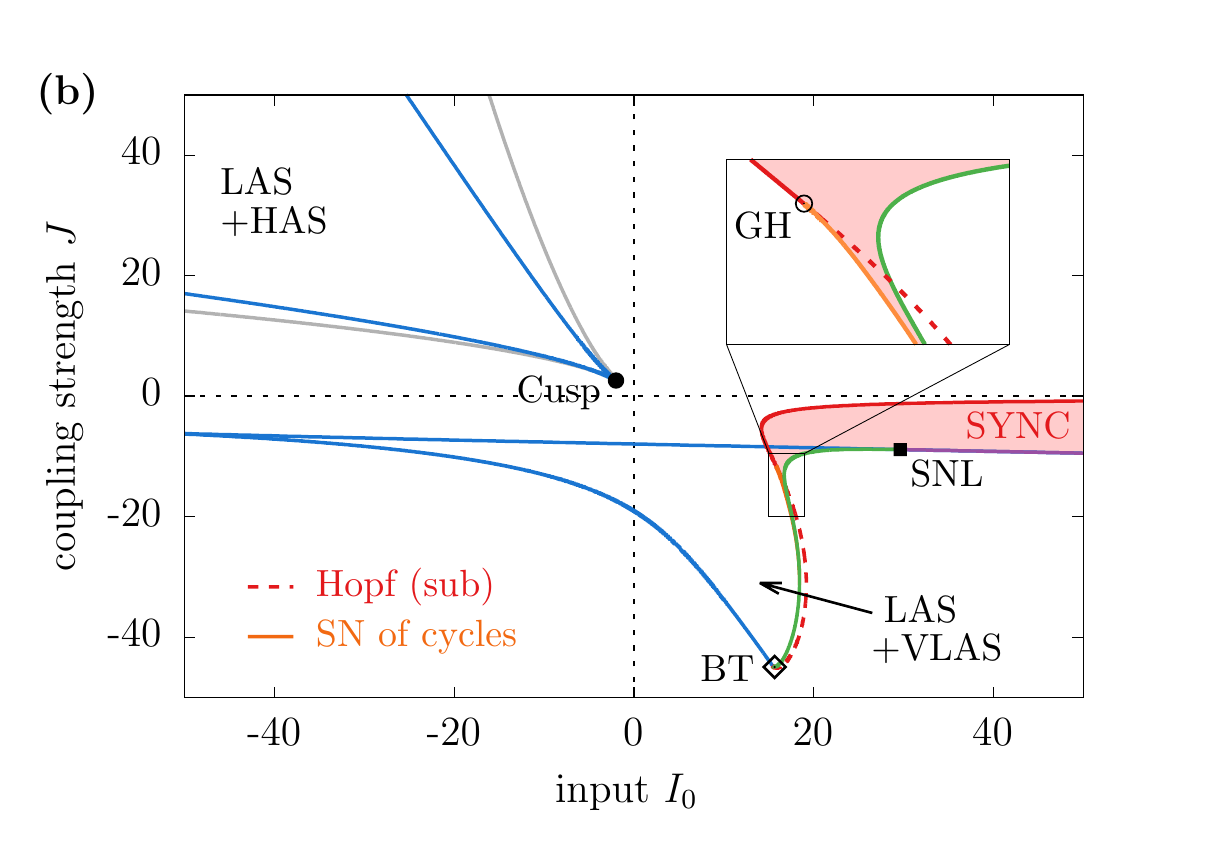}\hspace{0cm}\\[-0.5cm]
  \includegraphics[height=4.8cm]{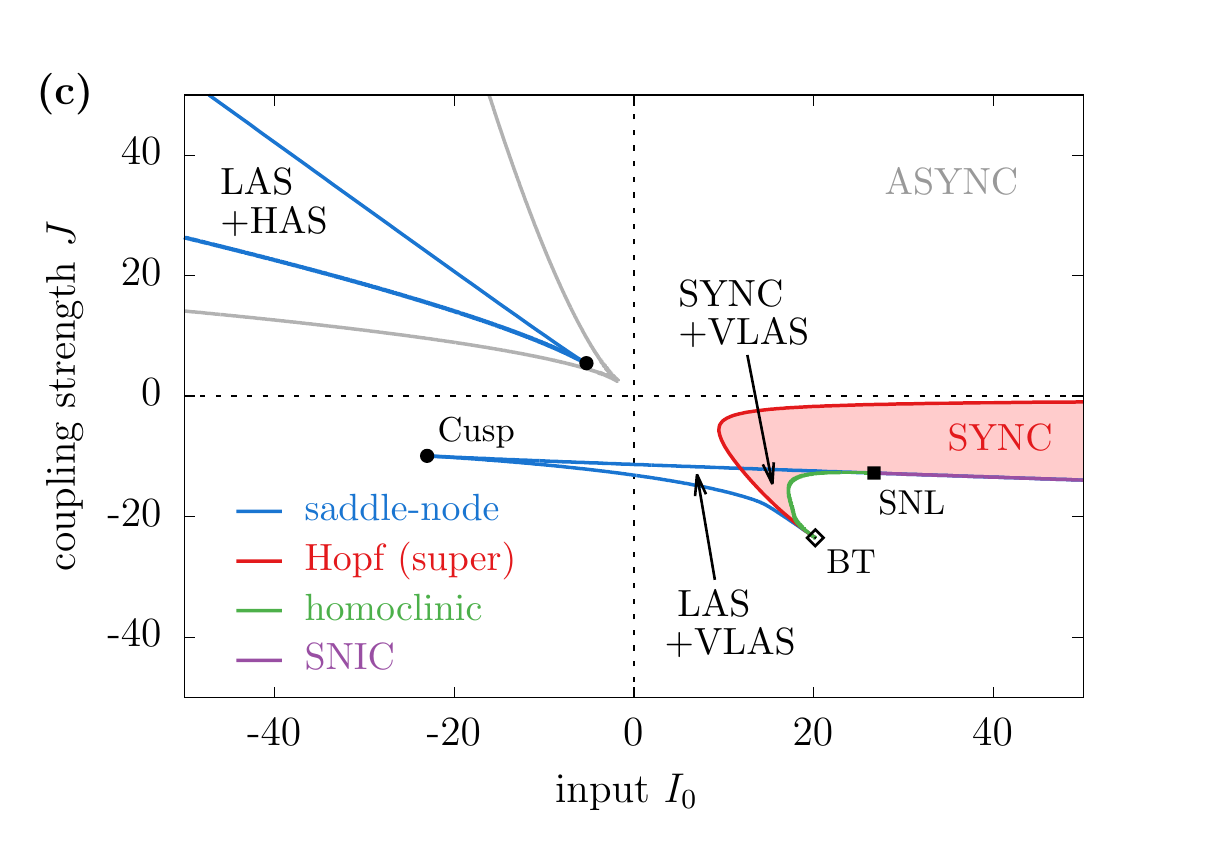}\hspace{-0.5cm}
  \includegraphics[height=4.8cm]{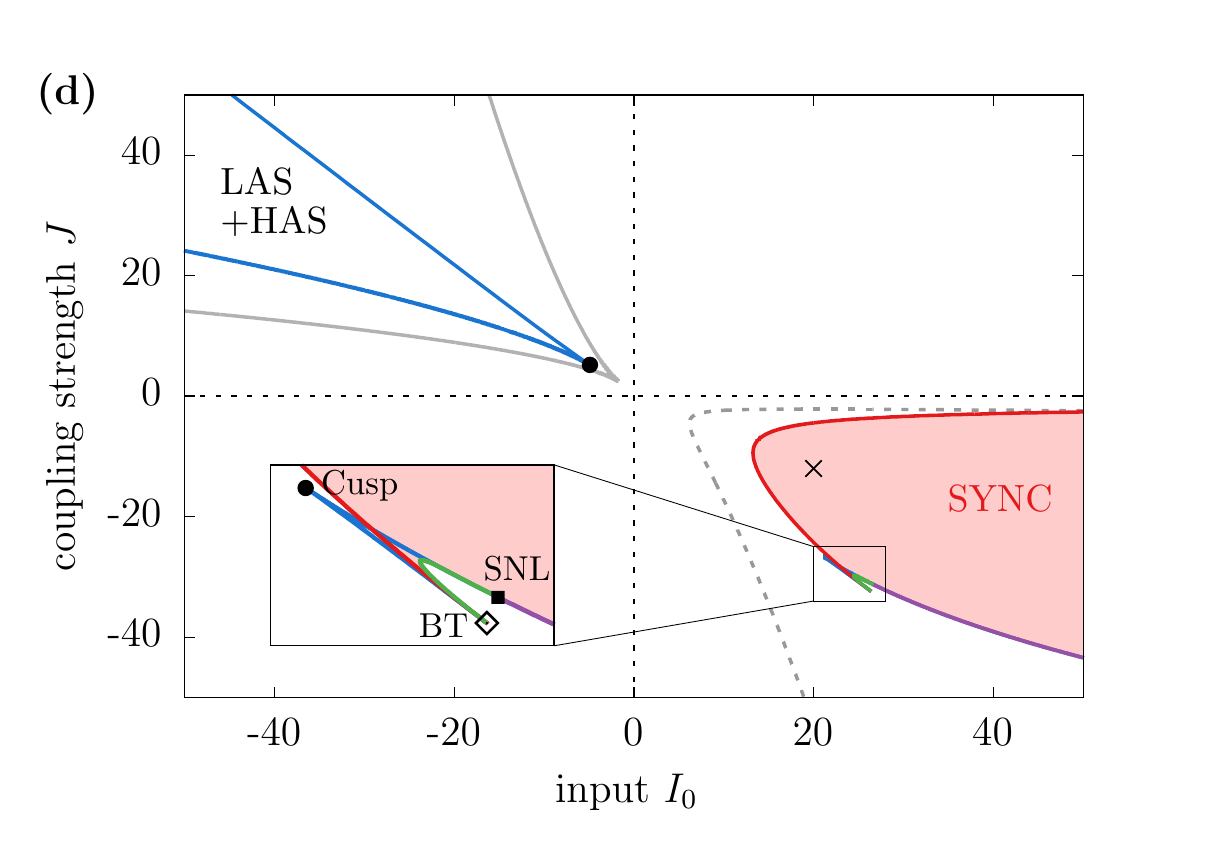}
  \caption{Phase diagrams of the RV dynamics \eqref{eq:FRE} for right-skewed pulses.
  (a) Comparison of pulse profiles that are shifted (blue; $\varphi=0, \psi=-2\arctan20$), asymmetric (red; $\varphi =\pi/12, \psi=\pi$) or both (green; $\varphi=\pi/2, \psi=-2\arctan20$); pulse width is $r=0.95$.
  The mean of the shifted pulse coincides with its peak phase $\theta=\psi$ (blue dashed), whereas the mean of the asymmetric pulse (red dotted) is shifted to the right of its peak.
  (b--d) Phase diagrams of the collective dynamics for pulses shown in (a). Collective oscillations are found in the red shaded regions (``SYNC'') for inhibitory coupling ($J<0$) and excitatory drive ($I_0>0$). Asymmetric pulses yield larger SYNC regions, while bistability regions (between low-activity states (LAS) or collective oscillations and very low-activity states (VLAS)) are drastically reduced. 
  Boundaries for codimension 1-bifurcations are colored lines: saddle-node (blue), supercritical Hopf (red solid), subcritical Hopf (red dashed), homoclinic (green), SNIC (violet), saddle-node of cycles (orange).
  Codimension 2-bifurcation points are marked by symbols: Cusp (circle), Bogdanov-Takens (BT, diamond), Saddle-Node Separatrix Loop (SNL, square).
  In (d) the $\times$-symbol denotes the parameter values for numerical simulations in \cref{fig:sim}(b)
  and the gray dashed curve is the supercritical Hopf boundary for the RV dynamics~\eqref{eq:FRE_syn} with first-order synaptic kinetics, $\tau_m=2\tau_d=10$ms, $\tau_r=0$ms, in the limit of $\delta$-spikes $(r,\varphi,\psi)\to(1,0,\pi)$.
  }
  \label{fig:KJ_pulses_all}
\end{figure*}

As before, one can perform linear stability analysis of the RV dynamics~\eqref{eq:FRE} now with mean presynaptic input $P_{r,\varphi,\psi}$ given by Eq.~\eqref{eq:p_mean} for general pulses \eqref{eq:KJ_theta}. 
As motivated in \cref{sec:3}, I will focus on pulses $p_{r,\varphi,\psi}$ whose mean is shifted to the right of the spike, that is, $\varphi > 0$ and/or $\psi> \pi$.
In \cref{fig:KJ_pulses_all}(b), I shift the narrow RP pulse ($r=0.95$, red in \cref{fig:RP_pulses_all}a) such that it reaches its peak when the presynaptic voltage crosses the virtual threshold $v_\text{thr}=-20$ when recovering from the reset after its spike. 
\changeB{This shift immediately affects the SYNC region of collective oscillations:
While the SYNC region literally falls out of view for narrow RP pulses (with $v_\text{thr}=\infty$), for shifted RP pulses the SYNC region moves closer to $J=0$ and grows as the homoclinic bifurcation curve (green) is pushed away from the Hopf curve (red).
Shifting the pulse thus facilitates neuronal synchrony.
Although the bifurcation structure appears more involved, the effect on the collective dynamics is small. The following technical details can hence be skipped, I report them here only for completeness.}
The cusp-shaped region of bistability between the LAS and the VLAS grows into the $(I_0<0)$-region---albeit with negligible effects on the overall dynamics because the VLAS has a diminutive basin of attraction (not shown)---and the upper saddle-node curve (blue) coincides with the homoclinic curve beyond the Saddle-Node-Separatrix-Loop (SNL) bifurcation point.
The violet curve to the right of the SNL point denotes a Saddle-Node on Invariant Circle (SNIC) bifurcation that terminates the collective oscillations.
Furthermore, the supercritical Hopf bifurcation (red solid) becomes subcritical (red dashed) at a Generalized Hopf (GH) bifurcation point, from which a saddle-node of cycles (SN of cycles) bifurcation curve (orange) emanates, see inset in \cref{fig:KJ_pulses_all}(b).
In the red region bounded by the supercritical Hopf/SN of cycles curve and the homoclinic curve, there is bistability between the VLAS and SYNC;
that is, the LAS loses stability and gives rise to collective oscillations at the Hopf bifurcation.
For excitatory coupling ($J>0$), on the other hand, shifting the RP pulse hardly affects the cusp-shaped region of bistability between the LAS and the HAS.

Introducing small right-skewed asymmetry ($\varphi > 0$) \change{allows for an even broader parameter region of collective oscillations, while simplifying the bifurcation scenario to great extent}.
In \cref{fig:KJ_pulses_all}(c), I consider pulses that are both shifted ($v_\text{thr}=-20$) and skewed ($\varphi = \pi/12$), see the green pulse in \cref{fig:KJ_pulses_all}(a).
The Hopf bifurcation is now always supercritical and the bistability region between LAS and VLAS shrinks drastically, whereas the bistability region between the VLAS and SYNC is only slightly increased. 
Compared to the case $\varphi=0$, also the bistability region between LAS and HAS shrinks for $J>0$.

\changeB{In \cref{fig:KJ_pulses_all}(d) I consider only skewed, but not shifted, pulses $p_{0.95,\pi/12,\pi}$ to disentangle the effect of the asymmetry from that of the shift parameter.}
For large $r=0.95$ and small $|\varphi| \ll 1$, an individual pulse (red curve in \cref{fig:KJ_pulses_all}a) reaches its maximum for $v_{max} = (r+1)/(r-1)/\tan(\varphi/2)$ almost immediately after the presynaptic spike at $\theta=\pi$. 
However, the asymmetry parameter $\varphi$ shifts the bulk of the pulse, and hence its mean, further away from $\pi$, which has a non-trivial effect on the collective dynamics:
one may have expected an overall picture of the collective dynamics to be somewhere in between that of non-shifted and shifted RP pulses; 
yet, the skewed pulse dramatically enhances the SYNC regime of collective oscillations, that dominates the parameter region for $J<0$ and $I_0>0$.
Moreover, the bistability regions within the triangle of Cusp, SNL and BT points become almost negligible.
\changeB{This scenario strongly resembles the case for synaptic dynamics with a dominant SYNC region bounded by a supercritical Hopf bifurcation~\citep{devalle2017,clusella_ruffini_2022}. 
In fact, if one considers first-order synaptic dynamics $\tau_d \dot S = - S + P_{r,\varphi,\psi}$ with synaptic decay time constant $\tau_d>0$, collective oscillations emerge already for $\delta$-spike-interactions, $(r,\varphi,\psi)\to (1,0,\pi)$. 
Consistent with the literature~\citep{devalle2017}, collective oscillations always emerge via a supercritical Hopf bifurcation (\cref{fig:KJ_pulses_all}d, dashed grey line).}
Noteworthy, the parameter regions of collective oscillations almost coincide for instantaneous asymmetric pulses and for synaptic kinetics triggered by $\delta$-spikes.


\subsection{Left-skewed pulses $(\varphi<0~\text{\normalfont and/or}~\psi<\pi)$}
\label{subsec:3D}
For biological plausibility, I previously considered pulses whose mean coincided with the QIF threshold at infinity or was shifted to the neuron's refractory phase after the spike (i.e.\ the mean of the pulse had a phase $\theta_\text{mean}\geq\pi$).
For the sake of completeness, I briefly report on pulses whose mean is advanced to the phase before the spike.
The red curve in \cref{fig:left-pulses}(a) describes a left-skewed pulse that I obtained from mirroring the asymmetric red pulse in \cref{fig:KJ_pulses_all}(a) by setting $\varphi = \pi/12 \mapsto -\pi/12$ while keeping $r=0.95$ and $\psi=\pi$, i.e.\ $v_\mathrm{thr}=\infty$, as before.
The mean is now shifted to the left (the dashed line indicates the spike at $\theta=\psi$), which changes the phase diagram drastically, compare~\cref{fig:left-pulses}(b) with \cref{fig:KJ_pulses_all}(d).
\changeB{First, the synchronous state of collective oscillations (SYNC) does not exist for recurrent inhibition, but requires excitation ($J>0$).}
Then, the cusp-shape bifurcation region that dominated the upper left part ($J>0,I_0<0$) in \cref{fig:RP_pulses_all,fig:KJ_pulses_all} has shrunk to a small triangular region bounded by the three codimension-2 bifurcation points (BT, SNL, and Cusp), see inset in \cref{fig:left-pulses}(b).
The Hopf bifurcation line emanating from the BT point separates the triangular bistability region into a lower part, where two asynchronous states (a low and a high activity state) coexist, and an upper part, where the high-activity asynchronous state (HAS) has lost stability and given rise to collective oscillations.
Overall, the phase diagram strongly resembles the case of gap junction-coupling in interplay with $\delta$-spikes~\citep[Fig.~7 of][]{Pietras_et_al_2019}.

What is more, the shape of left-skewed pulses has hardly any impact on the phase diagram.
For advanced Dirac $\delta$-pulses $p_{\delta,\psi}(\theta)$, \cref{eq:delta_pulse} with $\psi=2\arctan(v_\text{thr}) < \pi$, see the dark gray curve in \cref{fig:left-pulses}(a) with virtual threshold $v_\text{thr}=20$, the phase diagram is only slightly shifted downwards and the triangular bistability region has become a bit smaller (dark gray lines in \cref{fig:left-pulses}b).

\begin{figure}[!t]
  \centering
  \hspace*{0.0cm}\includegraphics[height=5.2cm]{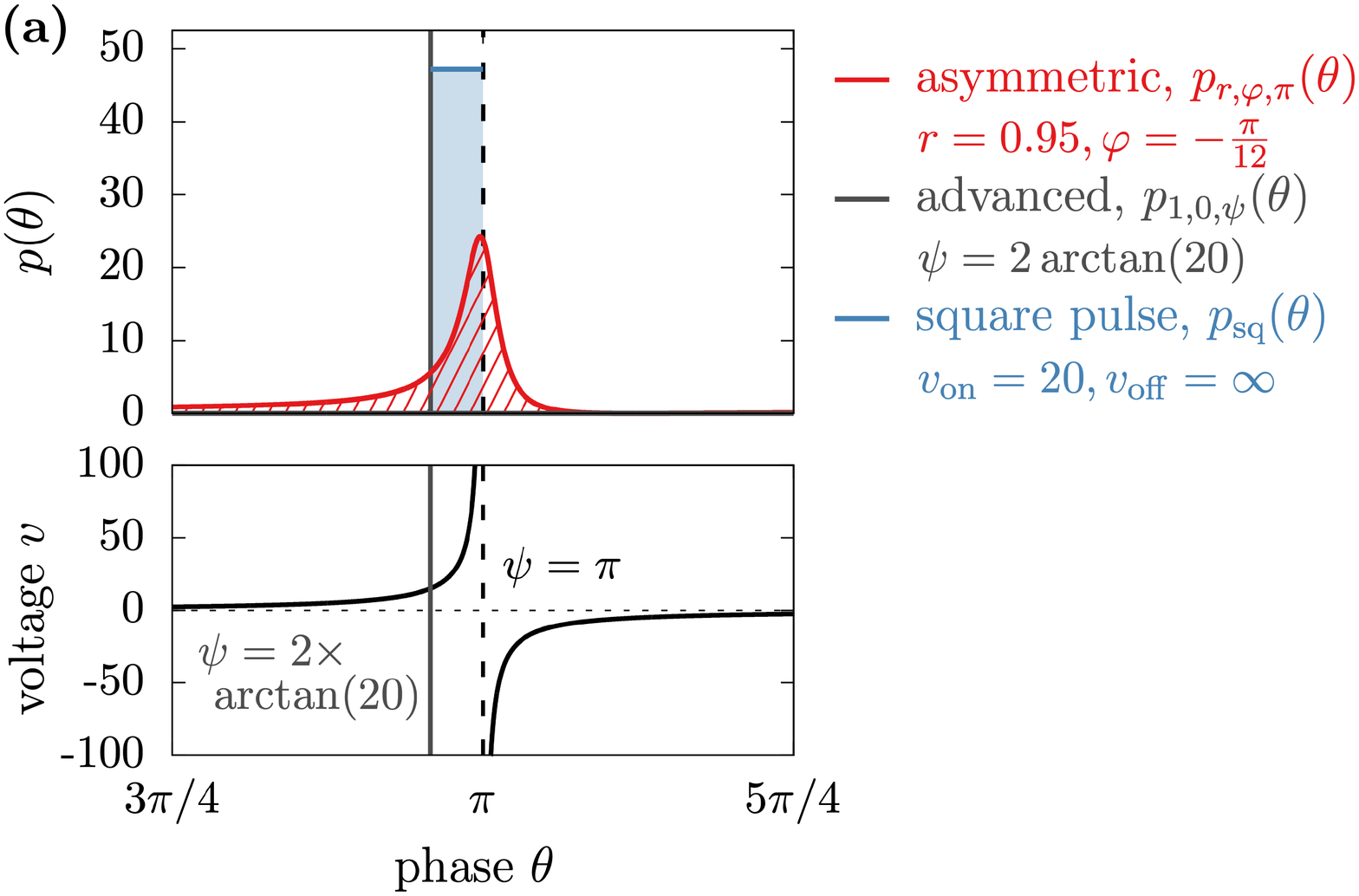} 
  \hspace{-0.3cm}\includegraphics[height=5.2cm]{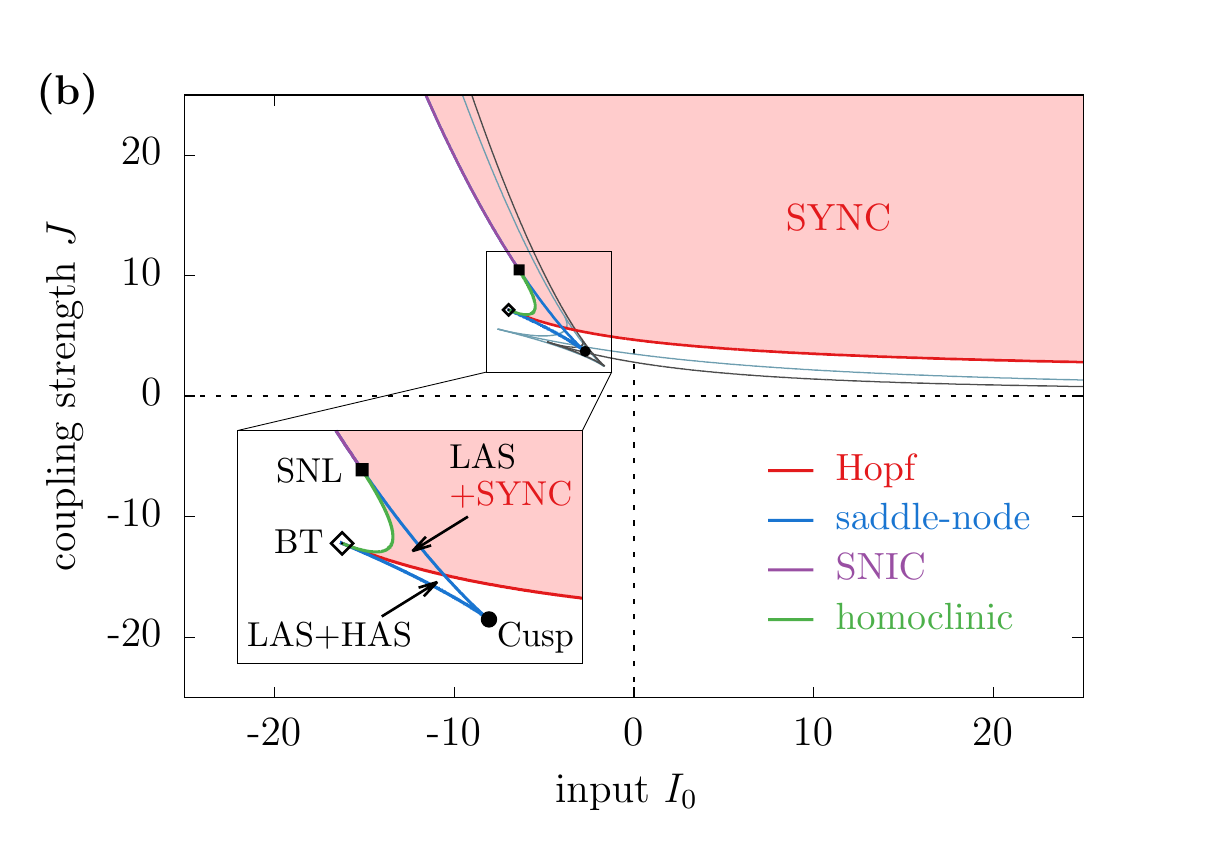}
  \caption{Left-skewed pulses induce collective oscillations for excitatory, but not for inhibitory neurons. (a) Asymmetric (red), advanced (dark-gray) and square pulses (gray-blue), whose mean lies to the left of the QIF threshold at $\theta=\pi$, where the voltage $v$ diverges (bottom).
  (b) The phase diagram using the asymmetric pulse $p_{r,\varphi,\pi}$ exhibits a dominant region for collective oscillation (SYNC, red shaded), which is bounded by a supercritical Hopf bifurcation curve (red) from below and by a SNIC (violet) and homoclinic (green) bifurcation from the left. In the triangular region between the Bogdanov-Takens (BT), Saddle-Node Separatrix Loop (SNL), and Cusp points, there is bistability between a high activity (HAS) and a low activity asynchronous state (LAS) below, and between the LAS and SYNC in the red shaded region above the Hopf curve, see the inset for a zoom.
  The dark-gray and gray-blue lines correspond to the phase diagram structure of the advanced and the square pulses depicted in (a).
  }
  \label{fig:left-pulses}
\end{figure}

For comparison, I also consider square pulses $p_\text{sq}(v)$ similar to those considered by~\cite{ratas_pyragas_2016},
\begin{equation}
p_\text{sq}(v) = \begin{cases} \pi / \vartheta, \quad &\text{if } v \geq v_\text{on} \text{ or } v \leq v_\text{off}, \\
0, &\text{otherwise,}\end{cases}
\label{eq:psq}
\end{equation}
where $\vartheta:=(\theta_\text{off}-\theta_\text{on})/2$ with onset and offset phases $\theta_\text{off,on}=2\arctan(v_\text{off,on})$ is chosen such that the square pulse $p_\text{sq}$ in the corresponding $\theta$-phase description is normalized to $2\pi$ (cf.~\cref{subsec:accessible}).
For $v_\text{off}\to-\infty$ and $v_\text{on} > 0$, \cref{eq:psq} coincides with the pulses used by~\cite{ratas_pyragas_2016}, but rescaled by $\pi/\vartheta$; see the blue pulse in \cref{fig:left-pulses}(a) with $v_\text{on} = 20$ for an example.
For $v_\text{on} > 0$ and $v_\text{off}<0$, the square pulse $p_\text{sq}$ does not terminate at the spike, but actually encloses it and extends its impact to the neuron's relative refractory period.

To obtain the mean pulse activity $P_\text{sq} = \langle p_\text{sq}\rangle$, it is no longer possible to follow the approach of \cref{subsec:recurrent} because the pulses $p_\text{sq}(v)$ \change{do not possess the desired mathematical properties, in particular they do not permit an analytic continuation in the complex plane.} 
Nonetheless, it is possible to obtain the RV dynamics for global square-pulse coupling on the Lorentzian manifold. 
To this end, one can average the square pulses~\eqref{eq:psq} with respect to the Cauchy-Lorentz distribution density $\mathcal W(v,t) $, \cref{eq:CL_density}, which results in the mean presynaptic pulse activity
\begin{equation}
\begin{aligned}
P_\text{sq}(R,V) &= \int_{-\infty}^\infty p_\text{sq}(v) \mathcal W(v,t) dv\\
&= \frac1\vartheta \left[ \arctan\!\left( \!\frac{V - v_\text{on}}{\pi \tau_m R}\!\right) \!-\! \arctan\!\left( \!\frac{V - v_\text{off}}{\pi \tau_m R}\!\right)\! \right] .
\label{eq:psq_mean}
\end{aligned}
\end{equation}
For $v_\text{off}\to-\infty$, \cref{eq:psq_mean} reduces to the global synaptic variable used by~\cite{ratas_pyragas_2016} rescaled with $\pi/\vartheta$.

\changeB{\cref{eq:psq_mean} can readily be used for the recurrent input $I_\text{syn} = J P_\text{sq}$ in the RV dynamics \eqref{eq:FRE}.
Linear stability analysis for square pulses with $v_\text{on}=20$ and $v_\text{off}=\infty$ (blue pulse in \cref{fig:left-pulses}a) 
yields a very similar phase diagram compared to the asymmetric and advanced pulses, see the gray-blue lines in \cref{fig:left-pulses}(b) and cf.~\citep{ratas_pyragas_2016}:
%
A large parameter region of collective oscillations occurs in the upper right quadrant for excitatory coupling $J>0$, which could have been anticipated along the same lines as in \cref{subsec:5A} because $\partial_V P_\text{sq}(R^*,V^*) = \pi\tau_m R^* / \vartheta / [(\pi\tau_m R^*)^2 + (V^* - v_\text{on})^2]>0$ for all fixed-point solutions $R^* > 0$.}

In sum, left-skewed pulses---which may be biologically less plausible because their mean \changeB{effect} occurs before the neuron actually spikes---, first, do not allow for collective oscillations among inhibitory QIF neurons, but only among excitatory neurons, and, second, the resulting phase diagram is largely insensitive to the shape of the pulse, which is in stark contrast to symmetric RP pulses (\cref{fig:RP_pulses_all}) and to right-skewed pulses (\cref{fig:KJ_pulses_all}).



\section{Discussion}
\label{sec:6}
The motivation for the work at hand has been to include and justify another degree of biological realism---namely, smooth pulsatile synaptic transmission---within the framework of an exact and low-dimensional mean-field theory for networks of globally coupled spiking neurons~\citep{montbrio_pazo_roxin_2015,pcp_arxiv_2022}.
\changeB{The resulting RV dynamics \eqref{eq:FRE} are admissible to mathematical analysis and allow for comprehensive insights how the pulse shape affects the collective dynamics.} 
Unless considering synaptic interaction via $\delta$-spikes, previous studies on pulse-coupled spiking $\theta$-neuron networks typically resorted to broad, symmetric, and rather inflexible Ariaratnam-Strogatz (AS) pulses, \cref{eq:AS_pulse}, more for mathematical convenience than for biological realism, see \cref{sec:2}.
The main reason for using AS pulses seems a historic one. To provide some context, I will review previous approaches to pulse-coupling between $\theta$-neurons in \cref{subsec:previous}.
Here, I also revisit an alternative approach to pulse-coupling between QIF neurons, which however has led to biologically contradictory results on the single neuron as well as on the network level, and which is restricted to the Lorentzian manifold.

To overcome previous limitations and to analyze the effect of biologically plausible pulse shapes---both symmetric and asymmetric---on the collective dynamics of globally pulse-coupled QIF and $\theta$-neurons (even beyond the Lorentzian manifold),
I have proposed a rather general family of pulse functions $p_{r,\varphi,\psi}(\theta)$, \cref{eq:KJ_theta},
that depend on the presynaptic voltage $v$ via their $\theta$-phase description $\theta(t)=2\arctan(v(t))$.
In the limit $(r,\varphi,\psi)\to(1,0,\pi)$, the pulse reduces to the $\delta$-spike commonly used in network models of spiking neurons.
Taking the mean over those $\delta$-spikes yields a direct connection with the population firing rate, $\langle p_{1,0,\pi}\rangle = \pi \tau_m R$, see \cref{eq:P2R}.
Tuning either of the three pulse parameters---its width $r$, asymmetry $\varphi$ and shift $\psi$---allows for diverse pulse shapes $p_{r,\varphi,\psi}$ that may mimic realistic synaptic transmission between pre- and postsynaptic neurons and, by adjusting the timing of the synapse, crucially influence the collective dynamics of the network.
\change{I have concentrated on non-negative pulses \eqref{eq:KJ_theta} only and could therefore reduce the number of parameters of the analytically tractable Kato-Jones pulses \eqref{eq:KJ_distribution} from four to three.
In \cref{subsec:D1}, I discuss these three pulse parameters in more detail and also draw connections to delay and to electrical coupling via gap junctions.
When loosening the modeling assumptions, such as the closeness to a SNIC bifurcation in \cref{thm:1}, the more general Kato-Jones pulses \eqref{eq:KJ_distribution} that change signs can also be thought possible as an extension of \cref{eq:KJ_theta}---studying their effects on the collective dynamics will be left for future work. 
In view of more general synaptic processes, I discuss in \cref{subsec:D2} whether instantaneous pulses of finite width can replace more complex synaptic transmission, such as synaptic kinetics and conductance-based synapses, beyond the interpretation inherent to \cref{thm:1}.}
Recall that \cref{thm:1} imposed certain constraints on the pulse width in order to establish the $\theta$-neuron network~\eqref{eq:thm3}---and its QIF-equivalent \cref{eq:qif1}---as the canonical pulse-coupled network model for weakly connected Class 1 excitable neurons.
For the discussion in \cref{subsec:D2}, I will discard those restrictions on the pulse function $p$. Instead, I will regard the $\theta$-model as the continuous analogue of the (discontinuous) QIF model, cf.~\cref{remark2}.
In this way, the pulses $p_{r,\varphi,\psi}(\theta)$ in the $\theta$-neuron network~\eqref{eq:thm3} can be probed against the hypothesis that they may, or may not, correspond to complex synaptic interactions in the QIF network~\eqref{eq:QIF}.


\subsection{Previous approaches to pulse-coupling}
\label{subsec:previous}
The phase-representation of the $\theta$-neuron suggests to invoke the theory of coupled phase oscillators to study synchronization and emergent collective behavior. 
Yet, insights from the literature on pulse-coupling in related phase models may not always carry over to $\theta$-neurons and one ought to be careful when drawing analogies. 
Revisiting the arguments that justified pulses of finite width between coupled phase oscillators will lead to the root of the ``pulse-problem'' in $\theta$-neurons.
To this end, I consider a network of $N$ globally coupled neurons described by variables $\theta_j\in\mathbb S^1$, $j=1,\dots,N$, with the dynamics
\begin{equation}
\dot \theta_j = \omega_j - b_j \cos(\theta_j) + \frac{1}{N} \sum_{k=1}^N q(\theta_j) p(\theta_k)  .
\label{eq:phasemodel1}
\end{equation}
The network dynamics \eqref{eq:phasemodel1} reduces to the network model \eqref{eq:thm3} of $\theta$-neurons for $\omega_j=1+\eta_j$, $b_j=1-\eta_j$ and $q(\theta)=1+\cos(\theta)$ with excitability parameter $\eta_j$.
More generally, \cref{eq:phasemodel1} describes the dynamics of active rotators~\citep{Shinomoto_Kuramoto_1986,Kuramoto1987} with pulse-response coupling; $p(\theta_k)$ is a pulse-like function and $q(\theta_j)$ plays a role analogous to that of a ``phase response curve''~\citep{winfree1980geometry}.
The parameter $b_j$ determines whether neuron $j$ is periodically spiking ($|\omega_j| > b_j$) or in an excitable regime ($|\omega_j|\leq b_j$).
Setting $b_j=0$ for all $j$ leads to the seminal Winfree model~\citep{winfree1967,winfree1980geometry}. 


In the context of neural oscillators, \cite{ermentrout_kopell_1990} used the Winfree model---\cref{eq:phasemodel1} with $b_j=0$---to describe the dynamics of periodically firing neurons (with frequency $\omega_j$) on their limit cycle, parameterized by the phase $\theta_j$.
In the course of an action potential, the presynaptic neuron $k$ emits a pulse $p(\theta_k)$ that is modulated by the response function $q(\theta_j)$ of postsynaptic neuron $j$.
For Hodgkin-Huxley-like neurons, the pulses $p(\theta)$ can be interpreted as instantaneous nonlinear conductances~\citep{ermentrout_kopell_1990,Wang_Rinzel_1992}.
By using graded potentials rather than fast spikes, synaptic interaction in these kinetics-based models is represented as a function of the presynaptic voltage.
Analogously, in kinetic models of synaptic transmission, the neurotransmitter concentration in the synaptic cleft can be reduced to a sigmoidal function of the presynaptic voltage~\citep{Destexhe1994}, see also~\citep{Koch_2004,Ermentrout_Terman_2010,Rothman2013}.
According to these theories, a voltage-dependent pulse can comprise the biochemical processes---at the presynaptic site---from the initiation of an action potential until the release of neurotransmitters to the synaptic cleft.
The resulting pulse can become arbitrarily broad and typically exhibits a steep increase until the neuron spikes and a moderate decrease afterwards~\citep{ermentrout_kopell_1990}.
Its asymmetric shape results from the interplay between the action potential shape and the synaptic activation, or neurotransmitter release, function.
Asymmetric pulses were shown to be critical to 
synchronization in neural networks consisting of two neurons~\citep{Wang_Rinzel_1992,skinner1994mechanisms,sato2007generalization}; a comprehensive picture how asymmetric pulses affect the collective behavior of large neural networks, however, has been lacking.

The argumentation that pulses $p(\theta)$ in the Winfree model result from the interplay between action potential shape and synaptic activation function, 
may justify the use of broad pulses also in networks of $\theta$-neurons.
This reasoning ignores, however, that the trajectory of an action potential in conductance-based neurons is radically compressed 
in the canonical $\theta$-model (\cref{fig:cartoon}) so the resulting pulse cannot become arbitrarily broad (\cref{fig:cartoon_pulse}).
Besides, the Winfree model describes the phase dynamics of periodically firing conductance-based neurons obtained through a proper phase reduction~\citep{ermentrout_kopell_1990},
whereas the $\theta$-neuron (albeit a phase model) does not result from phase reduction.
Instead, it canonically describes a particular dynamical regime of a neuron that does not need to be periodically spiking, but can also be excitable.
Hence, adopting the Winfree model, \cref{eq:phasemodel1} with $b_j=0$, with broad and possibly asymmetric pulses $p(\theta)$ for networks of $\theta$-neurons has to be regarded with care.
Still, as long as $\theta$-neurons are in the periodically firing regime, the Winfree model may be useful for studying their collective dynamics, at least to some extent~\citep{Hoppensteadt_Izhikevich_1997,Izhikevich_1999,Clusella_Pietras_Montbrio_2022}, 
and previous results about the effect of the response and pulse functions $q(\theta_j)$ and $p(\theta_k)$ on synchronization properties in the Winfree model can become applicable~\citep{ariaratnam2001phase,pazo_montbrio_2014,gallego_et_al_2017}.

To account for excitable neuronal dynamics, the Winfree model has to be augmented by a term proportional to $\cos(\theta_j)$, yielding the network model \eqref{eq:phasemodel1} of coupled active rotators.
A thorough analysis of \cref{eq:phasemodel1} with general pulse and response functions $p$ and $q$ is lacking, though.
So far, building on \citep{Kuramoto1987}, O'Keeffe and Strogatz studied how the width of (symmetric) pulses affects the collective dynamics, while considering a flat response function $q\equiv$~\emph{const} and identical $b_j=b$ \citep{okeeffe_strogatz_2016}.
They found that broad pulses may entail collective dynamics different from those generated by narrow pulses, which is consistent with results on the Winfree model by 
\cite{pazo_montbrio_2014} and \cite{gallego_et_al_2017}.
It remains unclear, however, how these results carry over to $\theta$-neuron networks including sinusoidal response functions and excitable dynamics.

A crucial tool for pinpointing the effect of the pulse width on the synchronization properties of large populations of Winfree oscillators and active rotators has been an exact dimensionality reduction first proposed by \cite{ott_antonsen_2008}.
Since then, a plethora of studies have adopted the ``Ott-Antonsen ansatz'' to facilitate mean-field analyses.
But to do so, they had to rely on analytically tractable pulse and response functions, e.g., resorting to symmetric and broad pulses as suggested by \cite{ariaratnam2001phase}, and thus often at the cost of biological realism.
By that time, the ``Ariaratnam-Strogatz'' (AS) pulse had already found its way into networks of $\theta$-neurons~\citep{goel2002synchrony,Ermentrout_2006} as a mathematical tractable alternative to similar, mainly symmetric, pulses of finite width~\citep[see, e.g.,][]{gutkin2001turning,osan2001two,borgers2003synchronization,borgers_kopell_2005,komek2012dopamine,Gutkin2014}, which were introduced either for numerical convenience (to smooth the discontinuous $\delta$-spikes) or on the premise that instantaneous pulse-coupling reflected realistic synaptic potentials. 
Typically, however, the authors did not specify which synaptic processes they meant to replace by broad pulses.
In \cref{subsec:D2}, I reexamine their premise and argue that pulses irrespective of their (finite) width cannot reflect real synaptic dynamics---at the postsynaptic site---and ought not to be used to replace synaptic transmission with synaptic kinetics or conductance-based synapses in networks of spiking neurons.
Regardless, the combination of the Ott-Antonsen ansatz with the analytically tractable, though biologically debatable, AS pulses turned out to be pivotal to studying the collective dynamics of pulse-coupled $\theta$-neurons~\citep{luke_barreto_so_2013,So-Luke-Barreto-14,luke2014macroscopic,Laing_2014,laing_2015,laing2016bumps,laing2016travelling,roulet2016average,laing2017,chandra2017modeling,laing2018dynamics,laing2018chaos,aguiar2019feedforward,lin_barreto_so_2020,laing2020effects,laing2020moving,means2020permutation,blasche2020degree, bick_goodfellow_laing_martens_2020,juettner_martens_2021,omel2022collective,birdac2022dynamics}.
As of today, however, pulse-coupling in $\theta$-neurons has widely eluded a biological justification and the role of the pulse shape for the collective dynamics remained unclear.
I hope that the ideas put forward in \cref{sec:2all} contribute to the discussion in a productive way. In particular, I offered two alternative interpretations of pulses:
As long as pulses of arbitrary shape are sufficiently narrow, they can reflect swift synaptic transmission including both the pre- and postsynaptic site (\cref{sec:2}).
Alternatively, pulses between $\theta$-neurons can describe voltage-gated conductances, or neurotransmitter release, at the presynaptic site (\cref{sec:3}).
This second interpretation builds on the change of coordinates through the inverse transform $\theta_j = 2 \arctan(v_j)$ from the QIF model (\cref{remark2}), which implicates, however, that the network of pulse-coupled $\theta$-neurons may no longer represent the canonical model for the universal class of weakly connected Class 1 excitable neurons.


Pulse-coupling approaches in networks of (quadratic) integrate-and-fire neurons are by default simpler, as synaptic transmission is generally modeled with $\delta$-spikes. 
Other forms of pulse-coupling are rarely found in the literature,
although the model formulation in terms of voltages $v$---as opposed to the phase description of the $\theta$-neuron---is in principle favorable for modeling voltage-dependent synaptic activation, which may entail pulses $p(v)$ with variable shapes.
The shape of the emitted pulse, however, depends not only on the activation function $p$, but also on the action potential.
This can become critical for integrate-and-fire neurons because their fire-and-reset mechanism implicates ``action potentials'' with a rather artificial shape:
at the moment of the spike, the neuron's voltage is instantaneously reset and then starts integrating again (\cref{fig:cartoon}c). 
A sigmoidal activation function of the presynaptic voltage, similar to those used in conductance-based neurons, then lops off the pulse of an integrate-and-fire neuron right after its spike.
The resulting pulse exhibits a moderate increase and a radical decrease much in contrast to those pulses generated by conductance-based neurons (cf.~\cref{fig1,fig1b}).

Biological concerns aside, networks of heterogeneous QIF neurons with synapses exhibiting such a sigmoidal voltage-dependence were studied by~\cite{ratas_pyragas_2016}.
The synaptic pulses of finite width turned out a necessary ingredient to synchronize QIF neurons and led to collective oscillations, whereas global coupling via instantaneous $\delta$-spikes is known to rule out collective oscillations of QIF neurons~\citep{montbrio_pazo_roxin_2015,juettner_martens_2021}.
Yet, the voltage-dependent pulses in~\citep{ratas_pyragas_2016} synchronized only excitatory, but not inhibitory neurons,
which is at odds with the theoretical finding that Class 1 neurons (including the QIF model) tend to be much more easily synchronized by mutual inhibition than excitation~\citep{wang2010review}.
A mechanistic explanation for this paradox has been elusive, but may fall back to the artificial shape of the QIF's action potential.
The conventional sigmoidal description of voltage-dependent pulses had thus to be revised to compensate for the abrupt fire-and-reset mechanism of QIF neurons and, eventually, to reveal general principles of emergent collective behavior in spiking neuron networks.
In \cref{sec:3}, I proposed such a compensation scheme by taking the QIF's action potential into account, which led to a variety of (a)symmetric pulses that, unfortunately, are analytically intractable to study the network dynamics in more detail. Therefore, I approximated those pulses by analytically more favorable pulses $p_{r,\varphi,\psi}$ in \cref{subsec:accessible}, which allowed for exact low-dimensional collective dynamics of globally pulse-coupled QIF neurons (\cref{sec:4}).
As their attractors lie on the invariant ``Lorentzian'' manifold, it sufficed to study the RV dynamics \eqref{eq:FRE} with the mean pulse activity $P_{r,\varphi,\psi}=\langle p_{r,\varphi,\psi}\rangle$,
which are two-dimensional and amenable to a comprehensive bifurcation analysis for instantaneous pulse-coupling through $p_{r,\varphi,\psi}$.

I would like to remark that in the framework of the RV dynamics \eqref{eq:FRE} on the Lorentzian manifold, pulsatile coupling is not restricted to the (smooth) pulses $p_{r,\varphi,\psi}$ proposed above. 
It is possible to design different voltage-dependent pulses $p(v)$ that allow for an accessible mean field $P(R,V)=\langle p \rangle = \int_{\mathbb{R}} p(v) \mathcal W(v,t)dv$ closed in $R$ and $V$ thanks to the averaging with respect to the invariant Cauchy-Lorentz distribution density $\mathcal W(v,t)$ of the total voltage density \eqref{eq:CL_density} of globally coupled QIF neurons. 
This approach was pursued by~\cite{ratas_pyragas_2016}, who used uniform (square) pulses $p_\text{sq}(v)$ that are activated when a neuron's voltage exceeds the value $v_\text{on} \leq \infty$ and terminated right at the spike time, see \cref{eq:psq} with $v_\text{off}=\infty$. 
By averaging the pulses $p_\text{sq}(v)$ with respect to $\mathcal W(v,t)$, one obtains the mean presynaptic pulse activity $P_\text{sq} = \langle p_\text{sq}\rangle$ of the square pulses $p_\text{sq}(v)$, \cref{eq:psq_mean},
which can readily be used as the recurrent input $I_\text{syn}$ in the RV dynamics~\eqref{eq:FRE}.
Importantly, this approach only applies to the dynamics on the Lorentzian manifold, where the voltages of globally coupled QIF neurons are known to be distributed according to Cauchy-Lorentz distribution density~\eqref{eq:CL_density}.
To capture transient dynamics beyond the Lorentzian manifold, one has to resort to the exact low-dimensional system~\eqref{eq:Pls} and the mean pulse activity $\langle p\rangle$ can no longer be found by averaging with respect to $\mathcal W(v,t)$, 
but has to be determined in terms of the macroscopic variables $\Phi,\lambda$ and $\sigma$.
This can be achieved, e.g., by following the strategy outlined in \cref{sec:colvar_derivation}. 
In general, this strategy requires certain assumptions on the pulse functions $p(v)$---which the smooth pulses $p_{r,\varphi,\psi}$ given by \cref{eq:KJ_theta} satisfy but $p_\text{sq}(v)$ does not \citep[see, e.g., Eq.~(4) in][and compare with Appendix F]{ratas2019noise}
---but eventually allows 
for a rigorous and exact mean-field reduction, as well as for the mathematical analysis of the collective dynamics, of pulse-coupled networks of spiking neurons.




\subsection{Pulse parameters, delays and gap junctions}
\label{subsec:D1}
In this section, I will discuss the three parameters---pulse width $r$, asymmetry $\varphi$, and shift $\psi$---in more detail that shape the smooth pulses $p_{r,\varphi,\psi}$ and, in turn, also the collective dynamics of pulse-coupled neuron networks.

The first parameter $r\in[-1,1]$ tunes the width of the pulse and interpolates between continuous, flat ($r=-1$) and discrete, event-triggered ($r=1$) synaptic transmission.
Even for large $0\ll r < 1$, the pulsatile transmission is smooth and the presynaptic neuron almost always sends out a signal.
Still, the larger $r$, the more localized is this signal around a particular (threshold) voltage $v\approx v_\text{thr}=\tan(\psi/2)$ or, equivalently, a particular central phase $\theta\approx\psi$.
The effect of the emitted pulse (with large $r$) on the postsynaptic neuron becomes negligible for voltages farther away from $v_\text{thr}$ (\cref{fig:RP_pulses_all}a, right panel).
If $v_\text{thr}$ coincides with the QIF spiking threshold $v_p$ at infinity ($v_\text{thr}=v_p\to\infty$), and in absence of asymmetry ($\varphi=0$), the ``Rectified-Poisson'' (RP) pulses $p_{r,0,\pi}$ resemble the ``Ariaratnam-Strogatz'' (AS) pulses $p_{\text{AS},n}$ that have been frequently employed in the context of $\theta$-neuron networks~\citep[see, e.g.,][]{goel2002synchrony,Ermentrout_2006,luke_barreto_so_2013,luke2014macroscopic,So-Luke-Barreto-14,Laing_2014,laing_2015,laing2016bumps,laing2016travelling,roulet2016average,laing2017,chandra2017modeling,laing2018dynamics,laing2018chaos,aguiar2019feedforward,lin_barreto_so_2020,laing2020effects,laing2020moving,means2020permutation,blasche2020degree,bick_goodfellow_laing_martens_2020,juettner_martens_2021,omel2022collective,birdac2022dynamics}.
The usability of AS pulses, however, is hampered by the series representation of the mean pulse activity $\langle p_{\text{AS},n}\rangle$. 
The larger $n$, the narrower the AS pulse, and the more convoluted $\langle p_{\text{AS},n}\rangle$, see \cref{appsec:C}.
In fact, $\langle p_{\text{AS},n}\rangle$ does not allow for a closed-form description in the exact low-dimensional system \eqref{eq:Pls}. 
By contrast, the mean RP pulse activity $P_{\text{RP},r}=\langle p_{r,0,\pi}\rangle$ can be expressed in terms of the three macroscopic variables $\Phi,\lambda,\sigma$ that completely determine the collective dynamics. 
On the Lorentzian manifold, $P_{\text{RP},r}$ simplifies to a concise function of the firing rate $R$ and the mean voltage $V$, see \cref{eq:RP_mean}, and allows for a straightforward bifurcation analysis even for arbitrarily narrow pulses (\cref{subsec:3B}).

The second parameter, $\psi\in[0,2\pi)$, shifts the pulse to the right ($\psi>\pi$) or to the left ($\psi<\pi$), but does not change its actual shape.
For symmetric pulses ($\varphi=0$), the pulse is strongest at the phase $\theta=\psi$, or when the presynaptic voltage $v$ crosses the ``virtual threshold'' $v_\text{thr}=\tan(\psi/2)$.
The nomenclature of a virtual threshold becomes rigorous for Dirac $\delta$-pulses $p_{1,0,2\arctan(v_\text{thr})}$, that are emitted at the moment when $v=v_\text{thr}$.
Depending on the sign of $v_\text{thr}$, the pulse is shifted to the right or to the left of the QIF spiking threshold $v_p=\infty$, which can be interpreted as an effective delay or advance of the postsynaptic response \citep[cf.][]{Ermentrout_1996,gutkin2001turning}.
Indeed, if the pulse reaches its peak when $v=v_\text{thr}\ll 0$, then the effect on the postsynaptic response is strongest after the actual spiking of the presynaptic neuron.
On the other hand, for $v_\text{thr}\gg 0$, the emitted pulse is strongest already before the neuron has spiked.

For threshold values after the spike, i.e.~$v_\text{thr} < 0$, and as expected for delayed synaptic transmission~\citep{roxin2005role,roxin_montbrio_2011}, collective oscillations can be found for inhibition ($J<0$) and (strong) excitatory drive $I_0>0$ (\cref{fig:KJ_pulses_all}b and c).
The results are qualitatively identical for sufficiently narrow pulses with width $0 \ll r \leq 1$ and threshold value $v_\text{thr} \ll 0$, so that it suffices to report the phase diagram for one set of parameters only (here, $r=0.95, v_\text{thr} = -20$).
As a word of caution I stress that shifted pulses mimic truly delayed pulses only on the single neuron level; on the macroscopic level, they may lead to very distinct collective dynamics. 
Reminiscent of that is the cusp-shaped region of bistability between two (very) low-activity states for $J<0$ in \cref{fig:KJ_pulses_all}(b and c).
Note also that in contrast to ``real'' time delays, instantaneous recurrent coupling with shifted pulses cannot store the neurons' history and transmit it unaltered after the delay.
That is why the RV dynamics \eqref{eq:FRE} remain low-dimensional, whereas real delays yield infinite-dimensional dynamics that can entail more complex collective behavior~\citep{pazo_montbrio_2016,Devalle2018}.

For advanced (or left-shifted) pulses with threshold values before the spike ($0 \ll v_\text{thr}<\infty$), collective oscillations occur for excitatory coupling ($J>0$) and cease via a supercritical Hopf bifurcation when decreasing the coupling strength (\cref{fig:left-pulses}), cf.~\cref{eq:dPdelta_dV}.
This insight also explains why the square pulses \eqref{eq:psq} used by \cite{ratas_pyragas_2016}, whose mean is shifted to the phase before the actual spike, only induce collective oscillations for excitatory neurons but not in the biologically more plausible case of inhibition (\cref{subsec:previous,subsec:3D}).
Moreover, the observed bifurcation scenario for left-shifted pulses resembles the case of gap junction-coupling \citep[cf.~Fig.~7 of][]{Pietras_et_al_2019}.
Gap junctions are, in general, known to promote neural synchrony and facilitate collective oscillations~\citep{laing_2015,Pietras_et_al_2019}.
QIF neurons that are globally coupled via gap junctions of strength $J$ (but not via pulses) have the subthreshold dynamics
\begin{equation}
    \tau_m \dot v_j = v_j^2 + i_0 + \frac{J}{N}\sum_{k=1}^N (v_k - v_j) + I_j \ .
    \label{eq:QIF_gap}
\end{equation}
On the Lorentzian manifold, the exact RV dynamics corresponding to the microscopic dynamics \eqref{eq:QIF_gap} in the limit $N\to\infty$ takes on the same form as Eq.~\eqref{eq:FRE} when identifying $V = \langle v_j \rangle - J/2$ as a shifted mean voltage, $I_0 = i_0 + J^2/4$ and recurrent synaptic input $I_\text{syn}=J V$.
Thus, in line with the results of \cref{subsec:5A}, the onset of collective oscillations in the gap junction-coupled network~\eqref{eq:QIF_gap} is neatly explained by an effective voltage-coupling.
The similarity between the collective dynamics for advanced pulses on one hand and for gap junctions on the other hand, supports the notion that pulse-interactions indeed induce an effective voltage component in the recurrent coupling.
Again a word of caution is due about the correct interpretation of temporally vs.\ effectively advanced pulses:
With instantaneous pulse-coupling, the mean field $\langle p_{r,0,2\arctan(v_\text{thr})} \rangle$ has an immediate effect on individual neurons.
That is, emitting an inhibitory pulse before the actual spike may possibly hinder the same neuron to actually spike. 
By allowing for synaptic kinetics, as in Eqs.~\eqref{eq:s_des} or \eqref{eq:biexp} below, one can alleviate this intricacy and indeed interpret the time at which a neuron crosses the virtual threshold $v_\text{thr}$ as the activation time $t_a$ of the postsynaptic response~\citep[cf.][]{Ermentrout_1996}.

Finally, the asymmetry parameter $\varphi\in[-\pi,\pi)$ skews the pulse $p_{r,\varphi,\psi}(\theta)$ and shifts its bulk to the right ($\varphi>0$) or to the left ($\varphi<0$) of the central phase $\theta=\psi$; note that $\varphi\neq0$ is only effective for pulses of finite width $r<1$.
Already a small value of $0<|\varphi|\ll 1$ can have a large effect on the network dynamics and easily induce collective oscillations (Figs.~\ref{fig:sim}b and \ref{fig:KJ_pulses_all}d).
If the pulses are slightly skewed to the phase after the spike ($\varphi>0$), collective oscillations emerge almost naturally for inhibition ($J<0$) with sufficient excitatory drive ($I_0>0$).
The effect is similar to that of synaptic kinetics and, indeed, for these right-skewed pulses, one retrieves a fast rise and slower decay of the postsynaptic response (\cref{fig1b}c2) as observed for first- and second-order synapses.
For left-skewed pulses ($\varphi<0$), by contrast, the mean of the pulse is advanced and collective oscillations emerge only for excitation ($J>0$; \cref{fig:left-pulses}), similar to left-shifted pulses ($\psi<\pi$, $v_\text{thr}>0$) or gap junctions.

\subsection{Can instantaneous pulses replace complex synaptic transmission?}
\label{subsec:D2}
In contrast to the mathematical abstraction of $\delta$-spikes, smooth pulses of finite width have often been assumed to be biologically more realistic and to better approximate postsynaptic responses like those of conductance-based, Hodgkin-Huxley-like neuron models~\citep{Ermentrout_2006,luke_barreto_so_2013,bick_goodfellow_laing_martens_2020}.
It seems daunting to include various levels of biochemical realism at a chemical synapse, so in general one jumps over the explicit modeling of (1) how an increase (depolarisation) in the voltage $v_\text{pre}$ of a presynaptic neuron activates voltage-gated Ca$^{2+}$ channels, (2) how the Ca$^{2+}$-influx induces the release of neurotransmitters that diffuse to the postsynaptic neuron and bind to specific receptors with different possible mechanisms~\citep{Kaeser_Regehr_2014},
and (3) how the binding of neurotransmitters triggers the opening of ionic channels and eventually generates a postsynaptic current.
Instead, the synaptic process is described phenomenologically---but not derived from first principles~\citep{Koch_2004}---by the synaptic input $I_\text{syn}$ to the postsynaptic neuron,
\begin{equation}
I_\text{syn} = - g_\text{syn}\big(v_\text{pre}(t), t\big) \big[v_\text{post}(t) - E_\text{syn} \big] \ ,
\label{eq:Isyn1}
\end{equation}
with reversal potential $E_\text{syn}$ and a synaptic conductance that is often represented as $g_\text{syn}(t) = \hat g_\text{syn} s(t)$ with maximal synaptic conductance $\hat g_\text{syn}\ge 0$ and a gating variable $s(t)$ that may be interpreted as the fraction of open channels releasing neurotransmitters. 
If the synaptic conductance is activated by a (sigmoidal-like) function $f(v_\text{pre})$ of the presynaptic membrane potential, as discussed in \cref{sec:3}, and follows first-order synaptic kinetics, the dynamics of the gating variable $s(t)$ is given by
\begin{equation}
\dot s= a_r f\big(v_\text{pre}(t-\tau_l)\big) (1-s) - a_d s \ .
\label{eq:s_des}
\end{equation}
The constants $a_r$ and $a_d$ determine the rise and decay times of the postsynaptic response~\citep{Wang_Rinzel_1992,Wang_Buzsaki_1996,Destexhe1994,Ermentrout_Terman_2010} and a possible latency time $\tau_l$ can account for finite axonal propagation times~\citep{Ermentrout_Terman_2010}.
Depending on the shape of the presynaptic action potential, 
$s(t)$ can actually begin to rise before the presynaptic voltage reaches its peak (corresponding to the neuron's spike time $T_\text{pre}$), especially when the action potential is broad~\citep{Ermentrout_1996}.
Alternatively to Eq.~\eqref{eq:s_des}, the time course of $s(t)$ can be described by the difference of two exponential functions,
\begin{equation}
s(t) = A \big[ e^{-(t-t_a)/\tau_r} - e^{-(t-t_a)/\tau_d} \big] , \; t \geq t_a,
\label{eq:biexp}
\end{equation}
with amplitude $A$ and activation time $t_a$, which is typically around the spike time $T_\text{pre}$ of the presynaptic neuron plus some latency $\tau_l$ (possibly due to finite axonal propagation speed and taking into account that the postsynaptic response may start before $T_\text{pre}$).
Biexponential synapses \eqref{eq:biexp} with characteristic latency, rise, and decay time constants $\tau_{l,r,d}$ are the gold standard in computational models of spiking neurons---they allow for mean-field approaches~\citep{treves_1993,Abbott_vanVreeswijk_1993,brunel_hakim_1999,brunel_wang_2003,brunel_hansel_2006} and are typically reported in the experimental neuroscience literature (although there are no coherent definitions for $\tau_{l,r,d}$). 

In network models of spiking QIF or $\theta$-neurons, 
pulses $p_{r,\varphi,\psi}$ of finite width ($r<1$), as described by \cref{eq:KJ_theta}, are an ideal candidate for the presynaptic voltage-dependent activation, or release, function $f(v_\text{pre})$ in \cref{eq:s_des}.
The pulse is activated already shortly before the QIF neuron reaches the peak of its action potential (\cref{fig1b}) and can last even through its recovery period, see \cref{sec:3}.
The versatility of the pulses $p_{r,\varphi,\psi}$ further allows to accentuate the synaptic activation, or the release of neurotransmitters, on either phase of the action potential through the asymmetry parameter $\varphi\neq0$ or the shift parameter $\psi\neq\pi$. 
Thereby, it is possible to account for physiological conditions under which the opening of voltage-gated Ca$^{2+}$ channels, and consequently neurotransmitter release, is advanced, e.g., at increased temperature~\citep{sabatini1996timing,sabatini1999timing,volgushev2004probability,yang2006amplitude,chao2019timing,van2020temperature}.

When including synaptic kinetics of the form \eqref{eq:s_des} that are activated by voltage-dependent pulses, one has to be careful how to incorporate the corresponding postsynaptic responses in mean-field models. 
Indeed, when summing over a large number of postsynaptic responses $s_j$, the product $f(v_\text{pre,j}) (1-s_j)$ with the presynaptic pulses $f(v_\text{pre,j})$ presents a nonlinear problem that can only be resolved approximately\footnote{When summing the responses $s_j$ over all presynaptic neurons $j$ with individual responses given by \cref{eq:s_des}, the mean response reads
$\dot S = \langle \dot s_j \rangle = a_r \langle p_{r,\varphi,\psi} \rangle - a_d S - a_r \langle p_{r,\varphi,\psi}(\theta_j) s_j \rangle$.
The last term $\langle p_{r,\varphi,\psi}(\theta_j) s_j \rangle $ represents an average of the product of the presynaptic pulse with the current state of the postsynaptic response.
In general, $p_{r,\varphi,\psi}(\theta_j)$ and $s_j$ are not independent but correlated, so that taking the mean of their product does not yield a closed equation in terms of $S$ and $\langle p_{r,\varphi,\psi} \rangle$.
For $\delta$-spikes (in the limit $(r,\varphi,\psi)\to(1,0,\pi)$)  and under the Poissonian assumption that presynaptic spike trains have a coefficient of variation (CV) close to $1$, one can approximate $\langle p_{1,0,\pi}(\theta_j) s_j \rangle \approx (\pi \tau_m R) S$, cf.~\citep{pietras_schmutz_schwalger_2022}, and possibly loosen the Dirac $\delta$-pulse assumption to obtain $\langle p_{r,\varphi,\psi}(\theta_j) s_j \rangle \approx \langle p_{r,\varphi,\psi} \rangle S$ for $r$ close to $1$.
Yet, the Poissonian assumption, and hence the foregoing approximation, is difficult to justify in strongly correlated collective states, e.g., of regular synchrony~\citep{clusella_montbrio_2022}.
}, even for the analytically tractable pulses $f(v_\text{pre,j})=p_{r,\varphi,\psi}(\theta_j)$.
Biexponential synapses~\eqref{eq:biexp} do not suffer from this shortcoming, which may explain their success in (mean-field models in) computational neuroscience. 
As the postsynaptic response $s_j$ is typically activated at the presynaptic spike time $T_j$, it is advantageous to rewrite~\eqref{eq:biexp} as
\begin{equation}
 \tau_d \dot s_j = -s_j +  u_j, \; \tau_r \dot u_j = -u_j + \tau_0 p_{1,0,\pi}\big(\theta_j(t-\tau_l)\big),
 \label{eq:sj}
\end{equation}
with normalization factor $\tau_0$; one retrieves~\eqref{eq:biexp} for $\tau_0=A(\tau_d-\tau_r)/(2\pi\tau_r\tau_d)$.
To avoid infinite-dimensional time-delayed neuronal dynamics when $\tau_l>0$, 
one can follow \citep{Ermentrout_1996} and use shifted Dirac $\delta$-pulses $p_{1,0,2\arctan(v_\text{thr})}$ with a negative virtual threshold $v_\text{thr}<0$ that effectively delay the postsynaptic response.
Likewise, one can use positive virtual thresholds $v_\text{thr} >0$ to account for activation times of the synapse before the actual spike time, $t_a < T_j$.
Since there are no nonlinear terms in \eqref{eq:sj}, one can readily average the postsynaptic responses over the population and obtain a concise mean-field model.

While I have shown how to interpret, and incorporate, the novel pulse function $p_{r,\varphi,\psi}(\theta)$ in the traditional framework of synaptic transmission, 
one question remains: 
are instantaneous pulses of finite width adequate for replacing the complex processes involved in generating a postsynaptic current $I_\text{syn}(t)$ of the form \cref{eq:Isyn1}?
In other words, can instantaneous pulses $p_{r,\varphi,\psi}$ approximate  (the effect of) $I_\text{syn}(t)$ sufficiently well as previously hypothesized?
To answer this question, I will decompose the problem into two parts because $I_\text{syn}$ comprises two distinct mechanisms: \emph{synaptic kinetics} (encoded in the dynamics of $s$) and \emph{conductance-based synapses} (due to the explicit voltage-dependence in $[E_\text{syn}-v_\text{post}(t)]$).
As I will argue below, instantaneous pulses of finite width are not suited to replace either synaptic kinetics (\cref{subsub:syn_kin}) or conductance-based synapses (\cref{subsub:coba}) in networks of QIF or $\theta$-neurons, but should rather be used complementary to  traditional synaptic transmission (\cref{subsub:pulse-triggered}).


\subsubsection{Pulses of finite width do not replace synaptic kinetics}
\label{subsub:syn_kin}
To focus on the effect of synaptic kinetics, it is convenient to approximate the term $E_\text{syn}-v_\text{post}(t) \approx v_\text{eff}$ in \cref{eq:Isyn1} by an effective potential $v_\text{eff}$.
This approximation is valid if $v(t)$ spends most of the time near its rest state, and the recurrent synaptic input $I_\text{syn}$ in the QIF dynamics \eqref{eq:QIF} is given by $I_\text{syn}(t)=J S(t)$ with $J=\hat g_\text{syn} v_\text{eff}$ and $S(t) = \langle s_j(t) \rangle$.
Depending on the sign of $v_\text{eff}$, the coupling $J$ is excitatory ($v_\text{eff}>0$) or inhibitory ($v_\text{eff}<0$).
The question is now whether the time course of $s_k(t)$ as a postsynaptic response to the spiking of presynaptic neuron $k$ can be equally well explained with an instantaneous pulse $s_k(t) = p_{r,\varphi,\psi}(\theta_k(t))$ or with a $\delta$-spike-triggered biexponential synapse \eqref{eq:biexp} with realistic latency, rise, and decay time constants.

\begin{figure}[!t]
\centering
\includegraphics[width=0.45\textwidth]{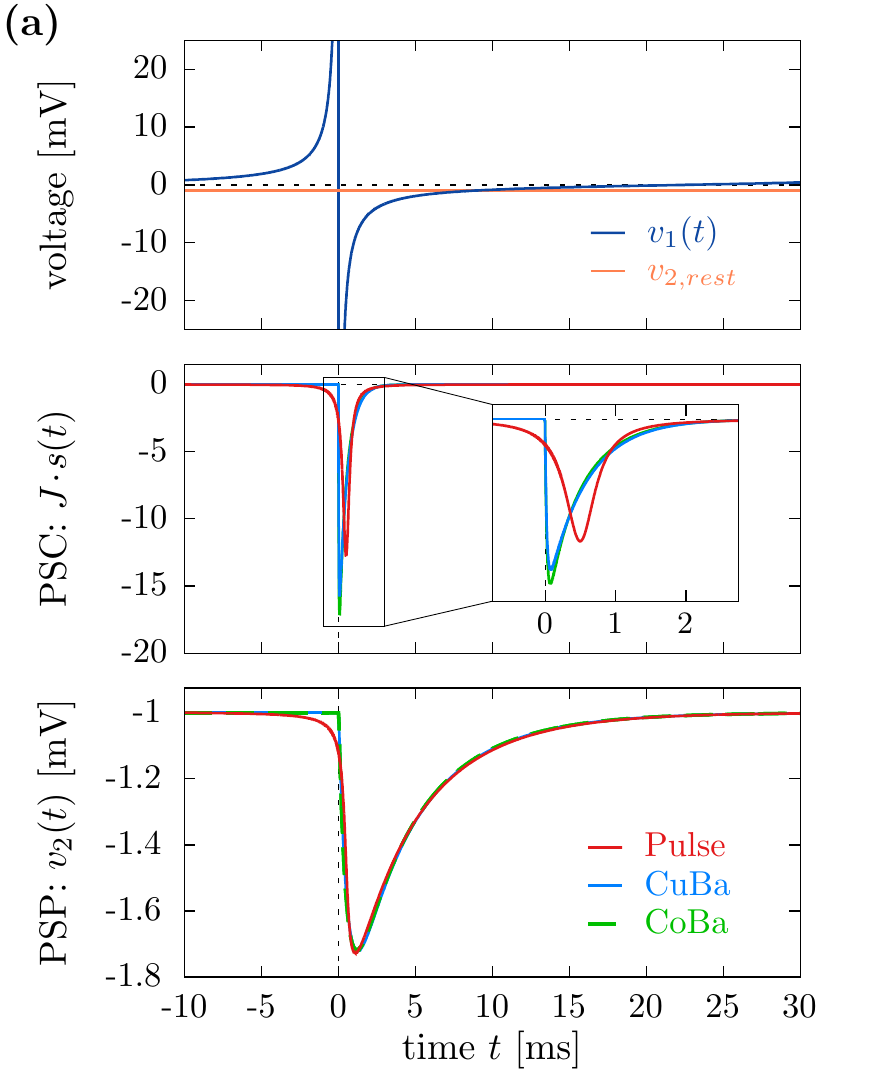} \hspace{0.0cm}
\includegraphics[width=0.45\textwidth]{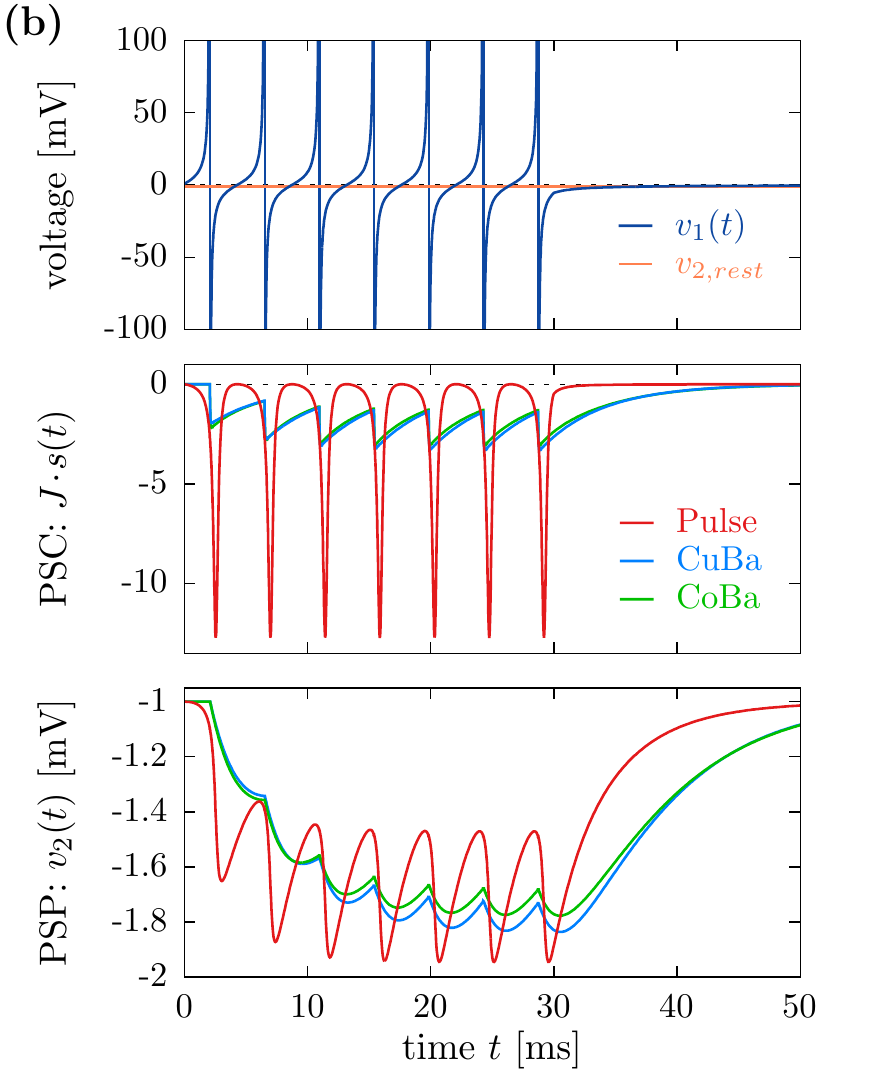} 
\caption{Pulse-coupling (red) does not replace traditional synaptic transmission via current-based (CuBa, blue) or conductance-based (CoBa, green) synapses with synaptic kinetics. 
Top: Presynaptic voltage traces $v_j$ of periodically spiking neuron $j=1$ (blue) and excitable neuron $j=2$ with constant input $I_2=-1$ and resting potential $v_{2,\text{rest}}=-1$ (orange).
Middle: Postsynaptic currents elicited by presynaptic neuron $1$ according to a pulse $J  p_{0.95,0,-2\arctan(20)}(\theta_1)$, CuBa synapse $J  s_1(t)$ or CoBa synapse $\hat g_\text{syn}s_1(t)[E_\text{syn}-v_2(t)]$, where $s_1(t)$ follows \cref{eq:sj} with $\tau_0=1,\tau_l=0$.
For the inhibitory synapse, $J=-\pi$ and $\hat g_\text{syn}=1.12\pi/4$ are chosen such that the PSPs (bottom) almost coincide; $E_\text{syn}=-5$ and $\hat g_\text{syn} 
 v_\text{eff} \approx J$ with $v_\text{eff}=(E_\text{syn}-v_{2,\text{rest}})=-4$.
Bottom: Postsynaptic potential/voltage response $v_2(t)$ according to PSCs.
(a) The voltage trace $v_2(t)$ in response to a single spike ($I_1=0.5$) can be sufficiently well described by pulses and more complex synaptic dynamics with short synaptic time constants $\tau_r=0.025$ms and $\tau_d=0.55$ms.
(b) Synaptic integration of rapidly incoming spikes ($I_1(t)=50$ for $t\leq 30$ms and $0$ afterwards) is characteristic for CuBa and CoBa synapses with realistic decay $\tau_d=5$ms, but not feasible with instantaneous pulses $p_{r,\varphi,\psi}$.}
\label{fig:limits}
\end{figure}

For simplicity, I consider two inhibitory QIF neurons coupled via shifted pulses $p_{r,0,\psi}$ of finite width $r=0.95$ (\cref{fig:KJ_pulses_all}a, blue curve).
The membrane time constant is fixed at $\tau_m=10$ms. 
Neuron $1$ receives a constant input $I_1=0.5$ and spikes at frequency $100/(\sqrt{2}\pi)\approx22.5$~Hz,
whereas neuron $2$ receives $I_2=-1$ and remains quiescent at its resting potential $v_{2,\text{rest}}=-1$. 
The shift parameter $\psi$ is chosen such that the pulse is strongest when the voltage $v$ of the presynaptic neuron crosses the virtual threshold $v_\text{thr} = -20 = \tan(\psi/2)$ after recovering from its reset. 
In this example, a pulse emitted by neuron $1$ is strongest $\sim\!\!0.5$ms after its spike and
the postsynaptic current of neuron $2$ is proportional to $s_1(t) = p_{0.95,0,-2\arctan 20}(\theta_1(t))$.
The postsynaptic potential (PSP) $v_2(t)$ evolves according to $\tau_m \dot v_2 = v_2^2 -1 + Js_1(t)$ with $J=-\pi$ and exhibits a rapid rise and more moderate decay with its peak at $\sim\!\!1.08$ms after the presynaptic spike.
One can fit the voltage response $v_2(t)$ of neuron $2$ to that produced by a biexponential synapse \eqref{eq:biexp} activated at the time of the presynaptic spike at $t=0$ with $\tau_l=0$ms.
The agreement between the voltage responses (PSPs) to the instantaneous pulse and to the biexponential synapse is remarkable, even though the perceived PSCs are quite different (\cref{fig:limits}a).
Nonetheless, the fitted rise and decay times, $\tau_r=0.025$ms and $\tau_d=0.55$ms, are by far shorter than those reported in the literature 
\citep[e.g., $\tau_r=0.5$ms and $\tau_d=0.5$ms used by][]{Wang_Buzsaki_1996,brunel_wang_2003}.
Furthermore, the duration from the presynaptic spike until the peak of the PSP differ by $0.2$ms, (peak pulse response at $1.08$ms vs.\ peak synaptic response at $1.28$ms).
This calls for introducing a latency time constant of $\tau_l=0.2$ms, which seems a reasonable value.
For wide pulses ($r=0.75$), however, the fitted latency increases and the postsynaptic response will be prompted at very low presynaptic voltage thresholds $v_\text{thr} \approx 0$ long before the actual spiking threshold.

Besides unrealistic synaptic time constants, another shortcoming of instantaneous pulses is the impossibility to integrate rapidly incoming inputs.
While this is not a problem for very fast synaptic kinetics (with short rise and decay time constants $\tau_{r,d}$),
synaptic integration becomes important for realistic synaptic decay of around $\tau_d=5$ms.
In the context of the two neurons from before, increasing the input of neuron $1$ to $I_1=50$ induces spiking at a faster frequency of $100\sqrt{50}/\pi\approx225$~Hz.
In the case of synaptic kinetics, the postsynaptic neuron $2$ integrates the subsequent spikes from neuron $1$ and both PSC and PSP exhibit increased baseline levels.
The inhibitory effect on the postsynaptic voltage $v_2$ is thus significantly larger as compared to instantaneously transmitted pulses of finite width (\cref{fig:limits}b).

In sum, instantaneous pulses of finite width can have a similar effect on individual postsynaptic voltage responses as biexponential synapses,
but the comparison is misleading.
First, the fitted time constants of biexponential synapses are far from realistic and, second, instantaneous pulses cannot integrate inputs, which is an important hallmark of synaptic kinetics. 
Hence, instantaneous pulses $p_{r,\varphi,\psi}$ do not replace synaptic kinetics.

\subsubsection{Pulses do not describe conductance-based synapses}
\label{subsub:coba}
Can instantaneous pulses of finite width substitute for conductance-based synapses?
To study the additional voltage-dependence in Eq.~\eqref{eq:Isyn1}, I will focus on instantaneous conductance-based synapses, that is, $g_\text{syn}(t) = (\hat g_\text{syn}/N) \sum_k \delta(t-T_k^l)$ is proportional to the spike train of the presynaptic neurons.
The voltage dynamics \eqref{eq:QIF} of globally coupled QIF neurons then becomes
\begin{equation}
\tau_m \dot v_j = v_j^2 + I_0 - \hat g_\text{syn} \tau_m R(t) [ v_j - E_\text{syn} ] + I_j(t)
\label{eq:QIF_coba}
\end{equation}
with $I_j$ as in \cref{eq:QIF_individual_input} and the same fire-and-reset rule at infinity as before. 
For positive maximal synaptic conductance $\hat g_\text{syn} > 0$, the reversal potential $E_\text{syn}$ determines whether the (net effect of the) recurrent coupling is excitatory ($E_\text{syn}>0$) or inhibitory ($E_\text{syn}<0$).

As in \cref{subsub:syn_kin} about synaptic kinetics, one can compare the postsynaptic voltage response to a presynaptic periodically spiking neuron when the neurons interact via ($\delta$-spike-activated) conductance-based synapses or via (instantaneous) pulses of finite-width.
With $I_{1,2}$ as before and setting $\hat g_\text{syn} (E_\text{syn} - v_\text{rest}) = J$ with $v_\text{rest}=-\sqrt{-I_2}=-1$, the PSC for the conductance-based synapse \eqref{eq:QIF_coba} is identical to a $\delta$-spike of strength $J$, and the voltage responses can only be approximated by sufficiently narrow pulses $p_{r,\varphi,\pi}$ with $r\to1$.
When augmenting the conductance-based synapse \eqref{eq:QIF_coba} by second-order synaptic kinetics, that is, $R(t)$ in \cref{eq:QIF_coba} is replaced by $s_j(t)$ according to \cref{eq:sj}, the resulting postsynaptic response almost coincides with the one by the biexponential synapse (\cref{fig:limits}).
Therefore, the pulse-approximation of the second-order conductance-based synapse suffers from the same shortcomings---unrealistic synaptic time constants and impossibility of synaptic integration---as is the case for the biexponential synapse.

On top of that, the pulse-approximation of conductance-based synapses becomes problematic in networks of (heterogeneous) QIF neurons 
even though the PSP amplitude of an individual neuron receiving conductance-based synaptic input matches the PSP for narrow pulse-coupling (\cref{fig:limits}).
But when a group of neurons receive heterogeneous inputs, they have variable resting potentials and, consequently, the resulting PSP amplitudes differ across the network. 
The PSPs may even vary in time as they depend on the current state of each neuron.
This feature of neural networks with conductance-based synapses is hardly possible to implement by pulses with a predefined shape that is moreover common to all neurons.

In addition, there are significant differences between instantaneous pulse-coupling and conductance-based synapses with respect to the collective dynamics.
As shown in \cref{sec:5}, instantaneous pulses of finite width generate, in general, collective oscillations in large networks of QIF neurons (albeit the parameter regions for symmetric pulses may seem degenerate).
By contrast, globally coupled QIF neurons with instantaneous conductance-based synapses \eqref{eq:QIF_coba} do not support collective oscillations that emerge via a Hopf bifurcation from an asynchronous state.
The proof follows the same lines as in \cref{subsec:5A}:
On the Lorentzian manifold, the RV dynamics \eqref{eq:FRE} for the microscopic dynamics \eqref{eq:QIF_coba} are
\begin{subequations}
\begin{align}
\tau_m\dot{R} &= \frac{\gamma}{\pi\tau_m} + 2RV - \hat g_\text{syn} \tau_m R^2\;, \\
\tau_m\dot{V} &= V^2 - (\pi \tau_m R)^2 + \hat g_\text{syn} \tau_m 
 R [E_\text{syn} - V ] + I_0\;.
\end{align}
\label{eq:FRE_cond}
\end{subequations}
The fixed-point solutions $(R^*,V^*)$ of \cref{eq:FRE_cond} satisfy $V^* = \hat g_\text{syn} \tau_m R^* / 2 - \gamma/ (2\pi\tau_m R^*)$.
A necessary condition for the oscillatory instability of $(R^*,V^*)$ via a Hopf bifurcation is that the trace of the Jacobian $\text{Jac}_\text{CoBa}$ of \eqref{eq:FRE_cond} vanishes.
Here, however,
\begin{equation*}
\text{tr} (\text{Jac}_\text{CoBa}) = -\frac{2\gamma}{\pi \tau_m R^*} - \hat g_\text{syn} \tau_m R^* < 0
\end{equation*}
is always negative because $\gamma, \hat g_\text{syn}, \tau_m R^* > 0$ by definition.
Hence, collective oscillations never occur through a Hopf bifurcation in networks of globally coupled QIF neurons with instantaneous conductance-based synapses that are triggered by presynaptic spikes.
In conclusion, pulses of finite width do not account for conductance-based synapses, either.

\subsubsection{Pulse-triggered synaptic kinetics}
\label{subsub:pulse-triggered}
Synaptic kinetics need not be triggered by $\delta$-spikes, but can also be initiated by general pulses $p_{r,\varphi,\psi}$,
which can be thought of as a combination of \cref{eq:s_des,eq:sj}:
Replacing the $\delta$-spikes in \cref{eq:sj} by general pulses $p_{r,\varphi,\psi}$, leads to the microscopic synaptic dynamics
\begin{equation}
 \tau_d \dot s_j = -s_j +  u_j, \quad \tau_r \dot u_j = -u_j + p_{r,\varphi,\psi}\big(\theta_j(t)\big).
 \label{eq:sj_p_general}
\end{equation}
The response dynamics \eqref{eq:sj_p_general} triggered by a narrow and possibly asymmetric pulse is more general than the conventional $\delta$-spike-triggered biexponential synapse \eqref{eq:biexp}.
First, it connects the response with the presynaptic action potential in a continuous manner, and thereby avoids the open discussion at which instant the synaptic response is actually triggered: Does the activation time $t_a$ in \cref{eq:biexp} denote the peak voltage of the action potential or a seemingly arbitrary threshold value?
And second, the pulse-triggered second-order dynamics~\eqref{eq:sj_p_general} can be used to fit more complex, experimentally verified impulse responses, e.g., for hippocampal neurons with an impulse response given by a multi-exponential function~\citep{BekkersStevens1996}
\[
s(t) = s_{\max} \{1-\exp[-(t-t_a)/\tau_r]\}^x \exp[-(t-t_a)/\tau_d] ,
\]
or, in the realm of biomechanics, for motor-unit twitches upon motoneuron discharge, whose impulse response is given by a generalized alpha-function~\citep{Fuglevand1993,RAIKOVA20021123,contessa_deLuca2013}
\[
s(t) = s_{\max} (t-t_a)^x \exp[-(t-t_a)/\tau_d] .
\]
In both cases, $x$ is a real-valued parameter (and not an integer), so that the dynamics of $s(t)$ cannot be described with (analytically tractable) differential equations.

On the network level, pulse-triggered synaptic kinetics~\eqref{eq:sj_p_general} further permit a concise mean-field reduction of the collective dynamics as before.
For global coupling of strength $J$, the recurrent synaptic input is given by $I_\text{syn} = J S(t)$ with $S(t) = \langle s_j(t) \rangle$.
Setting $\tau_0=1$ and $\tau_l=0$ in \cref{eq:sj_p_general} ensures that 
the macroscopic fixed points satisfy $S^*=P_{r,\varphi,\psi}(R^*,V^*)$, which allows for a direct comparison with instantaneous synaptic transmission in the limit $\tau_r=\tau_d\to0$, see Fig.~\ref{fig:KJ_pulses_all}(d) for an example; a comprehensive comparison, however, is beyond the scope of this paper.
The exact, augmented RV dynamics of globally coupled QIF neurons with second-order synaptic kinetics \eqref{eq:sj_p_general} read on the Lorentzian manifold
\begin{subequations}
\begin{align}
\tau_m\dot{R} &= \frac{\gamma}{\pi\tau_m} + 2RV \;, \\
\tau_m\dot{V} &= V^2 - (\pi \tau_m R)^2 + I_0 +J S(t) \; ,\\
\tau_d \dot{S} &= -S + U, \; \\
\tau_r \dot{U} &=-U + P_{r,\varphi,\psi}(R,V) \; ,
\end{align}
\label{eq:FRE_syn}
\end{subequations}
with the mean presynaptic activity $P_{r,\varphi,\psi}$ fully determined in terms of $R$ and $V$ according to \cref{eq:p_mean}.
For instantaneous rise time $\tau_r\to0$, \cref{eq:FRE_syn} describes the RV dynamics with first-order synaptic kinetics (with exponential decay $\tau_d$).
For $\tau_r=\tau_d$, the second-order synaptic kinetics of the RV dynamics \eqref{eq:FRE_syn} reduces to that of the so-called alpha-synapse.

For the sake of completeness, I also present the pulse-triggered conductance-based RV dynamics with synaptic kinetics by combining \cref{eq:FRE_cond,eq:sj_p_general,eq:FRE_syn}:
\begin{subequations}
\begin{align}
\tau_m\dot{R} &= \frac{\gamma}{\pi\tau_m} + 2RV - \hat g_\text{syn} R S  \;, \\
\tau_m\dot{V} &= V^2 - (\pi \tau_m R)^2 + \hat g_\text{syn}  S \big[ E_\text{syn} - V\big] + I_0 \; ,\\
\tau_d \dot{S} &= -S + U, \; \\
\tau_r \dot{U} &=-U + P_{r,\varphi,\psi}(R,V) \; ;
\end{align}
\label{eq:FRE_syn_cond}
\end{subequations}
similar conductance-based RV dynamics were reported by \cite{ratas_pyragas_2016}, 
but with a synaptic variable $S(t)$ for non-smooth pulses and first-order synaptic kinetics ($\tau_r=0$); see also the works by \cite{byrne2017mean,coombes_byrne_2019,keeley_byrne_2019,byrne2020next,byrne2022mean}.
Preliminary results suggest that the additional synaptic kinetics in the augmented RV dynamics~\eqref{eq:FRE_syn} or \eqref{eq:FRE_syn_cond} blur the effect of the pulse shape $p_{r,\varphi,\psi}$ on the collective dynamics, especially if the pulse is narrow (\cref{fig:first-order}).
It may hence suffice to resort to the conventional $\delta$-spike-interactions, setting $(r,\varphi,\psi)\to(1,0,\pi)$, when studying also synaptic kinetics.
A comprehensive analysis of the augmented RV dynamics~\eqref{eq:FRE_syn} and \eqref{eq:FRE_syn_cond} shall clarify this hypothesis, which I leave for future work.

\begin{figure}[t]
    \centering
    \includegraphics[height=5cm]{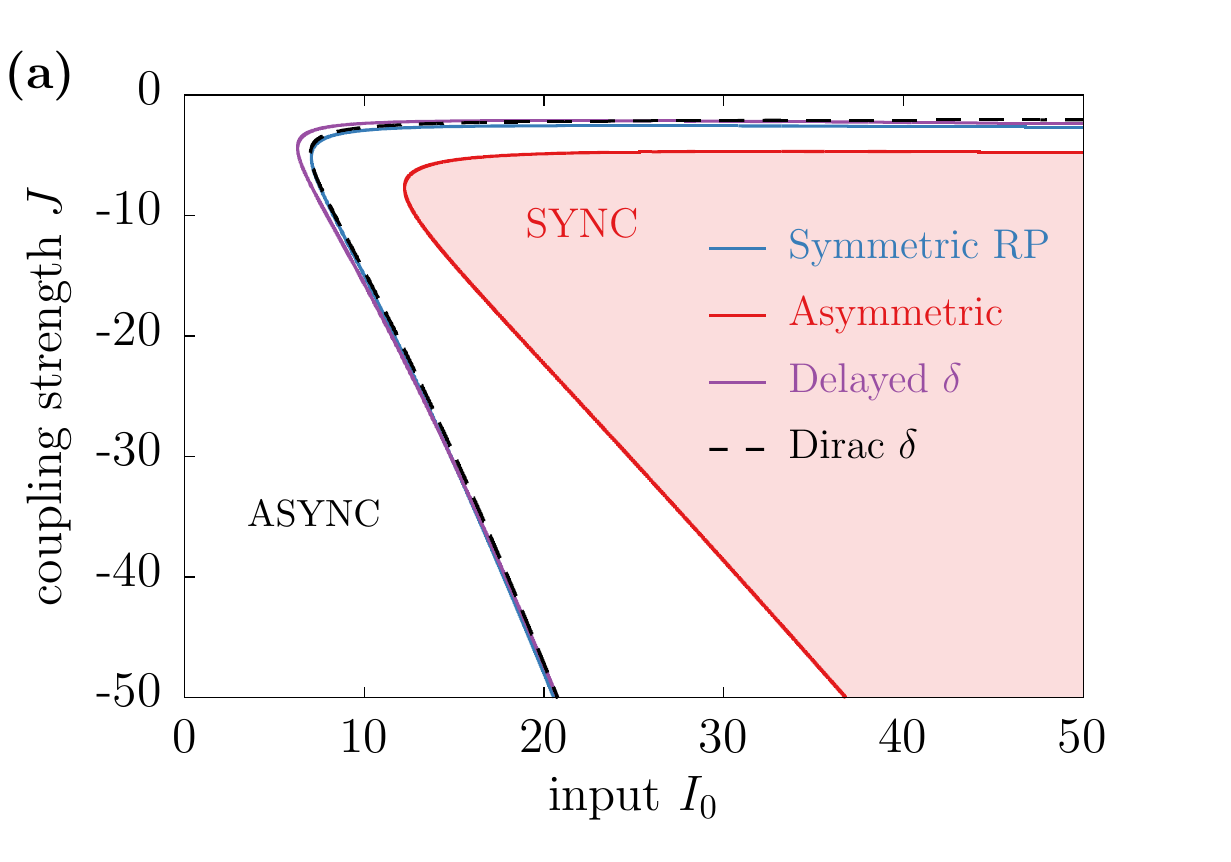} \hspace{-0.6cm}
    \includegraphics[height=5cm]{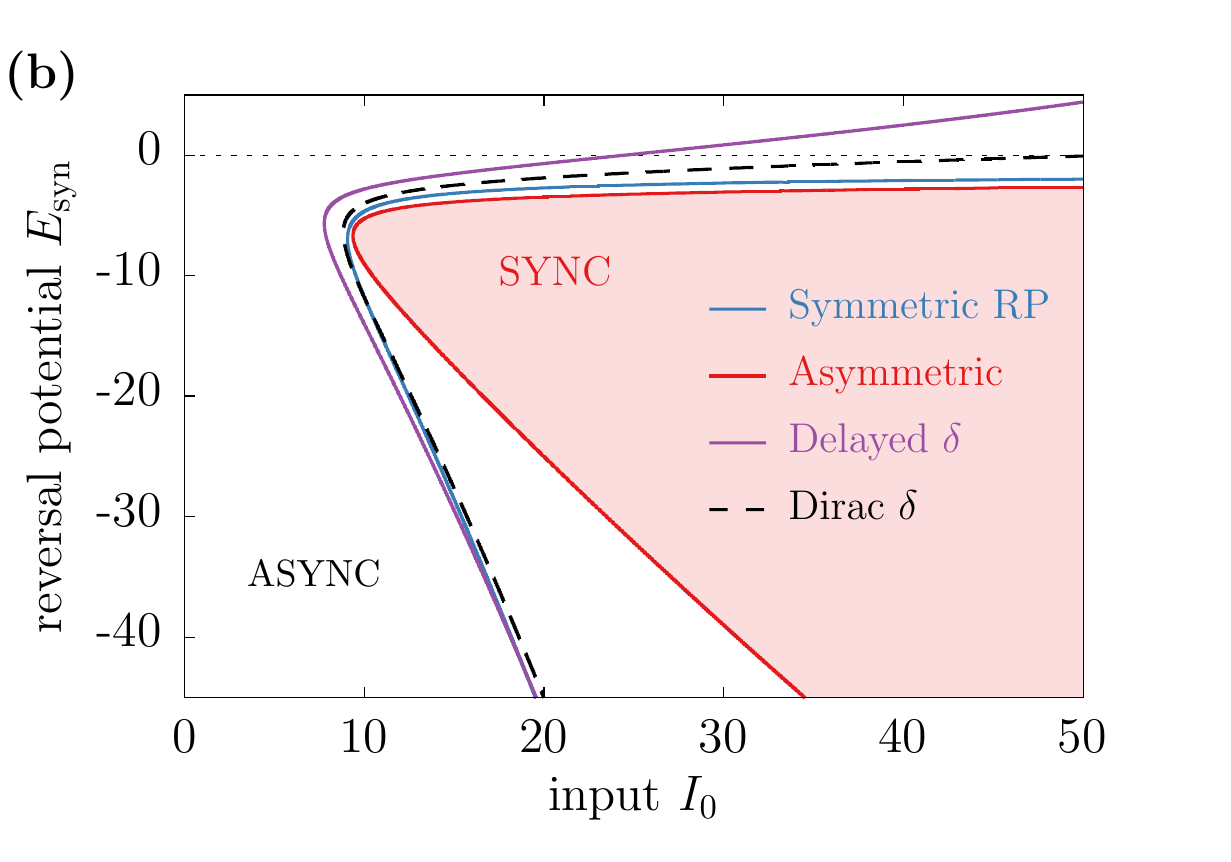}
    \caption{Collective oscillations among inhibitory QIF neurons due to first-order synaptic kinetics of (a) current-based and (b) conductance-based synapses are affected little by the pulse shape $p_{r,\varphi,\psi}$.
    The Hopf bifurcation boundaries of RP pulses ($r\!=\!0.95,\varphi\!=\!0,\psi\!=\!\pi$; blue) and of delayed $\delta$-pulses ($r\!=\!1,\varphi\!=\!0,\psi\!=\!2\arctan(-20)$; violet) coincide almost perfectly with the one for $\delta$-spikes ($r\!=\!1,\varphi\!=\!0,\psi\!=\!\pi$; black dashed); for asymmetric pulses ($r=0.95,\varphi\!=\!\pi/12,\psi\!=\!\pi$; red) the Hopf curve is slightly shifted to the right. 
    Collective oscillations are found in the ``SYNC'' region to the right of the Hopf boundaries (the shaded region indicates collective oscillations for the red asymmetric pulses).
    The Hopf boundaries are supercritical and were detected from \cref{eq:FRE_syn,eq:FRE_syn_cond} with $\tau_m=2\tau_d=10$ms and $\tau_r=0$.
    Note that the $y$-axis in (b) denotes the reversal potential $E_\text{syn}$ of conductance-based synapses with maximal synaptic conductance $\hat g_\text{syn}=1$; $E_\text{syn}$ plays a remarkably similar role as the coupling strength $J$ for current-based synapses in (a).}
    \label{fig:first-order}
\end{figure}


\section{Conclusions \& Outlook}
\label{sec:7}
Spiking neural networks are well-established in the neurosciences and a powerful tool in understanding cortical information processing, which originates from the exchange of action potentials between neurons. 
For computational advantages, mathematical tractability, and rigorous analysis, these networks use simple spiking neuron models that replicate the essential features of real neural dynamics, while interactions between neurons are modeled with infinitely narrow pulses \change{when the presynaptic neuron ``spikes''; these pulses, however, do not capture the more complex dynamics of real synapses.
Nonetheless and in} support of the spike-assumption, recent studies have led to believe that the shape of the action potential is indeed dispensable both on an individual \changeB{as well as on a network level:
(generalized leaky) integrate-and-fire (GLIF) models, which do not provide access to (realistic) action potentials, were shown to preserve intrinsic properties of single neurons and characteristic features of their spike generation observed in experiments~\citep{mensi2012parameter,pozzorini2015automated,teeter2018generalized}.
Moreover, GLIF point neurons capture experimentally observed population activity as faithfully as multi-compartment Hodgkin-Huxley-like models do ~\citep{rossert2017automated,arkhipov2018visual,billeh2020systematic}.}
It may be true that the shape of the action potential hardly affects the network dynamics, at least for uncoupled neurons.
But as soon as neurons are synaptically connected, the spotlight turns on the action potential.
The shape of pulsatile synaptic transmission between neurons directly depends on the action potential and can, ultimately and critically, affect the network dynamics.
By taking the shape of these interactions explicitly into account, the here proposed modeling framework enables a new perspective on pulse shape and voltage-dependent synchronization of spiking neuron networks, which has remained concealed for $\delta$-spike-interactions.


In the first part of this paper, I have proposed two biologically plausible interpretations of smooth pulsatile synaptic transmission in networks of spiking QIF and $\theta$-neurons---the difference between the two interpretations is subtle and based on the respective modeling assumptions.
\changeB{In general, pulses represent the interplay of a presynaptic action potential with a synaptic activation function $p$. 
For QIF and $\theta$-neuron models, $p$ needs to be modified to account for their simplified spiking behavior.}
When pulse-coupled networks of $\theta$-neurons are meant to replicate (weakly) connected Class 1 excitable neurons, 
the pulses must be sufficiently narrow as they reflect an \changeB{(almost)} instantaneous synaptic transmission process that can include both the presynaptic and the postsynaptic site; this interpretation carries over to QIF neurons (with infinite reset and threshold values) through the forward transformation $v=\tan(\theta/2)$.
Alternatively, and not necessarily on the premise that the neurons are Class 1 and weakly connected, QIF neurons emit pulses of arbitrary shape that can describe voltage-gated conductances, or the release of neurotransmitters, at the presynaptic site; this interpretation carries over to $\theta$-neurons via the inverse transformation $\theta=2\arctan v$.
\changeB{In either case, pulses may eventually be interpreted as a continuous generalization of the conventionally used $\delta$-spikes to install neurotransmitter-based chemical synapses.}

In the second part, I have put forward an exact low-dimensional macroscopic description for large networks of globally coupled QIF or $\theta$-neurons interacting via smooth pulses of various shapes that approximate the previously justified pulses.
The modeling framework allows for incorporating the recurrent synaptic input, mediated by a general family of pulse functions $p_{r,\varphi,\psi}$, in terms of a few macroscopic variables.
Thereby, one obtains system \eqref{eq:Pls} of three complex-valued ordinary differential equations that exactly describes the collective dynamics in the thermodynamic limit.
In the presence of (independent Cauchy white) noise or (Cauchy-Lorentz distributed) heterogeneity, the collective dynamics converges to an invariant manifold~\citep{pcp_arxiv_2022}, the so-called Lorentzian manifold~\citep{montbrio_pazo_roxin_2015}, on which the recurrent synaptic input is fully determined by the population firing rate $R$ and the mean voltage $V$.
On this manifold, the firing rate and voltage dynamics---the ``RV dynamics''~\cref{eq:FRE}---are closed in $R$ and $V$, remain two-dimensional for instantaneous pulse-coupling, and can readily be analyzed with respect to emergent collective behavior. 

For instantaneous synaptic interactions, I have proved that collective oscillations can only emerge when the recurrent input includes a voltage component.
This is the case, e.g., for electrical coupling via gap junctions.
In absence of gap junctions, the recurrent synaptic input can incorporate a voltage component---and hence allows for collective oscillations---if pulses transmitted via chemical synapses have a finite width or if the pulse peaks at a moment different from the neuron's spike time.
This insight strongly supports the voltage-dependent spike synchronization mechanism~\citep{devalle2017}, i.e.~a resonance in the neurons’ membrane and spiking dynamics~\citep{Pyle_Rosenblum_2017}, that is crucial for collective oscillations and typically not captured by traditional firing rate models.

Symmetric pulses centered about the spike time support collective oscillations in principle, but the parameter region where collective oscillations can be found appears somewhat degenerate due to the close vicinity of a Hopf bifurcation curve---where oscillations emerge supercritically---and a homoclinic bifurcation curve---at which oscillations are destroyed. 
Relaxing the symmetry condition of the pulse, or shifting the pulse peak away from the spike time, generates a wide region in parameter space where collective oscillations arise naturally and as the unique attractor of the network dynamics.
I have showed in networks of inhibitory QIF neurons with excitatory drive ($J<0,I_0>0$), that shifting the bulk of the pulse only slightly to the phase after the actual spike yielded robust ING oscillations.
Moreover, the collective dynamics generated by instantaneous asymmetric pulses resembled those generated by first-order synapses with $\delta$-spikes~\citep{devalle2017}.
In general, however, the correspondence between pulses of finite width and more detailed models of synaptic transmission is elusive.
Put differently, (narrow) pulse-coupling complements, but does not replace, synaptic kinetics and conductance-based synapses.

As an outlook and to bring the presented formalism even closer to experimental data, I leave for future work the analysis of pulse-triggered synaptic kinetics or conductance-based synapses, see the augmented RV dynamics~\eqref{eq:FRE_syn} and \eqref{eq:FRE_syn_cond}, or when incorporating finite threshold and reset values and asymmetric spikes with $v_p \neq - v_r$ as considered by \cite{montbrio_pazo_2020,gast2023macro}, see also \cref{appsec:E}.
I mention in passing that one can also transform the QIF dynamics around realistic resting potentials $\approx\!-70$mV, 
see \cref{appsec:A}, without changing the overall collective dynamics qualitatively.
In the current work, I did not consider habituation nor activity-dependent modulation of synaptic transmission, which allowed me to disentangle the effect of the pulse shape on the collective dynamics. 
To do so, I introduced voltage-dependent pulses in \cref{sec:3} that arise through the interplay of a synaptic activation function with the action potential and, thus, directly account for its shape.
Now, the shape of the action potential, i.e.~its waveform, is not only cell type- and temperature-dependent~\citep{sabatini1999timing,gray1996chattering,bean2007action},
but may also undergo dynamic changes through various plasticity mechanisms~\citep{sabatini1999timing,bean2007action,byrne1996presynaptic}. 
These can lead to action potential broadening or amplitude reduction, which subsequently affects neurotransmitter release and synaptic transmission and may hence directly, or indirectly, contribute to learning and memory storage~\citep{gershman2023}.
The proposed pulse-coupling is versatile to incorporate dynamical changes of the action potential, which will open up new avenues for investigating network effects of (pre-)synaptic, intrinsic, and homeostatic plasticity 
complementary to previously proposed mean-field approaches~\citep{tsodyks1998neural,zierenberg2018homeostatic,schmutz2020mesoscopic,pietras_schmutz_schwalger_2022,taher2020exact,gast2020meanfield,gast2021meanfield,Bandyopadhyay2021.10.29.466427,chen2022exact,taher2022bursting,ferrara2023population,gast2023macro},
and leverage more detailed network simulations as in \citep{lavi2015shaping} that combine the computational advantages of spiking neural networks with biological realism at the microscopic level, including a more realistic synaptic transmission process.
\subsection*{Acknowledgements}
I want to thank A.~Pikovsky, E.~Montbri\'o, R.~Cestnik and A.~Daffertshofer for fruitful discussions.
The project has received funding from the European Union’s Horizon 2020 research and innovation programme under the Marie Sklodowska-Curie grant agreement No 101032806.

\renewcommand{\baselinestretch}{1}
\bibliographystyle{apalike}

\newpage
\appendix

\counterwithin*{equation}{subsection}

\setcounter{page}{1}
\renewcommand{\thepage}{S\arabic{page}}
\setcounter{equation}{0}
\setcounter{theorem}{0}
\renewcommand{\theequation}{\thesubsection\arabic{equation}}
\renewcommand{\thesubsection}{\Alph{subsection}}
\renewcommand{\thesubsubsection}{\Alph{subsection}.\arabic{subsubsection}}


\section*{Appendix}
\subsection{Proof ot \cref{thm:1}}
\label{appsec:proof}

A mathematically rigorous exposition of \cref{thm:1},
based on Theorem 1 in \citep{Izhikevich_1999} and building on the Assumptions \ref{item1}--\ref{item5} in the main text, reads:

\begin{theorem}
  \label{thm:app}
  Consider a weakly connected neural network of the form 
  \begin{equation}
  \dot X_j = F_j(X_j,\lambda) + \varepsilon G_j(X_1,\dots,X_N; \lambda,\rho, \varepsilon), \quad X_j \in \mathbb{R}^m, \quad j=1,\dots,N.
  \label{eq:app_thm1}
  \end{equation}
  The functions $F_j:\mathbb{R}^m \times \Lambda \to \mathbb{R}^m$ describe the individual dynamics of each neuron $j=1,\dots,N$, where $\lambda\in \Lambda$ summarizes system (control) parameters and is an element of some Banach space $\Lambda$, typically $\Lambda = \mathbb{R}^l$ for some $l>0$.
  The functions $G_j:\mathbb{R}^m \times \dots \times \mathbb{R}^m \times \Lambda \times \mathbb{R}^n \times \mathbb{R} \to \mathbb{R}^m$ describe weak perturbations of order $\mathcal{O}(\varepsilon)$ for small $|\varepsilon|\ll1$, which are due to recurrent coupling through the other neurons or due to external inputs $\rho = \rho_0 + \mathcal O(\varepsilon) \in \mathbb R^n$ for some $n>0$.
  We assume that each (uncoupled) subsystem $\dot X_j = F_j(X_j,\lambda)$ undergoes a SNIC bifurcation for some $\lambda=\lambda_0$,
  that each function $G_j$ has the pair-wise connected form 
  \begin{equation}
    G_j(X_1,\dots,X_N; \lambda_0,\rho_0, 0)  = \sum_{k=1}^N G_{jk} (X_j,X_k)
    \label{eq:app_thm_coupling}
  \end{equation}
  and each $G_{jk} (X_j,X_k)=0$ for $X_k$ from some open neighborhood of the saddle-node bifurcation point.
  Moreover, we assume that in the uncoupled system, $\varepsilon=0$, the phase space of each $X_j\in\mathbb R^m$, $j=1,\dots,N$, has an attractive normally hyperbolic compact invariant manifold $M_j\subset \mathbb R^m$ that is homeomorphic to $\mathbb S^1$.

  Then, there is $\varepsilon_0 >0$ such that for all $\varepsilon < \varepsilon_0$ and all $\lambda = \lambda_0 + \mathcal O(\varepsilon^2)$, there is a piece-wise continuous transformation that maps solutions of \eqref{eq:app_thm1} to those of the canonical network model
  \begin{equation}
  \begin{aligned}
  \theta'_j = (1-\cos \theta) + (1+\cos \theta_j) c_j(\theta_1,\dots,\theta_N;\varepsilon) + R_j(\theta_j,\varepsilon), \quad j=1,\dots,N,
  \label{eq:app_thm2}
  \end{aligned}
  \end{equation} 
  where $\theta_j \in \mathbb{S}^1$ are phase variables, $'=d/d\tau$ and $\tau=\varepsilon t$ is slow time. 
  The $c_j$ are connection functions and $R_j$ small remainders, see \cref{eq:app_proof_connfunc,eq:app_proof_remainder} below.

  The canonical model \eqref{eq:app_thm2} of the weakly coupled neural network \eqref{eq:app_thm1} can be approximated by
  \begin{equation}
    \theta'_j = (1-\cos \theta_j) + (1+\cos \theta_j) \Big[ \eta_j + \sum_{k=1}^N p_{jk} (\theta_k) \Big]
    \label{eq:app_thm3}
  \end{equation}
  with constants $\eta_j\in \mathbb{R}$, and the functions $p_{jk}(\theta_k)$ are given by \cref{eq:app_pulse_def} and describe smoothed $\delta$-pulses of strength $s_{jk}=\mathcal O(\|G_{jk}\|)$, i.e.~$s_{jk}$ is proportional to the amplitude of the coupling in \eqref{eq:app_thm1}; here, $\|.\|$ is the supremum norm.
  The duration of the emitted pulses is short as the $p_{jk}$ satisfy $p_{jk}(\theta_k) = 0$ if $|\theta_k - \pi| > 2\sqrt\varepsilon$ for all $j,k=1,\dots,N$, see \cref{fig:cartoon_pulse}, 
  which can, but need not, justify their approximation by discontinuos Dirac $\delta$-pulses at the time of a neuron's spike.
\end{theorem}

\begin{proof} \phantom{text}\newline
  \emph{Step 1: Invariant manifold reduction---}
  Each uncoupled subsystem $\dot X_j = F_j(X,\lambda)$ has an attractive normally hyperbolic compact invariant manifold $M_j$. Hence, the direct product $M = M_1 \times \dots \times M_N$ is also normally hyperbolic.
  Fenichel's theorem \cite[Theorem 4.1 in][]{Hoppensteadt_Izhikevich_1997} states that the network dynamics \eqref{eq:app_thm1} then persists to have an attractive normally hyperbolic compact invariant manifold $M_\varepsilon \subset \mathbb R^m \times \dots \times \mathbb R^m$ for sufficiently small coupling strengths $\varepsilon \ll 1$. 
  The persistence of normally hyperbolic manifolds further ensures that there is an open neighborhood $W$ of both $M$ and $M_\varepsilon$ as well as a continuous function $h_\varepsilon \colon W \to M$ that maps solutions of \eqref{eq:app_thm1} to those of the local model
  \begin{equation}
    \dot X_j = F_j(X_j,\lambda) + \varepsilon \tilde G_j(X_1,\dots,X_N; \lambda,\rho, \varepsilon), \quad X_j \in M_j,
    \label{eq:app_proof1}
  \end{equation}
  see \cite[Theorems 4.3 and 4.7 in][]{Hoppensteadt_Izhikevich_1997}.
  
  For the $X_j$ are $\mathcal O(\varepsilon^2)$-close to a SNIC bifurcation, the invariant manifolds $M_j$ are homeomorphic to $\mathbb{S}^1$, which allows one to rewrite \eqref{eq:app_proof1} as
  \begin{equation}
    \dot x_j = f_j(x_j,\lambda) + \varepsilon g_j(x_1,\dots,x_N; \lambda,\rho, \varepsilon), \quad x_j \in \mathbb S^1,
    \label{eq:app_proof2}
  \end{equation}
  where the $f_j$ correspond to $F_j$. 
  Due to the pair-wise connections in the $\varepsilon$-perturbed system \eqref{eq:app_proof1}, one can construct the mappings from the attractive normally hyperbolic compact invariant manifolds $M_{j,\varepsilon}$ corresponding to the $\varepsilon$-perturbations onto the (unperturbed) manifolds $M_{j}$ such that the $g_j$ share the pair-wise connected form of the $G_j$, i.e.
  \[
    g_j(x_1,\dots,x_n; \lambda_0, \rho_0, 0) = \sum_{k=1}^N g_{jk} (x_j,x_k) .
  \]
  The construction of such mappings that restrict the dynamics of \eqref{eq:app_proof1} onto $M = M_1 \times \dots \times M_N$ ensures that the $g_{jk}(x_j,x_k)=0$ for $x_k$ from a small, open neighborhood of the saddle-node bifurcation point.
  A more rigorous proof of the properties of the $g_j$ follows along the lines of the proof of, and the comment after, Theorem 4.7 as well as the lines of the theorem and the remarks in Section 5.2.1 of \cite{Hoppensteadt_Izhikevich_1997}.
  
  Without loss of generality, we now assume that the saddle-node bifurcation on the limit cycle occurs at $x_j=0$, $j=1,\dots,N$, when $\lambda = 0$.
  Thus,
  \[
    f_j(0,0) = 0, \quad \frac{\partial f_j(0,0)}{\partial x_j} = 0, \quad \frac{\partial^2 f_j(0,0)}{\partial x_j^2} > 0 \quad \text{for all } j=1,\dots, N.
  \] 
  Moreover, we introduce the parameters $\lambda_2$ and $\rho_0$ as
  \[
    \lambda = \varepsilon^2 \lambda_2 + \mathcal O(\varepsilon^3) \quad \text{and} \quad \rho =\rho(\varepsilon) = \rho_0 + \varepsilon \rho_1 + \mathcal O(\varepsilon^2).
  \]
  Finally, we assume that $g_{jk}(x_j,x_k)=0$ for $|x_k| <\sqrt{\varepsilon}/p_k$ for some constants $p_k$ to be defined below.

  \emph{Step 2: Reduction of the canoncical model---}
  \Cref{eq:app_thm2} can be obtained from \cref{eq:app_proof2} as in \cite[Theorem 8.11 in][]{Hoppensteadt_Izhikevich_1997}.
  The connection functions obey
  \begin{equation}
    c_j(\theta_1,\dots,\theta_N;\varepsilon) = \begin{cases}
      \vspace{0.35em}
      \parbox[t][0.5em][c]{1.5cm}{$\eta_j$,} &\parbox[t]{7.5cm}{if no neuron is firing a spike,\\
      i.e.\ $\forall k\in\{1,\dots,N\}: \quad |\theta_k - \pi| > 2\sqrt{\varepsilon}$,}\\
      \parbox[t][0.5em][c]{1.4cm}{$\mathcal O(1/\varepsilon),$}  &\parbox[t]{7.5cm}{if at least one neuron is firing a spike,\\
      i.e.\ $\exists k\in\{1,\dots,N\}: \quad |\theta_k - \pi| \leq 2\sqrt{\varepsilon}$.}
    \end{cases}
    \label{eq:app_proof_connfunc}
  \end{equation}
  The first case follows from the local canonical model for multiple saddle-node bifurcations in weakly connected neural networks \cite[cf.~Theorem 5.4 in][]{Hoppensteadt_Izhikevich_1997} when noticing that $\partial_{x_k} g_j(0,\dots,0;0,\rho_0,0)=0$ because $g_{j,k}(x_j,x_k)$ vanishes in an open neighborhood of $|x_k|=0$;
  the constants $\eta_j$ are defined below.
  The second case follows analogously to the proof of \cite[Theorem 8.8 in][]{Hoppensteadt_Izhikevich_1997} when replacing $\sqrt[4]{\varepsilon}$ by $\sqrt{\varepsilon}$, see also below.
  Moreover, it follows from the VCON Lemma 8.2 with $\delta=\sqrt{\varepsilon}$ in \citep{Hoppensteadt_Izhikevich_1997} that the remainders satisfy
  \begin{equation}
    R_j(\theta_j,\varepsilon) = \begin{cases}
      \vspace{0.35em}
      \parbox[t][0.5em][c]{1.75cm}{$\mathcal O(\sqrt{\varepsilon})$,} \quad &\parbox[t]{7.5cm}{if no neuron is firing a spike,\\
      i.e.\ $\forall k\in\{1,\dots,N\}: \quad |\theta_k - \pi| > 2\sqrt{\varepsilon}$,}\\
      \parbox[t][0.5em][c]{2.2cm}{$\mathcal O(\sqrt{\varepsilon} \log \varepsilon),$}  &\parbox[t]{7.5cm}{if at least one neuron is firing a spike,\\
      i.e.\ $\exists k\in\{1,\dots,N\}: \quad |\theta_k - \pi| \leq 2\sqrt{\varepsilon}$.}
    \end{cases}
    \label{eq:app_proof_remainder}
  \end{equation}

  \emph{Step 3: Approximation of the canoncical model---}
  We now follow the proof of \cite[Proposition 8.12 in][]{Hoppensteadt_Izhikevich_1997}, expand \eqref{eq:app_proof2} as a Taylor series and consider the first parts, 
  \begin{equation}
    \dot x_j = p_j x_j^2 + \varepsilon^2 p_j \eta_j + \varepsilon \sum_{k=1}^N g_{jk} (x_j,x_k) ,
    \label{eq:app_proof3}
  \end{equation}
  where 
  \[
    p_j = \frac{1}{2} \frac{\partial^2 f_j(0,0)}{\partial x_j^2}
    \quad \text{and} \quad
    \eta_j = \big[ D_\lambda f_j(0,0) \cdot \lambda_{2} + D_\rho g_j \cdot \rho_1 + \partial_\varepsilon g_j \big] / p_j;
  \]
  here, $D_\lambda f_j$ denotes the component-wise partial derivatives of $f_j$ with respect to the (possibly multidimensional) bifurcation parameter $\lambda\in\Lambda$.
  Likewise, $D_\rho g_j$ denotes the component-wise partial derivatives of $g_j$ with respect to an external input $\rho$ and $\partial_\varepsilon g_j$ is the partial derivative of $g_j$ with respect to $\varepsilon$, both terms are evaluated at
  $(x_1,\dots,x_N; \lambda, \rho,\varepsilon)=(0,\dots,0; \lambda_0, \rho_0,0)$.
  If neuron $j$ does not fire, i.e.\ $|x_j|<\sqrt{\varepsilon}/p_j$, we can use the change of variables
  \begin{equation}
    x_j = \frac{\varepsilon}{p_j} \tan(\theta_j/2),
    \label{eq:app_proof_changevariables}
  \end{equation}
  and transform \eqref{eq:app_proof3} to
  \begin{equation}
    \dot \theta_j = \varepsilon \big[ (1 - \cos \theta_j) + (1+\cos\theta_j) \eta_j\big] + (1+\cos \theta_j) p_j \sum_{k=1}^N g_{jk} (x_j,x_k) \ .
    \label{eq:app_proof4}
  \end{equation}
  If the other neurons do not fire, either, then $|x_j|<\sqrt{\varepsilon}/p_j$ for all $j=1,\dots,N$, the last sum vanishes and neuron $j$ evolves on the slow time scale $\tau = \varepsilon t$ according to
  \begin{equation}
    \theta_j'(\tau) = (1 - \cos \theta_j) + (1+\cos\theta_j) \eta_j \ .
  \end{equation}

  We now determine the effect of an incoming spike. 
  Say, neuron $k\neq j$ starts to fire a spike at time $t=0$, that is, $x_k(0) = \sqrt{\varepsilon}/p_k$.
  Neglecting the terms of order $\varepsilon$ in \eqref{eq:app_proof4} yields
  \begin{equation}
    \dot \theta_j = (1+\cos \theta_j) p_j g_{jk} \big(x_j(t),x_k(t)\big) \ ,
    \label{eq:app_proof5}
  \end{equation}
  which we must integrate from initial condition $\theta_j(0)=\theta_j^\text{old}$ until the end of the spike at time $t^\ast$ when $x_k(t^\ast) = -\sqrt{\varepsilon}/p_k$.
  In other words, the incoming pulse $p_{jk}(t)$ through the spiking of neuron $k$ is evolving on the slow time scale as
  \begin{equation}
    p_{jk}(\tau) = \frac{p_j}{\varepsilon} g_{jk} \left( \frac{\varepsilon}{p_j} \tan\left(\frac{\theta_j(\tau/\varepsilon)}{2} \right), \frac{\varepsilon}{p_k} \tan\left(\frac{\theta_k(\tau/\varepsilon)}{2} \right) \right) 
    \label{eq:app_proof6}
  \end{equation}
  for $\tau \in (0,t^\ast/\varepsilon)$;
  note that $p_{jk}(\tau) = \mathcal O(1/\varepsilon)$ as already indicated by \cref{eq:app_thm2} with the connection functions $c_j$ given by \eqref{eq:app_proof_connfunc}.
  As we assumed that neuron $j$ does not spike, we can approximate its phase at first order to be constant $\theta_j(t)\approx \theta_j^\text{old}$, so that the pulse function \eqref{eq:app_proof6} only depends on the phase of neuron $k$,
  \begin{equation}
    p_{jk}(\theta_k) := \frac{p_j}{\varepsilon} g_{jk} \left( \varepsilon \tan\left( \theta_j^\text{old}/{2} \right) / p_j ,  \varepsilon \tan\left(\theta_k/{2} \right) / p_k \right) .
    \label{eq:app_pulse_def}
  \end{equation}
  Together, this yields the slow dynamics of the weakly pulse-coupled canonical network \eqref{eq:app_thm3}.
  Moreover, as $g_{jk}(x_j,x_k)\neq 0$ only if $|x_k| > \sqrt{\varepsilon}/p_k$, one can use \cref{eq:app_proof_changevariables}
  and Taylor expand $1/\tan(x)$ about $x = \pi/2$ to find
  \[
    |x_k| > \sqrt{\varepsilon}/p_k \quad \Longleftrightarrow\quad \sqrt{\varepsilon} > \frac{1}{\left|\tan(\theta_k/2)\right|} \approx \frac{|\theta_k-\pi|}{2} \quad \text{if}\quad |\theta_k -\pi| \ll 1,
  \]
  such that $p_{jk}(\theta_k) \neq 0$ only if $|\theta_k-\pi| < 2\sqrt{\varepsilon}$ and it is zero otherwise.

  \emph{Step 4: Pulse-interpretation as smoothed $\delta$-pulses---}
  Since $p_{jk}(\theta_k)$ as defined by \cref{eq:app_pulse_def} is a smooth function of $\theta_k$ and nonzero only in a small neighborhood of $\pi$, see \cref{fig:cartoon_pulse}, one can naturally identify $p_{jk}(\theta_k)$ as a smoothed $\delta$-pulse.

  Another, mathematically elegant way to see this was proposed by Izhikevich and Hoppensteadt, who computed the phase shift corresponding to an incoming spike of neuron $k$ by solving \cref{eq:app_proof5} through separation of variables, see \cite[Proposition 8.12 in][]{Hoppensteadt_Izhikevich_1997}:
  The new phase $\theta_j^\text{new}$ after the end of the spike can be approximated as
  \begin{equation}
    \theta_j^\text{new} = 2 \arctan \left( \tan \frac{\theta_j^\text{old} }2 + s_{jk}\right) - \theta_j^\text{old} 
  \end{equation}
  with 
  \begin{equation}
    s_{jk} = \frac{1}{2} \frac{\partial^2 f_j(0,0)}{\partial x_j^2} \int_{\mathbb{S}^1} \frac{g_{jk}(x_j^\text{old},x_k)}{f_k(x_k,0)} dx_k \ ;
  \end{equation}
  the $s_{jk}$ are then proportional to $\|g_{jk}\|$ and thus correspond to the amplitude of the coupling functions $G_{jk}$.
  This leads to the approximate network model  
  \begin{equation}
    \begin{aligned}
    \theta'_j = (1-\cos \theta) + (1+\cos \theta_j) \eta_j + \sum_{k=1}^N w_{jk}(\theta_j) \delta(\theta_k - \pi)
    \end{aligned}
    \label{eq:app_proof7}
  \end{equation} 
  where $\delta$ is the Dirac $\delta$-function---justified by the short duration of an incoming spike, lasting $4\sqrt{\varepsilon}$ units of slow time---and each function $w_{jk}$ has the form
  \begin{equation}
    w_{jk} (\theta_j) = 2 \arctan \left( \tan \frac{\theta_j}2 + s_{jk}\right) - \theta_j
  \end{equation}

  For small $|s_{jk}|$, the function $w_{jk}$ can be Taylor expanded as $w_{jk}(\theta_j)= (1+\cos \theta_j)s_{jk} + \mathcal{O}(|s_{jk}|^2)$, which yields the Dirac $\delta$ pulse-coupled model
  \begin{equation}
    \begin{aligned}
    \theta'_j = (1-\cos \theta_j) + (1+\cos \theta_j) \Big[ \eta_j + \sum_{k=1}^N s_{jk} \delta(\theta_k - \pi)\Big]
    \end{aligned}
    \label{eq:app_proof8}
  \end{equation}
  as an approximation of \cref{eq:app_proof7} neglecting higher-order coupling terms.
  
  The canonical model \eqref{eq:app_proof7} with the additional remainder $\mathcal O(\sqrt{\varepsilon} \log \varepsilon)$ as in \eqref{eq:app_proof_remainder} was first presented as Theorem 1 in \citep{Izhikevich_1999}, yet without a proof.
  As the remainder smoothes the $\delta$-pulse,
  this motivates the notion that the pulses $p_{jk}(\theta_k)$ are indeed smoothed Dirac $\delta$-pulses.

\end{proof}

\begin{remark}\label{rem:app_1}
  The assumption of attractive normally hyperbolic compact invariant manifolds $M_j$ of the uncoupled subsystems $X_j$ is typically difficult to see, or to guarantee, in high-dimensional systems, but it is crucial to a rigorous derivation of the canonical weakly pulse-coupled network \eqref{eq:app_thm3} of $\theta$-neurons.
\end{remark}
\begin{remark}\label{rem:app_2}
  In order to identify the parameters of the reduced system \eqref{eq:app_thm3} from high-dimensional dynamics \eqref{eq:app_thm1} of weakly coupled Class 1 excitable systems close to SNIC bifurcations, a bottleneck is the construction of the mapping from the $\varepsilon$-perturbation $M_\varepsilon$ onto the normally hyperbolic compact invariant manifold $M$.
\end{remark}

Specifically, \cref{rem:app_1,rem:app_2} mean that even for a concrete example of weakly coupled, Class 1 excitable, conductance-based Hodgkin-Huxley-like neurons, it can be laborious---if at all possible---to reduce the system to a network of pulse-coupled $\theta$-neurons.
First, one has to prove that the dynamics of each uncoupled Hodgkin-Huxley-like neuron evolves on an attractive normally hyperbolic compact invariant manifold that is homeomorphic to $\mathbb{S}^1$.
Second, if these normally hyperbolic manifolds exists, it is another challenge to construct the mapping from the $\varepsilon$-perturbed original dynamics to the dynamics on these manifolds, which at the same time ought to respect the pair-wise coupling structure as in \eqref{eq:app_thm_coupling}.
Such a reduction goes beyond the scope of the manuscript.
The current manuscript rather analyzes general principles of how the shape of putatively reduced pulse-functions $p_{jk}(\theta_k)$ affect the synchronization properties and collective dynamics of pulse-coupled $\theta$-neurons.

\subsection{Conductance-based Wang-Buzs\'aki neuron model}
\label{appsec:WB}

The Wang-Buzs\'aki model is a conductance-based Hodgkin-Huxley-like neuron model first presented in \citep{Wang_Buzsaki_1996},
according to which an interneuron is described by a single compartment and obeys the current balance equation
\begin{subequations}
\begin{equation}
    C\frac{dv}{dt} = - g_{Na}m_\infty^3(v) h (v-E_{Na}) - g_{K} n^4(v-E_{K}) - g_{L}(v-E_{L}) + I_0 + I_\text{app},
\end{equation}
where $C=1~\mu$F/cm$^2$ is the capacitance, $g_L=0.1$mS/cm$^2$ the leak conductance, $E_L=-65$mV, and the passive time constant is $\tau = C/g_l = 10$ms.
The spike-generating Na$^+$ and K$^+$ voltage-dependent ion currents are of Hodgkin-Huxley type with activation and inactivation variables that are governed by the equations:
\begin{align}
    \dot h &= \phi \big[ \alpha_h(v) (1-h) - \beta_h(v) h \big],\\
    \dot n &= \phi \big[ \alpha_n(v) (1-n) - \beta_n(v) n \big],
\end{align}
with the following functions
\begin{align}
    m_\infty(v) &=  \alpha_m(v) / \big[ \alpha_m(v) + \beta_m(v)  \big],\\
    \alpha_m(v) &= 0.1(v+35) / \{ 1 - 0.1\exp[-(v+35)]\},\\
    \beta_m(v) &= 4\exp[-(v+60)/18],\\
    \alpha_h(v) &= 0.07\exp[-(v+58)/20],\\
    \beta_h(v) &= 1 / \{ 1 + \exp[-0.1(v+28)]\},\\
    \alpha_n(v) &= 0.01(v+34) / \{ 1 - \exp[-0.1(v+34)]\},\\
    \beta_n(v) &= 0.125\exp[-(v+44)/80].
\end{align}
\end{subequations}
The standard parameters are the following~\citep{Wang_Buzsaki_1996}:
\begin{equation}
\begin{gathered}
    C = 1.0~\mu\text{F/cm}^2, \; (g_{Na}, g_K, g_L) = (35, 9, 0.1) \text{mS/cm}^2, \\ 
    (E_{Na}, E_K, E_L) = (55, -90, -65)\text{mV}, \; \phi = 5.
\end{gathered}
\end{equation}
For a direct current $I_0 + I_\text{app} = I_\text{SNIC} \approx 0.1601~\mu$A/cm$^2$, the Wang-Buzs\'aki neuron undergoes a SNIC bifurcation and starts firing with arbitrarily low firing rate for applied currents larger than $I_\text{SNIC}$.
In line with~\citep{devalle2017}, I set $I_0 = I_\text{SNIC}$ and $I_\text{app}=0.5~\mu$A/cm$^2$ to obtain the black voltage trace with a narrow action potential in \cref{fig1}.

By updating some of the parameters, especially slowing down the spike generating currents, I obtained the red voltage trace with a broader action potential in \cref{fig1} with the updated values:
\begin{equation}
    g_L \mapsto 0.1 g_L, \; g_{Na} \mapsto 0.28 g_{Na}, \; g_K \mapsto 0.38 g_K, \; \phi \mapsto \phi/3, \; I_\text{app} \mapsto 1.35 I_\text{app}.
\end{equation}

\subsection{Derivation of \cref{eq:QIF} from a QIF model with biophysical parameters}
\label{appsec:A}
In a network of globally coupled, quadratic integrate-and-fire (QIF) neurons $j=1,\dots N$, the membrane potential $v_j$ follows the subthreshold dynamics according to~\citep{latham_et_al_2000,Izhikevich_2007}:
\begin{equation}
 C \frac{dv_j}{dt} = g_L \frac{(v-v_\text{rest}) (v-v_\text{thres})}{v_\text{thres}-v_\text{rest}} + \tilde I_j +  \tilde I_\text{syn} + g (\tilde I_\text{gap} - v_j)
 \label{eq:app_QIF1}
\end{equation}
with membrane time constant $\tau_m = C/g_L$ as the product of capacitance $C$ and membrane resistance $1/g_L$, and resting ($v_\text{rest}$) and threshold ($v_\text{thres}$) voltages.
The terms $\tilde I_j$, $\tilde I_\text{syn}$ and $\tilde I_\text{gap}$ denote individual external inputs as well as common recurrent synaptic input due to all-to-all chemical and electrical coupling, respectively.
Neuron $j$ elicits a spike once it reaches the apex potential $v_p$ and is subsequently reset to $v_r$. 
Introducing $\tau = \tau_m (v_\text{thres}-v_\text{rest})$ and $\tilde g = g (v_\text{thres}-v_\text{rest})/g_L$, one has
\begin{equation}
\begin{aligned}
 \tau \dot v_j &= v_j^2 - (v_\text{rest} +v_\text{thres} +\tilde g) v_j  + v_\text{rest}v_\text{thres} +  ( \tilde I_j +\tilde I_\text{syn} + g\tilde I_\text{syn} )  (v_\text{thres}-v_\text{rest}) /g_L\\
 &=: v_j^2 - \alpha v_j + I_j(t) 
 \end{aligned}
\end{equation}
where the input current $I_j(t)$ comprises quenched and noisy individual inputs as well as common recurrent synaptic inputs due to all-to-all connectivity:
\begin{equation}
I_j(t) = I_0 + \gamma \big[ c \; \! \eta_j + (1-c) \xi_j(t) \big] + I_\text{syn}(t)
\end{equation}
with some common input $I_0$, independent normalized Cauchy white noise $\xi_j\neq\xi_k$ of strength $(1-c)\gamma\geq0$ and Cauchy-Lorentz distributed time-independent inputs $\eta_j$ of strength $c\gamma\geq 0$, where $c\in[0,1]$.
The effect of quenched heterogeneity and Cauchy noise is identical on the collective dynamics~\citep{pcp_arxiv_2022,clusella_montbrio_2022}, that is, the parameter $c\in[0,1]$ becomes redundant on the macroscopic level. 
In other words, the parameter $\gamma >0$ captures the combined effect of both noise and heterogeneity. 
Note, however, that there are important differences on the individual neuron level for different choices of $c$~\citep{clusella_montbrio_2022}.
Furthermore, for constant $\alpha$ one shifts $v_j \mapsto v_j - \alpha/2$ and updates $I_0 \mapsto I_0 - \alpha^2/4$ to obtain the QIF dynamics \eqref{eq:QIF} of the main text,
\begin{equation}
\tau_m\dot{v}_j = v_j^2 +  I_0 + I_\text{syn}(t) + \gamma \big[ c \; \! \eta_j + (1-c)\xi_j(t)\big]
\label{eq:app_QIF}
\end{equation}
with a slight abuse of notation of the time constant $\tau_m$ for consistency.
As the quadratic term in~\eqref{eq:app_QIF} causes the voltage to diverge in finite time, one can consider apex and reset values at infinity, that is, $v_p = -v_r = \infty$. 
In this case, the QIF neuron is equivalent to the $\theta$-model~\citep{Ermentrout_1996}. 
With the transformation $v_j = \tan(\theta_j/2)$, neuron $j$ spikes when its phase $\theta_j$ crosses through $\theta=\pi$ with positive speed. The phase dynamics corresponding to \cref{eq:app_QIF} read (in Stratonovich interpretation of the white noise)
\begin{equation}
\tau_m\dot{\theta}_j = 1-\cos\theta_j + (1+\cos\theta_j) \big[I_0 +   I_\text{syn} + \gamma\big( c \; \! \eta_j + (1-c)\xi_j(t)\big)\big] \; .
\label{eq:theta}
\end{equation}
While the QIF dynamics \eqref{eq:app_QIF} represents a discontinuous dynamical system that may encompass 
real conceptual and mathematical difficulties due to the instantaneous reset whenever a neurons crosses the threshold at infinity \citep{cessac_vieville_2008,kevrekidis_siettos_kevrekidis_2017}, it can sometimes be advantageous to consider the $\theta$-model dynamics \eqref{eq:theta}, which circumvents the fire-and-reset discontinuity and can thus be treated within the theory of continuous dynamical systems. 

\subsection{Firing rate and voltage equations}
\label{appsec:B}
The goal is to provide an exact low-dimensional description of the collective dynamics of the QIF neurons \eqref{eq:QIF}. 
In the thermodynamic limit, i.e.\ for infinitely many neurons $N\to\infty$, the collective state is described by the probability density, either $\mathcal{W}(v,t)$ in voltage space, or $\mathcal{P}(\theta,t)$ in the $\theta$-phase description. 
The change from one density to the other can be achieved via $\mathcal{W}(v,t)dv = \mathcal{P}(\theta,t)d\theta$.
While the evolution of densities is governed by infinite-dimensional partial differential equations, here the aim is to distill the dynamics of two characteristic macroscopic observables of the network: the mean firing rate $R$ (the fraction of neurons spiking per infinitesimal time interval divided by that interval) and the mean voltage $V=\langle v_j \rangle$.
Conveniently, these two properties can be expressed using the Fourier expansion of $\mathcal{P}(\theta,t)=(2\pi)^{-1} \{ 1 + \sum_{n\geq 1}  Z_n(t) e^{-in\theta} + c.c. \}$. 
The $Z_n(t)$ are the Kuramoto-Daido order parameters, whose dynamics can be obtained from \eqref{eq:theta} as~\citep{pcp_arxiv_2022}
\begin{equation}
\tau_m\dot Z_n=n \left[ i\omega Z_n+h Z_{n-1}-h^*Z_{n+1}- \gamma(Z_n+\tfrac{1}{2}Z_{n-1}+\tfrac{1}{2}Z_{n+1}) \right] \; ,
\end{equation}
where $\omega = 1+I_0+ I_\text{syn}$ and $h=\frac{i}{2}(I_0+ I_\text{syn}-1)$.

The firing rate is the flux of the probability density through the threshold $\pi$ and the mean voltage is the population average of the membrane potentials $v_j=\tan(\theta_j/2)=\sin\theta_j/(1+\cos\theta_j)$.
More compactly, one can find their expressions in terms of the order parameters
\begin{subequations}
\begin{align}
R(t) &= \frac{2}{\tau_m} \mathcal{P}(\pi,t) =\frac1{\pi\tau_m} \Big\{ 1 + \sum_{n=1}^\infty (-1)^n \left[ Z_n(t) + Z_n^*(t) \right] \Big\} \; ,\label{eq:app_RZ}\\ 
V(t) &= \lim_{\epsilon\to0} \left\langle \frac{\sin\theta}{1+\cos\theta+\epsilon} \right\rangle = i \sum_{n=1}^\infty (-1)^n  \left[ Z_n(t) - Z_n^*(t) \right] \; .\label{eq:app_VZ}
\end{align}\label{eq:app_RVZ}
\end{subequations}
Capitalizing on the similarity between $R$ and $V$, Eqs.~\eqref{eq:app_RVZ} can be combined~\citep{montbrio_pazo_roxin_2015} to obtain~\citep{pcp_arxiv_2022}
\begin{equation}
 \pi\tau_m R - iV = 1 + 2\sum_{n=1}^\infty (-1)^n Z_n = \Phi + \lambda \frac{\mathcal M(-\sigma)}\sigma,
 \label{eq:app_RV_Pls}
\end{equation} 
where $\mathcal M(k)$ is a constant function that depends on the distribution $\mathcal{W}(v,0)$ of the QIF neurons' initial voltages $v_j(0)$,
and the dynamics of the three complex variables $\Phi, \lambda, \sigma\in \mathbb C$ are governed by the ODEs
\begin{align}
\tau_m\dot{\Phi} = i \Phi^2 -  i I_0 - i I_\text{syn}(t) + \gamma , \quad  
\tau_m \dot{\lambda} = 2i\Phi\lambda , \quad
\tau_m \dot{\sigma} = i\lambda \;.
\label{eq:app_Pls}
\end{align}
The firing rate and voltage (RV) dynamics of the QIF network follow from Eq.~\eqref{eq:app_RV_Pls} and are hence six-dimensional. 
For non-vanishing Cauchy white noise and/or for heterogeneous input currents following a Cauchy-Lorentz distribution of finite width, i.e.~for $\gamma >0$, we have proven that $\lambda \to 0$ for $t\to \infty$ and, hence, the collective dynamics becomes two-dimensional~\citep{pcp_arxiv_2022}. 
On the so-called Lorentzian manifold~\citep{montbrio_pazo_roxin_2015}, $\{\lambda=0\}$, which is closely related to the Ott-Antonsen manifold \citep{ott_antonsen_2008},
the collective dynamics is solely governed by $\Phi=\pi \tau_m R-iV$, yielding the RV dynamics~\citep{montbrio_pazo_roxin_2015,clusella_montbrio_2022}
\begin{subequations}
\begin{align}
\tau_m\dot{R} &= \frac{\gamma}{\pi} + 2RV \;, \\
\tau_m\dot{V} &= V^2 - (\pi \tau_m R)^2 + I_\text{syn} + I_0\;,
\end{align}
\label{eq:app_FRE}
\end{subequations}
which also describes all possible attractors of the full dynamics~(\ref{eq:app_RV_Pls}~\&~\ref{eq:app_Pls}).
Correspondingly, on the Ott-Antonsen manifold the collective dynamics is uniquely described by $Q=Z_1$, which is the Kuramoto order parameter, and all higher modes are powers of $Q$: $Z_n = Q^n$, see~\citep{luke_barreto_so_2013,So-Luke-Barreto-14,Laing_2014,laing_2015}.

As the focus of this paper lies on analyzing the steady states of the system, it suffices to consider the RV dynamics~\eqref{eq:app_FRE} on the invariant Lorentzian manifold.
Moreover, when starting on the Lorentzian manifold, i.e.~when initializing the voltages $v_j(0)$ according to a Cauchy-Lorentz distribution (which corresponds to a wrapped Cauchy distribution of initial phases $\theta_j(0)$),
then for all times $t\geq 0$ the voltage distribution $\mathcal W(v,t)$ is given by the Cauchy-Lorentz distribution~\citep{montbrio_pazo_roxin_2015,Pietras_et_al_2019,clusella_montbrio_2022}
\begin{equation}
    \mathcal W(v,t) = \frac{\tau_mR(t)}{[v-V(t)]^2 + [\pi \tau_m R(t)]^2} \; .
\end{equation}

\subsection{Chemical interaction via general pulses}
\label{appsec:C}
Chemical synaptic interactions are modeled as the postsynaptic response $s(t)$ to a presynaptic pulse $p(t)$.
In the following, I briefly present possible choices of the pulse profile $p(\theta) = p(\theta(t))$. 
As in the main text, I assume that the pulse function is smooth and has the following properties:
\begin{enumerate}
\item \label{pulse_cond1} $p$ is localized around $\theta=\pi$ (or when the membrane potential diverges, $v\to \infty$, respectively),
\item \label{pulse_cond2} $p(\theta)\geq 0$ is non-negative and vanishes at least once, 
\item \label{pulse_cond3} the total area is normalized $\int_{0}^{2\pi}p(\theta)d\theta=2\pi$. 
\end{enumerate}
The pulse function is periodic and thus admits an expansion as a Fourier series
\begin{equation}
p(\theta)=\sum_{k=-\infty}^\infty c_k e^{ik\theta},\quad c_0=1.
\label{eq:app_pulsefunc}
\end{equation}
Then, the mean presynaptic pulse activity $P = \langle p\rangle$ is represented via the order parameters as
\begin{equation}
P= \langle p\rangle=\int_0^{2\pi} p(\theta)\mathcal P(\theta,t) d\theta = 1+\sum_{k=1}^\infty (c_k^* Z_k + c_k Z_k^*) \ .
\label{eq:app_pav}
\end{equation}
As a side note, I remark that one could have defined the mean presynaptic activity \eqref{eq:app_pav} also in the QIF voltage description as $\langle p \rangle = \int_\mathbb{R} p(v)\mathcal W(v,t)dv$. 
In the phase description, however, one can capitalize on the fact that $p(\theta)$ and $\mathcal P(\theta,t)$ admit a discrete Fourier series representation, which facilitates the proposed reduction to great extent.
The voltage formulation of the QIF neurons may therefore be interpreted based on the transformation $v=\tan(\theta/2)$ to the $\theta$-model~\citep{pcp_arxiv_2022}, e.g., the pulse function in $v$ can be obtained via $p(\theta) = p(2\arctan(v(t)))$.

\paragraph{$\delta$-spikes.} When neurons only emit a unitary pulse at the instant when they spike, i.e.\ when $\theta=\pi$, then the pulse profile can be expressed as the limit of a Dirac $\delta$-function $p_{\delta,\pi}(\theta)=2\pi\delta(\theta-\pi)$. The Fourier coefficients are $c_k=(-1)^{|k|}$
and the synaptic activity reduces via Eq.~\eqref{eq:app_RZ} to the mean firing rate
\begin{equation}
P_{\delta,\pi} = \langle p_{\delta,\pi}\rangle=1+\sum_{k=1}^\infty (-1)^k(Z_k+Z_k^*)=\pi \tau_m R(t) \ .
\label{eq:app_dd}
\end{equation}

\paragraph{Ariaratnam-Strogatz (AS) pulses.} 
\cite{ariaratnam2001phase}
suggested a family
of pulse profiles $p_{\text{AS},n}(\theta)=a_n (1-\cos\theta)^n$ with $a_n=\frac{2^n(n!)^2}{(2n)!}$,
where $n\in\mathbb N$ is an integer parameter. 
The AS pulses are smooth, but in the limit $n\to\infty$ they converge to $\delta$-spikes. The average synaptic activity is a finite sum of the order parameters
\begin{equation}
P_{\text{AS},n} = \langle p_{\text{AS},n}\rangle=1+(n!)^2 \sum_{k=1}^n (-1)^k \frac{ Z_k + Z^\ast_k }{(n+k)!(n-k)!} \ .
\label{eq:app_as}
\end{equation}
\paragraph{Rectified-Poisson (RP) pulses.} 
\cite{gallego_et_al_2017} suggested the following
pulse form
\[
p_{\text{RP},r}(\theta)=  \frac{(1-r)(1-\cos\theta)}{1+2r\cos\theta+r^2} \ ,
\]
where the ``sharpness'' parameter $ r\in (-1,1)$ determines the width of the pulse:
For $r=-1$, $p_{\text{RP},r}=1$ is flat; for $r=0$, $p_{\text{RP},r}(\theta)=1-\cos\theta$ coincides with the AS pulse with $n=1$; and in the limit $r\to 1$, $p_{\text{RP},r}(\theta)$ becomes a $\delta$-spike.
The Fourier coefficients are $c_k=\frac{1+r}{2r}(-r)^{|k|}$ and the average activity is the infinite series
\begin{equation}
    P_{\text{RP},r} = \langle p_{\text{RP},r} \rangle = 1 + \frac{1+r}{2r}\sum_{k=1}^\infty (-r)^k \big[ Z_k + Z^\ast_k \big] .
    \label{eq:app_rp}
\end{equation}

\paragraph{Kato-Jones pulses.} 
The pulse shapes above are symmetric about $\theta=\pi$ with $p(0)=0$. Allowing for asymmetric pulses and loosen pulse-assumption~\ref{pulse_cond1},
one can generalize the RP pulse in the same way as \cite{kato_jones_2015} generalized the wrapped Cauchy distribution (a.k.a.~Poisson kernel). I assume that the Fourier coefficients $c_k$ of the pulse function \eqref{eq:app_pulsefunc} are
\[
c_k= a e^{i\varphi} \big(re^{i\psi}\big)^k,\quad k=1,2,\dots, \quad c_{-k} = c_k^* \; .
\]
One recovers the RP pulse for 
$\varphi = 0$ and $a = \frac{1+r}{2r}$. By imposing pulse-property~\ref{pulse_cond2} from above, I set $a=(1-r^2)/(2r)/(1-r \cos \varphi)$ and obtain the general pulse profile $p_{r,\varphi,\psi}$
with parameters $r,\varphi$ and $\psi$
\begin{equation}
p_{r,\varphi,\psi}(\theta)=1+\frac{1-r^2}{1-r\cos\varphi}\frac{\cos(\theta - \psi - \varphi)-r\cos(\varphi)}{1-2r\cos(\theta-\psi)+r^2} \ , 
\label{eq:app_kj}
\end{equation}
where the parameter $\varphi \in [-\pi,\pi)$ 
governs the asymmetry of the pulse:
for $\varphi>0$, the pulse is skewed to the right of the peak, and for $\varphi<0$ to the left.
The average synaptic activity is represented as an infinite series 
\begin{equation}
    P_{r,\varphi,\psi} = \langle p_{r,\varphi,\psi} \rangle  = 1 +\frac{1-r^2}{2r(1-r\cos\varphi)}\sum_{k=1}^\infty r^k \Big[ e^{-i(\varphi+k\psi)} Z_k + e^{i(\varphi+k\psi)} Z_k^* \Big]  .
    \label{eq:app_p_kj}
\end{equation}

\subsection{Recurrent coupling in terms of collective variables}
\label{sec:app_pulses}
As shown in Appendix~\ref{appsec:C}, it is possible to represent the relevant observables $R,V$ and the mean fields $P = \langle p \rangle$
governing the dynamics of the neural population via the order parameters $Z_k$.
Such representations can directly be used for numerical simulations of the ensemble of neurons in the thermodynamic limit, with a 
proper truncation of the infinite series.
However, it is advantageous to express the mean presynaptic pulse activity $P$ through the new complex variables $\Phi,\lambda,\sigma$, and eventually through $R$ and $V$ on the Lorentzian manifold.

For any $P$ expressed as an infinite series of the form \eqref{eq:app_p_kj}---in fact, the $\delta$-spike and the RP pulse \eqref{eq:app_rp} are special cases of the Kato-Jones pulse \eqref{eq:app_kj}---,
one can employ the same approach as in the derivation of \eqref{eq:app_RV_Pls}. 
As I will present in Appendix~\ref{sec:colvar_derivation},
one finds the full expression for the mean presynaptic activity $P_{r,\varphi,\psi}=P_{r,\varphi,\psi}(\Phi,\lambda,\sigma)$ with Kato-Jones pulses $p_{r,\varphi,\psi}$ given by \cref{eq:app_kj} in terms of the complex variables $\Phi,\lambda,\sigma$:
\begin{equation}
\begin{aligned}
P_{r,\varphi,\psi} &= \text{Re} \left\{
\frac{(1-r^2)(1+\Phi)e^{-i\varphi} + (r-\cos\varphi) u}{r(1-r\cos\varphi) u}+\right. \\
&\hspace{1.3cm} \left. \frac{2 \lambda (1-r^2)e^{-i(\varphi+\psi)}}{u \left[ \lambda(1+re^{-i\psi}) - \sigma u \right](1-r\cos\varphi)}
\mathcal M\Big(-\sigma+ \frac{\lambda}{u}(1+re^{-i\psi})\Big)
\right\}
\end{aligned}
\label{eq:appD_eq1}
\end{equation}
where I introduced the auxiliary variable $u = 1-re^{-i\psi}+ \Phi (1+re^{-i\psi})$.
For $\delta$-spikes, $(r,\varphi,\psi)\to(1,0,\pi)$, one finds that $u\to 1$ and $\langle p_{1,0,\pi} \rangle = \langle p_{\delta,\pi} \rangle = \text{Re} \big[\Phi + \lambda \mathcal M(-\sigma)/\sigma\big] = \pi\tau_m R$ as expected from \eqref{eq:app_RV_Pls}.

\subsection*{On the Lorentzian manifold}
Since the interest lies in asymptotic dynamical regimes, I focus on the dynamics on the Lorentzian manifold and take $\lambda\to0$, where $\Phi=\pi\tau_m R - iV$ holds. 
Hence, the mean presynaptic activity simplifies as
\begin{align}
P_{r,\varphi,\psi}&(R,V) = \text{Re} [\alpha/\beta] = a/b, \quad \text{where} \label{eq:app_p_kj_RV}\\
\alpha &= (1-r^2)(1+\pi\tau_m R - iV)e^{-i\varphi} \nonumber\\
&\hspace{0.4cm} + (r-\cos\varphi) \big[ 1-re^{-i\psi}+ (\pi\tau_m R - iV)(1+re^{-i\psi})\big]\nonumber\\
\beta &= r(1-r\cos\varphi)  \big[ 1-re^{-i\psi}+ (\pi\tau_m R - iV) (1+re^{-i\psi})\big] \nonumber\\
a &= (1 -2r\cos\varphi +r^2) \big[1+(\pi\tau_mR)^2+V^2\big]+(1-r^2)2\pi\tau_m R \nonumber\\
&\hspace{0.4cm} + \big[1-(\pi\tau_mR)^2 - V^2\big] \big[ \cos(\psi+\varphi) - 2r \cos(\psi) + r^2 \cos(\psi-\varphi) \big] \nonumber\\
&\hspace{0.4cm} + 2V \big[ \sin(\psi+\varphi) - 2r \sin(\psi) + r^2 \sin(\psi-\varphi) \big] \nonumber\\
b &= (1-r\cos\varphi)\left\lbrace(1+r^2)\big[1+(\pi\tau_mR)^2+V^2\big]+(1-r^2)2\pi\tau_m R\right\rbrace \nonumber\\
&\hspace{0.4cm} - 2r(1-r\cos\varphi)\left\lbrace 2V\sin\psi + \big[1-(\pi\tau_m R)^2 - V^2\big] \cos\psi\right\rbrace \nonumber
\end{align}
and is fully expressed by the firing rate $R$ and mean voltage $V$.
In sum, one obtains a closed set of firing rate equations \eqref{eq:app_FRE} for $R$ and $V$ with the synaptic input $I_\text{syn}$ as a function of $R$ and $V$ through \eqref{eq:app_p_kj_RV}, which is the formula \eqref{eq:p_mean} of the main text.
I will now present special cases of pulse profiles $p_{r,\varphi,\psi}$ such that \eqref{eq:app_p_kj_RV} simplifies further.

\subsubsection{Symmetric pulses $(\varphi=0)$}
For pulses $p_{r,\varphi,\psi}$ that are symmetric ($\varphi=0$) about the pulse peak at $\psi \in [0,2\pi)$ with width $r \in (0,1]$, one has
%
\begin{align}
P_{r,0,\psi}&(R,V) = a_s/b_s, \quad \text{where} \\
a_s &= (1+r)2\pi\tau_m R+(1-r) \left[ 1 + (\pi\tau_mR)^2 + V^2 + c_s\right]\nonumber\\
b_s &= (1-r^2)2\pi\tau_m R + (1+r^2)\big[1+(\pi\tau_mR)^2+V^2\big] -2rc_s \nonumber\\
c_s &= 2V\sin\psi + \big[1-(\pi\tau_m R)^2 - V^2\big] \cos\psi \nonumber
\end{align}

For symmetric pulses $p_{r,0,\pi}$ about the peak at $\psi=\pi$, one retrieves the RP pulse with width $r\in(0,1]$. Here, the pulse peak $\psi = 2\arctan(v_\text{thr})$ coincides with the QIF threshold $v_\text{thr}=v_p \to \infty$, and the pulse is strongest when a neuron spikes. One finds
\begin{equation}
P_{\text{RP},r}(R,V) = \frac{2\pi\tau_m R[1+r+(1-r)\pi\tau_m R] + 2(1-r)V^2}{[1+r+(1-r)\pi \tau_mR]^2 + (1-r)^2V^2} \quad \stackrel{r\to1}{\longrightarrow} \quad \pi\tau_m R.
\end{equation}

For shifted RP pulses $p_{r,0,\psi}$, $\psi\neq \pi$, one can describe the mean presynaptic activity $P_{r,0,\psi}(R,V)=a_s/b_s$ in terms of the virtual threshold $v_\text{thr} = \tan(\psi/2)$ as:
\begin{align}
a_s&= 4\pi\tau_m R (1+v_\text{thr}^2) + 2 (1-r) \big[ (\pi \tau_m R v_\text{thr})^2 +(1+V v_\text{thr})^2 -\pi \tau_m R (1+v_\text{thr}^2)\big]\nonumber\\
b_s &= 
4\big[(\pi\tau_mR)^2 + (V-v_\text{thr})^2\big] + (1-r)^2 (1+v_\text{thr}^2)\big[ (1-\pi\tau_mR)^2 + V^2\big] \nonumber\\
&\hspace{0.4cm} -4(1-r)\big[ (V-v_\text{thr})^2-\pi\tau_m R(1-\pi\tau_mR+v_\text{thr}^2) \big]\nonumber
\end{align}

\subsubsection{Dirac $\delta$-pulses with peak phase $\psi \in [0,2\pi)$}
For Dirac $\delta$-pulses in the limit $r\to1$ and $\varphi=0$, a unitary pulse is emitted at phase $\theta=\psi$, 
that is whenever a neuron's membrane potential $v$ increases through the virtual threshold $v_\text{thr} = \tan(\psi/2)$.
Then,
\begin{align}
P_{1,0,\psi}(R,V) &= \frac{2\pi\tau_m R}
{1+(\pi\tau_mR)^2+V^2 - \left\{ 2V\sin\psi + [1-(\pi\tau_m R)^2 - V^2] \cos\psi\right\}} \\
&\stackrel{\psi\to \pi}{\longrightarrow} \; \pi \tau_m R  \ ,\nonumber 
\end{align}
or alternatively, cf.~Eq.~\eqref{eq:deltav_mean},
 \begin{equation*}
P_{1,0,\psi}(R,V)  = P_{\delta,v_\text{thr}}(R,V) = \frac{\pi\tau_m R (1+v_\text{thr}^2)}{(\pi\tau_mR)^2 + (V-v_\text{thr})^2} \quad \stackrel{v_\text{thr}\to \infty}{\longrightarrow} \quad \pi \tau_m R \, .
\end{equation*}
As discussed in \cref{subsec:D1}, this expression can be interpreted as an advanced ($\psi<\pi$) or delayed ($\psi>\pi$) pulse-coupling. It is, however, not a delay in the traditional sense $P_{\delta,v_\text{thr}} \neq \pi \tau_m R(t-\tau_l)$ with fixed latency $\tau_l$, but rather the latency arises as a combination of the virtual voltage threshold $v_\text{thr}=\tan(\psi/2)$ and the membrane time constant $\tau_m$. 
Still, one may consider the resulting RV dynamics \eqref{eq:FRE} with shifted $\delta$-pulse-coupling~\eqref{eq:deltav_mean} as a low-dimensional approximation of the infinite-dimensional collective dynamics with real delays, $\tau_l\neq 0$.


\subsubsection{Asymmetric pulses with $\psi=\pi$}
For asymmetric pulses, $\varphi\neq0$, of finite width $r\in(0,1)$, the pulse has no longer its peak at $\psi$, but it will be shifted to the right ($\varphi>0$) or to the left ($\varphi<0$). 
In fact, the pulse peak is no longer at $\theta=\psi$, or equivalently at $v=v_\text{thr}$, 
but is shifted towards $v_{max} = \frac{(1+r)v_\text{thr} \cos(\varphi/2)+(1-r)\sin(\varphi/2)}{(1+r) \cos(\varphi/2)-(1-r) v_\text{thr} \sin(\varphi/2)}$.
Likewise, the pulse vanishes for $v_{min} = \frac{(1+r)v_\text{thr} \sin(\varphi/2)-(1-r)\cos(\varphi/2)}{(1+r) \sin(\varphi/2)+(1-r) v_\text{thr} \cos(\varphi/2)}$. 
The corresponding peak and trough phases $\theta_{max/min}$ can be found via $\theta_{max/min}=2\arctan(v_{max/min})$.

For the virtual threshold $v_\text{thr}$ coinciding with the QIF threshold, $v_\text{thr} = v_p\to\infty$ or $\psi=\pi$, one has $v_{max}=(1+r)/(1-r)/\tan(-\varphi/2)$ and $v_{min}=(1+r)\tan(\varphi/2)/(1-r)$, so that for $\varphi>0$ the pulse is strongest after the neuron has spiked (because $(1+r)/(1-r)>0$ for $r\in(0,1)$).
In fact, for positively slightly skewed, narrow pulses, i.e.\ for $0<\varphi\ll \pi$ and $0\ll r <1$, the pulse sharply increases shortly before a neuron spikes but reaches its peak value only after the spike, and then decays moderately as the neuron returns to its resting potential, see Fig.~\ref{fig:KJ_pulses_all}(a, red and green curves).
Such a behavior is characteristic for second-order synapses with distinct rise- and decay-time constants, see, e.g., \citep{clusella_ruffini_2022}. 
Including a peak phase $\psi$ different from $\pi$ further allows to shift the onset of the pulse and, thus, to model also a latency effect of synaptic transmission.

Setting $\psi=\pi$, one obtains on the Lorentzian manifold
\begin{align}
P_{r,\varphi,\pi}&(R,V) = a_p/ [b_p (1-r\cos \varphi)], \quad \text{where} \label{eq:KJasymm_mean}\\ 
a_p &= b_p - \left[ (1+r)^2 -(1-r)^2 (\pi^2 \tau_m^2 R^2+V^2) \right] \cos\varphi -2(1-r^2)V\sin\varphi \nonumber\\
b_p &= \left[ (1+r)+(1-r)\pi \tau_m R\right]^2 + (1-r)^2V^2 \nonumber
\end{align}
For $\psi\neq \pi$, the mean presynaptic activity $P_{r,\varphi,\psi}$ has to be determined from \eqref{eq:app_p_kj_RV}.

\subsection{Derivation of mean-field and coupling expressions}
\label{sec:colvar_derivation}

In order to derive the expression for the recurrent coupling in terms of the collective variables $\Phi,\lambda,\sigma$, see \cref{appsec:C}, I
follow the same procedure as in \citep{pcp_arxiv_2022}.
In brief, I use Eq.~(15) of \cite{cestnik_pikovsky_2022chaos},
\[ 
Z_n(t) = \sum_{m=0}^n \binom{n}{m} \beta_m(t) Q(t)^{n-m}\;,
\] 
and the complex-valued variables $Q, y,$ and $s$ that are obtained from $\Phi,\lambda,\sigma$ through the transformations
\begin{equation}
Q = \frac{1-\Phi}{1+\Phi},\quad y = \frac{2\lambda}{(1+\Phi)^2}, \quad s = \sigma - \frac{\lambda}{1+\Phi}\;.
\label{eq:app_trans}
\end{equation}
Together with the identity $\sum_{n=0}^\infty \binom{n+k-1}{k-1} x^n = (1-x)^{-k}$, one finds
\begin{align*}
\sum_{n=0}^\infty \xi^n Z_n &= \sum_{n=0}^\infty \xi^n \sum_{m=0}^n \binom{n}{m} \beta_m Q^{n-m} = \sum_{m=0}^\infty b_m \sum_{n=m}^\infty \xi^n \binom{n}{m} Q^{n-m}\\
&= \sum_{m=0}^\infty b_m \xi^m \sum_{n=0}^\infty \xi^n \binom{n+m}{m} Q^{n} 
= \sum_{m=0}^\infty b_m \xi^m \sum_{n=0}^\infty \binom{n+m}{m} \left(\xi Q \right)^{n}\\
&=  
\frac1{1-\xi Q} + \sum_{m=1}^\infty \alpha_m \left(y \xi \right)^m\big(1 - \xi Q\big)^{-(m+1)} \\
&=(1-\xi Q)^{-1}\left[1+\sum_{m=1}^\infty \alpha_m \left(\frac{y\xi }{1-\xi Q}\right)^m\right] ,
\end{align*}
where first one has to use $b_0\equiv 1$ and then substitute $\beta_m(t) = \alpha_m(t) y(t)^m$.
Now I use the definitions of generating functions $\mathcal{A}$ and $\mathcal{M}$ [Eqs.~(19) and (21) of \cite{cestnik_pikovsky_2022chaos}]:
\[
\mathcal{A}(k)=\sum_{m=1}^\infty \alpha_mk^m=\frac{k}{k-s}\mathcal{M}(k-s)
\]
This yields with $k=y\xi /(1-\xi Q$)
\begin{align*}
\sum_{n=0}^\infty \xi^n Z_n &=(1-\xi Q)^{-1}\left[1+\sum_{m=1}^\infty \alpha_mk^m\right] \\
&= (1-\xi Q)^{-1}\left[1+\frac{y\xi }{y\xi -s(1-\xi Q)}
\mathcal{M}\left(-s+\frac{y\xi }{1-\xi Q}\right)\right] .
\end{align*}
For any complex constant $C\in\mathbb C$, one then obtains (recall that $Z_0=1$ for normalization)
\begin{gather*}
\sum_{n=1}^\infty \big[C\xi^n Z_n + C^*\big(\xi^*\big)^n Z_n^*\big] =\\
2 \text{Re}\left\{\frac{C}{1-\xi Q}\left[1+\frac{y\xi }{y\xi -s(1-\xi Q)}
\mathcal{M}\left(-s+\frac{y\xi }{1-\xi Q}\right)\right]\right\} - 2 \text{Re}(C) \; .
\end{gather*}
For example, one finds an expression for the firing rate $R$ by setting $C=1$ and $\xi=-1$ and using \eqref{eq:app_RZ}, 
\begin{gather*}
\pi\tau_m R = 1 + \sum_{n=1}^\infty \left[ (-1)^n Z_n + c.c.\right] =\\
2\text{Re} \left\{(1+Q)^{-1}\left[1+\frac{y}{y+s(1+Q)}\mathcal{M}\left(-\frac{y+s(1+Q)}{1+Q}\right)\right]\right\}-1 .
\end{gather*}
Inserting transformation \eqref{eq:app_trans} yields expression \eqref{eq:app_RV_Pls}:
$\pi\tau_m R = \text{Re} \left\{ \Phi + \lambda \mathcal M(-\sigma)/\sigma\right\}$.

One can further apply the calculation presented above to calculate the mean presynaptic activity $P=\langle p \rangle$ for the pulses $p_{r,\varphi,\psi}$ given by \eqref{eq:app_kj}.
Setting $C=e^{-i\varphi}$ and $\xi=re^{-i\psi}$, one finds that \eqref{eq:app_p_kj} reduces to
\begin{equation}
\begin{gathered}
    P_{r,\varphi,\psi}  = \frac{r-\cos\varphi}{r(1-r\cos\varphi)} +\frac{1-r^2}{r(1-r\cos\varphi)} \times\\
    \times \text{Re}\left\{\frac{e^{-i\varphi}}{1-re^{-i\psi} Q}\left[1+\frac{yre^{-i\psi} }{yre^{-i\psi} -s(1-re^{-i\psi} Q)}
\mathcal{M}\left(-s+\frac{yre^{-i\psi} }{1-re^{-i\psi} Q}\right)\right]\right\}.
    \label{eq:app_p_kj_Qys}
\end{gathered}
\end{equation}
With the transformation \eqref{eq:app_trans}, Eq.~\eqref{eq:app_p_kj_Qys} can be expressed in the variables $\Phi,\lambda,\sigma$ and one finally obtains Eq.~\eqref{eq:appD_eq1},
\begin{align*}
P_{r,\varphi,\psi}(\Phi,\lambda,\sigma) &= 
\text{Re} \left\{
\frac{(1-r^2)(1+\Phi)e^{-i\alpha} + (r-\cos\alpha) u}{r(1-r\cos\alpha) u} \ +\right.\\
&\hspace{0.6cm}\left.
+\frac{2 \lambda (1-r^2)e^{-i(\alpha+\psi)}}{u \left[ \lambda(1+re^{-i\psi}) - \sigma u \right](1-r\cos\varphi)}
\mathcal M\left( -\sigma + \frac{\lambda(1+re^{-i\psi})}{u}\right)
\right\}
    \label{eq:app_p_kj_Pls}
\end{align*}
with the auxiliary variable $u = 1-re^{-i\psi}+ \Phi (1+re^{-i\psi})$.

\subsection{Finite resetting \& nonsymmetric spikes}
\label{appsec:E}
In \citep{montbrio_pazo_2020}, Montbri\'o and Paz\'o showed that the assumption of symmetric threshold and reset values $v_p = -v_r \to \infty$ can be relaxed in the QIF-RV dynamics on the Lorentzian manifold, leading to novel low-dimensional firing rate equations with arbitrary spike asymmetry.
This may become relevant in the case of the QIF model with biophysically realistic parameter values, see \cref{appsec:A}. 
For instance, \citep{latham_et_al_2000} considered the following parameters
\begin{equation}
    v_r =-65 mV, \quad v_t =-50 mV, \quad v_p = 20mV, \quad v_r = -80mV,
    \label{eq:bio_param}
\end{equation}
for which the resetting is nonsymmetric.
The spike-asymmetry can be quantified with a real, positive parameter $a>0$ as the ratio
\begin{equation}
a \equiv \frac{v_p}{-v_r}.
\end{equation}
When $a=1$, the resetting rule of the QIF model is symmetric;
this spike-symmetry is implicitly assumed in the transformation between the QIF and the $\theta$-neuron models.
For $a>1$, the spike is net-positive, whereas for $a<1$ it is net-negative, see~\citep{montbrio_pazo_2020}.
The QIF dynamics~\eqref{eq:app_QIF1} with parameters \eqref{eq:bio_param} yields a spike-asymmetry of $a \approx 3.4$.

To determine the collective dynamics with nonsymmetric spikes beyond the Lorentzian manifold and in terms of the three complex variables $\Phi,\lambda,\sigma$, one can follow~\cite{montbrio_pazo_2020} and consider the following derivation as an approximation of the actual network dynamics for finite apex and repolarization values $v_p = -a v_r < \infty$.
The firing rate $R$ will still be computed as the flux at infinity.
But for the mean voltage $V$, one has to compute
\begin{equation}
 V(t) = \lim_{v_r\to -\infty} \int_{v_r}^{v_p=-a v_r} v \mathcal W(v,t)dv = \left[ \text{p.v.} \int_{-\infty}^\infty + \lim_{v_r\to -\infty} \int_{v_r}^{-a v_r} \right] v \mathcal W(v,t)dv
\end{equation}
where I split the integral into two parts: one integral with symmetric integration limits (Cauchy principal value), and another one with the remaining integration interval~\citep{montbrio_pazo_2020}.
I have already computed the first integral 
\begin{equation}
V_s = \text{p.v.} \int_{-\infty}^\infty v\mathcal W(v,t)dv
\end{equation} 
with symmetric resetting in Eq.~\eqref{eq:app_VZ}
where I used the identities $\mathcal W(v,t)dv = \mathcal P(\theta,t)d\theta$ and $v=\tan(\theta/2)$.
To compute the second integral, one finds analogously
\begin{align*}
\lim_{v_r\to -\infty} \int_{v_r}^{-a v_r} v \mathcal W(v,t)dv &= \lim_{v_r\to -\infty} \int_{v_r}^{-a v_r} \frac{(\sin\theta) \mathcal P(\theta,t)}{1+\cos\theta} d\theta \\
&=  \frac{\log(a)}{\pi} \left[ 1 + \sum_{n=1}^\infty (-1)^n [Z_n + Z_n^*] \right] = \tau_m \log(a) R(t)\;.
\end{align*}
Hence, the mean voltage $V$ with asymmetric spikes $a>0$ becomes
\begin{equation}
V(t) = V_s(t) + \tau_m \log(a) R(t) \;.
\label{eq:app_V_to_Vs}
\end{equation}

The relation~(\ref{eq:app_RV_Pls}) between the complex variables $\Phi,\lambda,\sigma\in \mathbb C$ and the firing rate $R$ and mean voltage $V$ then becomes
\begin{equation}
\Phi + \lambda \frac{\mathcal M(-\sigma)}\sigma =  1 + 2\sum_{n=1}^\infty (-1)^n Z_n =  \pi\tau_m R - i V_s = \tau_m [\pi + i \log(a)] R - iV,
 \label{eq:app_RV_Pls_asymm}
\end{equation} 
with $\mathcal M(k)$ as defined above and
the dynamics of $\Phi, \lambda, \sigma$ continue to be governed by system~\eqref{eq:app_Pls}.
On the Lorentzian manifold, $\lambda\to0$, we then obtain the firing rate equations \citep[see also][]{montbrio_pazo_2020}
\begin{subequations}
\begin{align}
\tau_m\dot{R} &= \frac{\gamma}{\pi} + 2RV_s\;, \\
\tau_m\dot{V}_s &= V_s^2 - (\pi \tau R)^2  + I_\text{syn} + I_0\;.\label{eq:app_VFRE_asymm}
\end{align}
\label{eq:app_FRE_asymm}
\end{subequations}
The two ordinary differential equations for the mean firing rate R and the auxiliary variable 
$V_s= V - \tau_m\log(a) R $
describe the dynamics of the ensemble exactly in the limit of infinite $N$, $v_p$, and $-v_r$.
The dynamics~\eqref{eq:app_FRE_asymm} for asymmetric spikes takes on an identical form as the RV dynamics~\eqref{eq:FRE} in the main text. 
For instantaneous global coupling via $\delta$-spikes (i.e.~Dirac $\delta$-pulses emitted at the presynaptic spike times), $I_\text{syn}$ is proportional to $\tau_m R$ and the spike asymmetry only shifts the mean membrane potential as in \cref{eq:app_V_to_Vs}, but has no qualitative effect on the collective dynamics.
For electrical coupling via gap junctions of strength $g>0$, the asymmetric spike effectively introduces a (linear) virtual chemical coupling as $I_\text{syn}$ now incorporates the additional term $ g \tau_m\log(a) R$ and thus influences the onset of collective oscillations~\citep{montbrio_pazo_2020}.
For instantaneous pulse-coupling via pulses $p_{r,\varphi,\psi}$ of strength $J$, the asymmetric spike will influence the collective dynamics in a non-trivial way as the recurrent synaptic input $I_\text{syn}$ in \cref{eq:app_FRE_asymm} depends on the mean pulse activity 
\begin{equation}
    I_\text{syn} = J P_{r,\varphi,\psi} (R,V) = J P_{r,\varphi,\psi} (R,V_s + \tau_m \log(a) R) \; .
    \label{app:eq_VFRE_asymm_Isynpulse}
\end{equation}
The explicit voltage-dependence of the pulse-coupling thus introduces a virtual chemical coupling mediated by the firing rate $R$, which enters the RV dynamics \eqref{eq:app_FRE_asymm} in a highly nonlinear way and in interplay with $V_s$ via \cref{eq:p_mean}.

It is an open question how this nonlinear virtual coupling affects the collective dynamics.
Another open question concerns the simulation of QIF neurons with finite $v_p$ and $v_r$ interacting via smooth pulses $p_{r,\varphi,\psi}(2 \arctan v)$, as these pulses technically require that the voltage variable $v\in(-\infty,\infty)$ can take on arbitrary values. 
It remains to be shown how the collective dynamics of such microscopic network simulations correspond to the exact macroscopic theory, \cref{eq:app_FRE_asymm,app:eq_VFRE_asymm_Isynpulse}.

\end{document}